\documentclass{llncs}

\usepackage{xcolor,pgf,tikz,pgffor}

\usepackage{latexsym,amssymb,amsmath}
\usepackage{subfigure}
\usepackage{url,xspace}
\usepackage{stmaryrd}
\usepackage{enumerate}
\usepackage{yhmath}
\usepackage[textwidth=12.2cm,textheight=19.3cm]{geometry}
\usepackage{wrapfig}
\usepackage{bibunits}
\usetikzlibrary{arrows,shapes,snakes}
\usetikzlibrary{decorations.shapes}

\usepackage{defs}

\newtheorem{lemmastyleS}{Lemma}

\newtheorem{remarkstyleS}{Remark}

\newtheorem{propositionstyleS}{Proposition}

\makeatletter
\let\c@corollarystyleS\c@theoremstyleS
\let\c@lemmastyleS\c@theoremstyleS
\let\c@remarkstyleS\c@theoremstyleS
\let\c@propositionstyleS\c@theoremstyleS
\makeatother

\title{Emptiness and Universality Problems in\\
  Timed Automata with Positive Frequency\thanks{This work has been
    partly supported by the ESF project GASICS, the ANR project
    ANR-2010-BLAN-0317, the Tournesol Hubert Curien partnership STP,
    the ARC project AUWB-2010--10/15-UMONS-3, the FRFC project
    2.4515.11 and a grant from the National Bank of Belgium.  }}
\author{Nathalie Bertrand\inst{1}
  \and
  Patricia Bouyer\inst{2}
  \and
  Thomas Brihaye\inst{3}
\and 
Am\'elie Stainer\inst{1}}
 
\institute{INRIA Rennes, France \and LSV - CNRS \& ENS Cachan, France
  \and Universit\'e de Mons, Belgium}

\begin{document}

\pagestyle{plain}

\maketitle

\setlength{\abovedisplayskip}{4pt}
\setlength{\belowdisplayskip}{4pt}
\begin{abstract}
  The languages of infinite timed words accepted by timed automata are
traditionally defined using B\"uchi-like conditions. These acceptance
conditions focus on the set of locations visited infinitely often
along a run, but completely ignore quantitative timing aspects. In
this paper we propose a natural quantitative semantics for timed
automata based on the so-called frequency, which measures %\emph{i.e.},
the proportion of time spent in the accepting locations. %  We define
% %timed languages accepted with positive frequency,      
% frequency-languages and
We study 
various properties of timed languages accepted with positive
frequency, and 
% we study 
in particular 
% the decidability of
the emptiness and universality problems.
% for such timed languages.

% Classically, when studying infinite timed words, B\"uchi-like
% conditions are used in ordrer to define the timed languages accepted
% by timed automata. These acceptance conditions focus on the set of
% locations visisted infinitely often, but completely ignore the
% (quantitative) time aspects of the underlying executions.  In this
% paper we propose a rather natural quantitative semantics for timed
% automata based on the proportion of time spent in the accepting states
% (called the \textit{frequency}). This allows us to define
% frequency-languages \fbox{TB:frequency timed languages???} and to
% obtain (un)decidability results concerning language emptiness and
% universality depending on the number of clocks in the
% automata. \fbox{TB: Etre plus precis dans les resultats?}

\end{abstract}
\begin{bibunit}[myalpha]
\section{Introduction}\label{section:intro}

The model of timed automata, introduced by Alur and Dill in the
90's~\cite{AD-tcs94} is commonly used to represent real-time
systems. Timed automata consist of an extension of finite automata
with continuous variables, called clocks, that evolve synchronously
with time, and can be tested and reset along an execution. Despite
their uncountable state space, checking reachability, and more
generally $\omega$-regular properties, is decidable \textit{via} the
construction of a finite abstraction, the so-called region automaton.
% , which abtracts the behaviours of a timed automaton.
This fundamental result made timed automata very popular in the
formal methods community,
% to model real-time systems,
and lots of work has been done towards their verification, including
the development of dedicated tools like Kronos
%~\cite{BDMOTY98} 
or Uppaal.
%~\cite{BDL04}. \pat{pas de ref ici si manque de place amha}

More recently a huge effort has been made for modelling % more
quantitative aspects encompassing timing constraints, such as
costs~\cite{ALP-hscc01,BFHLPRV-hscc01} or
probabilities~\cite{KNSS-tcs02,BBBBG-lics08}.
% More recently, several extensions of timed automata have been
% proposed in order to encompass quantities such as
% costs~\cite{ALP-hscc01,BFHLPRV-hscc01} or
% probabilities~\cite{KNP-per09,BBBBG-lics08}\nat{adequate ref for
%   PRISM?}, with the intention to express and verify properties such
% as:
It is now possible to express and check properties such as:
 ``the minimal cost to reach a given state is smaller than $3$'',
or ``the probability to visit infinitely often a given location is
greater than $1/2$''. As a consequence, from qualitative verification,
the emphasis is now put on quantitative verification of timed
automata.

In this paper we propose a quantitative semantics for timed automata
based on the proportion of time spent in critical states (called the
\textit{frequency}). Contrary to probabilities or
volume~\cite{AD-concur09}
% ~\cite{AD-formats09,AD-concur09}\pat{en citer un
%   seul~?}\thomas{OK aussi. Le CONCUR?}
that give a value to sets of
behaviours of a timed automaton
% allow one to measure the set of behaviours 
(or a subset thereof), the frequency assigns a real value (in $[0,1]$)
to each execution of the system. It can thus be used in a
language-theoretic approach to define quantitative languages
associated with a timed automaton, or boolean languages based on
quantitative criteria \emph{e.g.}, one can consider the set of timed
words for which there is an execution of frequency greater than a
threshold $\lambda$.

Similar notions were studied in the context of untimed systems. For
finite automata, %and more generally games on finite arena,
mean-payoff conditions have been
investigated~\cite{CDH-acmtocl10,ADMW-fossacs09,CDEHR-concur10}: with each run is
associated the limit average of weights encountered along the
execution. Our notion of frequency extends mean-payoff conditions to
timed systems by assigning to an execution the limit average of time
spent in some distinguished locations. It can also be seen as a timed
version of the asymptotic frequency considered in quantitative
fairness games~\cite{BFMM-qapl10}. Concerning probabilistic models, a
similar notion was introduced in constrained probabilistic B\"uchi
automata yielding the decidability of the emptiness problem under the
probable
semantics~\cite{TBG-fsttcs09}. %\nat{citer aussi Meyer et al @concur10?}
Last, the work closest to ours deals with double-priced timed
automata~\cite{BBL-fmsd08}, where the aim is to synthesize schedulers
which optimize on-the-long-term the reward of a system. %\pat{il faut
 % en dire plus~?}\thomas{OK pour moi.}
% \fbox{dvp}

% \fbox{restriction a une horloge}
Adding other quantitative aspects to timed automata comes often with a
cost (in terms of decidability and complexity), and it is often
required to restrict the timing behaviours of the system to get some
computability results, see for instance~\cite{OW05}. The tradeoff is
then to restrict to single-clock timed automata.
% , and this is what we do in this work. \pat{boaf boaf}
%
Beyond introducing the concept of frequency, which we believe very
natural, the main contributions of this paper are the following.
First of all, using a refinement of the region graph, we show how to
compute the infimum and supremum values of frequencies in a given
single-clock timed automaton, as well as a way to decide whether these
bounds are realizable (\emph{i.e.}, whether they are minimum and
maximum respectively). The computation of these bounds together with
their realizability can be used to decide the emptiness problem for
languages defined by a threshold on the frequency.
% for single-clock timed automata.
Moreover, in the restricted case of deterministic timed automata, it
allows to decide the universality problem for these languages. Last
but not least %:-)\thomas{On laisse le smiley?}
we discuss the universality problem for frequency-languages. Even
under our restriction to one-clock timed automata, this problem is
non-primitive recursive, and we provide a decision algorithm in the
case of Zeno words when the threshold is $0$. Our restriction to
single-clock timed automata is crucial since at several points the
techniques employed do not extend to two clocks or more. In
particular, the universality problem becomes undecidable for timed
automata with several clocks.
% The rest of the paper is organized as follows...\nat{peut \^etre pas
%   n\'ecessaire, d\'ej\`a dit plus haut}
%Due to space limitation, most proofs are skipped in the core of the
%paper, and are postponed to the Appendix.

\section{Definitions and preliminaries}\label{section:defs}

In this section, we recall the model of timed automata, introduce the
concept of frequency, and show how those can be used to define timed
languages. We then compare our semantics to the standard semantics
based on B\"uchi acceptance.

% Finally, we define the corner-point abstraction which is an useful
% tool to decide several problems on languages.
\subsection{Timed automata and frequencies}

We start with notations and useful definitions concerning timed
automata~\cite{AD-tcs94}.

Given $X$ a finite set of clocks, a \emph{(clock) valuation} is a
mapping $v:X \rightarrow \mathbb{R}_+$. We write $\mathbb{R}_+^X$ for
the set of valuations. We note $\overline{0}$ the valuation that
assigns $0$ to all clocks. If $v$ is a valuation over $X$ and $t\in
\mathbb{R}_+$, then $v+t$ denotes the valuation which assigns to every
clock $x\in X$ the value $v(x)+t$.  For $X' \subseteq X$ we write
$v_{[X' \leftarrow0]}$ for the valuation equal to $v$ on $X \setminus
X'$ and to $\overline{0}$ on $X'$.

% Given $M$ a non-negative integer, an {\it $M$-bounded guard} (or
% simply guard)
A \emph{guard} over $X$ is a finite conjunction of constraints of the
form $x \sim c$ where $x\in X,\; c\in \mathbb{N}$ and $\mathord\sim
\in \{<,\le,=,\ge,>\}$. We denote by $G(X)$ the set of 
%$M$-bounded
guards over $X$. Given $g$ a guard and $v$ a valuation, we write $v
\models g$ if $v$ satisfies $g$ (defined in a natural way).
% , which is defined by $v \models (x \sim c)$ if $v(x) \sim c$

\begin{definition}
  \label{def-ta}
  A \emph{timed automaton} 
  % (TA for short)
  is a tuple $\A=(L,L_0,F,\Sigma,X,E)$ such that: $L$ is a finite set
  of locations, $L_0\subseteq L$ is the set of initial
  locations, 
  % \pat{ca vous dirait d'avoir plusieurs \'etats initiaux possibles ?
  %   Ou bien ca fait planter quelque part ?}
  $F\subseteq L$ is the set of accepting locations, $\Sigma$ is a
  finite alphabet, X is a finite set of clocks
  % $M\in \mathbb{N}$
  and $E\subseteq L\times G(X) \times \Sigma \times 2^X \times L$ is a
  finite set of edges.
\end{definition} 
% Constant $M$ is called the maximal constant of $\A$.  For example,
% the timed automaton $\A$ of the Fig.~\ref{fig:exTA} has $1$ for
% maximal constant.

% , and we will refer to $(|X|,M)$ as the \emph{resources} of $\A$.
The semantics of a timed automaton $\A$ is given as a timed transition
system $\mathcal{T}_{\A}=(S,S_0,S_F,(\mathbb{R}_+ \times
\Sigma),\rightarrow)$ with set of states $S=L \times \mathbb{R}^X_+$,
initial states $S_0=\{(\ell_0,\overline{0}) \mid \ell_0 \in L_0\}$,
final states $S_F =F \times \mathbb{R}^X_+$ and transition relation
$\mathord\rightarrow \subseteq S\times (\mathbb{R}_+ \times \Sigma)
\times S$, composed of moves of the form
% $(\ell,v)\stackrel{\tau,a}{\longrightarrow}(\ell',v')$
$(\ell,v)\xrightarrow{\tau,a}(\ell',v')$ with $\tau >0$ whenever there
exists an edge $(\ell,g,a,X',\ell') \in E$ such that $v+\tau \models
g$ and $v'=(v+\tau)_{[X'\leftarrow 0]}$.

A \emph{run} % (or \emph{execution})
$\rho$ of $\A$ is an infinite
sequence of moves starting in some $s_0 \in S_0$, \textit{i.e.},
$\rho=s_0 \xrightarrow{\tau_0,a_0} s_1\cdots
\xrightarrow{\tau_k,a_k}s_{k+1} \cdots$.
% The run is {\it
%   initial} whenever $s_0 \in S_0$.
% $\rho=s \stackrel{\tau_0,a_0}{\longrightarrow}s_1\cdots
% \stackrel{\tau_k,a_k}{\longrightarrow}s_{k+1} \cdots$.
A \emph{timed word} over $\Sigma$ is an element $(t_i,a_i)_{i \in
  \mathbb{N}}$ of $(\mathbb{R}_+ \times \Sigma)^\omega$ such that
$(t_i)_{i \in \mathbb{N}}$ is increasing.
%\pat{on ne definit pas les mots finis~?} 
The timed word is said to be \emph{Zeno} if the sequence $(t_i)_{i \in
  \mathbb{N}}$ is bounded from above. The timed word associated with
$\rho$ is $w=(t_0,a_0) \ldots (t_k,a_k) \ldots$ where $t_i =
\sum_{j=0}^i \tau_j$ for every $i$. A timed automaton $\A$ is
\emph{deterministic}
% (abbreviated DTA)
whenever, given two edges $(\ell,g_1,a,X'_1,\ell')$ and
$(\ell,g_2,a,X'_2,\ell')$ in $E$, $g_1 \wedge g_2$ cannot be
satisfied. In this case,
% In particular, 
for every timed word $w$, there is at most one run reading $w$. An
example of a (deterministic) timed automaton is given in
Fig.~\ref{fig:exTA}. As a convention locations in $F$ will be depicted
in black.
% \pat{ai mis les etats en noir car je pense que ca passe mieux a
%   l'impression, mais je n'y suis pas fortement attachee}
\begin{figure}[htbp]
\begin{center}
\begin{tikzpicture}
\everymath{\scriptstyle}
\draw(-.8,0) node (init) {};
\draw(0,0) node [circle,draw,inner sep=1.5pt] (l0) {$\ell_0$};
\draw(3,0) node [circle,draw,inner sep=1.5pt,fill=black] (l1) {$\textcolor{white}{\ell_1}$};
\draw(6,0) node [circle,draw,inner sep=1.5pt] (l2) {$\ell_2$};

\draw[-latex'] (init) -- (l0) node[pos=.5,above]{};
\draw[-latex'] (l1) .. controls +(30:1cm) and +(150:1cm) .. (l2) node[pos=.5,above]{$x<1,a,x:=0$};
\draw[-latex'] (l2) .. controls +(210:1cm) and +(330:1cm) .. (l1) node[pos=.5,below]{$x<1,a$};
\draw[-latex'] (l0) -- (l1) node[pos=.5,above]{$x=1,a,x:=0$};
\end{tikzpicture}\caption{Example of a timed automaton $\A$ with $L_0
  = \{\ell_0\}$ and $F=\{\ell_1\}$.}\label{fig:exTA}
\end{center}
\end{figure}
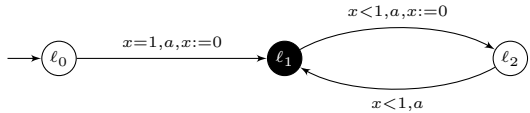

\begin{definition}
  \label{def-freq}
  Given $\A=(L,L_0,F,\Sigma,X,E)$ a timed automaton and a run $\rho =
  (\ell_0,v_0) \xrightarrow{\tau_0,a_0} (\ell_1,v_1)
  \xrightarrow{\tau_1,a_1} (\ell_2,v_2) \cdots $
  % a run in 
  of $\A$, the \emph{frequency} of $F$ along $\rho$, denoted
  $\freq{\rho}$, is defined as $\limsup_{n \rightarrow \infty}
  (\sum_{i \leq n | \ell_i \in F} \tau_i)/(\sum_{i \leq n} \tau_i)$.
%  \frac{\sum_{i \leq n | \ell_i \in F} \tau_i}{\sum_{i \leq n}
%    \tau_i}$.
% : 
%   % \pat{je noterais bien plutot $\mathsf{freq}_{\A}(\rho)$, vous en
%   %   dites quoi ?}\thomas{OK pour moi!}
%   \[
%   \freq{\rho} = \limsup_{n \rightarrow \infty} \frac{\sum_{i \leq n |
%       \ell_i \in F} \tau_i}{\sum_{i \leq n} \tau_i}.
%   \] 
\end{definition} 
Note that the choice of $\limsup$ is arbitrary, and the choice of
$\liminf$ would be as relevant. Furthermore notice that the limit may
not exist in general.
% does not necessarily exist.  below we will see a run of $\A$ where
% there is no limit.

%Run $\rho$ is said \textit{accepting with positive frequency}
%\pat{vocabulaire ok ? a uniformiser}
%  (under the quantitative semantics) 
%if $\freq{\rho}$ is positive, and 
A timed word $w$ is said \emph{accepted with positive frequency} by
$\A$ if there exists a run $\rho$ which reads $w$ and such that
$\freq{\rho}$ is positive. The \emph{positive-frequency language} of
$\A$ is the set of timed words that are accepted with positive
frequency by $\A$.  
%
%.\pat{ai enleve des defs
%  inutiles} % We write $\mathcal{L}^{>0}(\A)$
% % (resp. $\mathcal{L}^{>0}_{nZ}(\A)$, $\mathcal{L}^{>0}_{Z}(\A)$) 
% for
% % the language of $\A$, that is
% the set of timed words $w$ 
% % (resp. non-Zeno, Zeno timed words) 
% such that there exists an accepting run with positive frequency
% reading $w$, and we call such a language a \emph{positive-frequency
%   language}.
%We say that
%two timed automata $\A$ and $\B$ are \emph{equivalent} whenever
%$\L(\A) = \L(\B)$.
Note that we could define more generally languages where the frequency
of each word should be larger than some threshold $\lambda$,
% $\mathcal{L}^{\sim \lambda}(\A)$ with the obvious meaning
% (comparison with some threshold $\lambda)$,
but even though some of our results apply to this more general
framework we prefer focusing on languages with positive frequency.

\begin{example}
  % As an example, let us
  % \amelie{ai modifie cet ex pour ajouter les actions et que rien ne
  %   depace en marge... a relire}\thomas{me semble OK}
  We illustrate the notion of frequency on runs of the deterministic
  timed automaton $\A$ of Fig.~\ref{fig:exTA}. First, the only run in
  $\A$ `reading' the word $(1,a).((\frac{1}{3},a).(\frac{1}{3},a))^*$
  has frequency $\frac{1}{2}$ because
  % . Indeed, 
  the sequence $\frac{n/3}{1+(2n)/3}$ converges to $\frac{1}{2}$.
  Second, the Zeno run reading
  $(1,a).(((\frac{1}{2^k},a).(\frac{1}{2^k},a))^k)_{k\ge1}$ in $\A$ has frequency
  $\frac{1}{3}$ since the sequence
  $\frac{\sum_{k\ge1}{1}/{2^k}}{1+\sum_{k\ge1}{1}/{2^{k-1}}}$
  converges to $\frac{1}{3}$.  Finally, the run in $\A$ reading the word 
  $(1,a).(((\frac{1}{2},a).(\frac{1}{4},a))^{2^{2k}}.((\frac{1}{4},a).(\frac{1}{2},a))^{2^{2k+1}})_{k\ge1}$
  has frequency $\frac{4}{9}$. Note that the sequence under
  consideration does not converge, but its $\limsup$ is $\frac{4}{9}$.
\end{example}

\subsection{A brief comparison with usual semantics}% (subsection 2 or 3?)}
The usual semantics for timed automata considers a B\"uchi acceptance
condition. We naturally explore differences between this usual
semantics, and the one we introduced based on positive frequency. The
expressiveness of timed automata under those acceptance conditions is
not comparable, as witnessed by the automaton represented in
Fig.~\ref{fig:expr}: on the one hand, its positive-frequency language
% under the frequency-acceptance
is not timed-regular (\textit{i.e.} accepted by a timed automaton with
a standard B\"uchi acceptance condition), and on the other hand, its
B\"uchi language
% under the standard semantics 
cannot be recognized by a timed automaton with a positive-frequency
acceptance condition.
% \thomas{J'ajouterai qq details quitte a les mettre en annexe...}

\begin{figure}[ht!]
\begin{center}
\subfigure[Expressiveness.]{
\begin{tikzpicture}
\everymath{\scriptstyle}
\draw(-3.8,0) node (init) {};
\draw(-3,0) node [circle,draw,inner sep=1.5pt,fill=black] (l2) {$\textcolor{white}{\ell_0}$};
\draw(-1,0) node [circle,draw,inner sep=1.5pt] (l0) {$\ell_1$};

\draw[-latex'] (init) -- (l2) node[pos=.5,above]{};
\draw[-latex'] (l2) .. controls +(30:.8cm) and +(150:.8cm) .. (l0) node[pos=.5,above]{$x=1,a,\{x\}$};
\draw[-latex'] (l0) .. controls +(120:1cm) and +(60:1cm) .. (l0) node[pos=.5,above]{$x=1,b,\{x\}$};
\draw[-latex'] (l0) .. controls +(210:.8cm) and +(330:.8cm) .. (l2) node[pos=.5,below]{$x=1,a,\{x\}$};
\end{tikzpicture}\label{fig:expr}}
\subfigure[Universality (non-Zeno).]{
\begin{tikzpicture}
\everymath{\scriptstyle}
\draw(-3.8,0) node (init) {};
\draw(-.1,) node (init') {};
\draw(-3,0) node [circle,draw,inner sep=1.5pt] (l2) {$\ell_0$};
\draw(-1,0) node [circle,draw,inner sep=1.5pt,fill=black] (l0) {$\textcolor{white}{\ell_1}$};

\draw[-latex'] (init) -- (l2) node[pos=.5,above]{};
\draw[-latex'] (l2) .. controls +(30:.8cm) and +(150:.8cm) .. (l0) node[pos=.5,above]{$\Sigma$};
%\draw[-latex'] (l0) .. controls +(120:1cm) and +(60:1cm) .. (l0) node[pos=.5,above]{$x=1,b,\{x\}$};
\draw[-latex'] (l0) .. controls +(210:.8cm) and +(330:.8cm) .. (l2) node[pos=.5,below]{$\Sigma$};
\end{tikzpicture}\label{fig:univnZ}}
\subfigure[Universality (Zeno).]{
\begin{tikzpicture}
\everymath{\scriptstyle}
\draw(-3.8,0) node (init) {};
\draw(-4,-.6) node (init') {};
\draw(-3,0) node [circle,draw,inner sep=1.5pt,fill=black] (l2) {$\textcolor{white}{\ell_0}$};
\draw(-1,0) node [circle,draw,inner sep=1.5pt] (l0) {$\ell_1$};

\draw[-latex'] (init) -- (l2) node[pos=.5,above]{};
\draw[-latex'] (l2) -- (l0) node[pos=.5,above]{$\Sigma$};
\draw[-latex'] (l0) .. controls +(120:1cm) and +(60:1cm) .. (l0) node[pos=.5,above]{$\Sigma$};
\end{tikzpicture}\label{fig:univZ}}
\label{fig:countbuch}\caption{Automata for the comparison with the usual semantics.}
\end{center}
\end{figure}
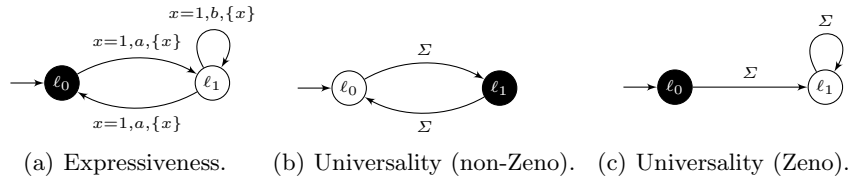

%One of the main contributions 
The contribution of this paper is to study properties of the
positive-frequency languages. We will show that we can get very fine
information on the set of frequencies of runs in \emph{single-clock}
timed automata, which implies the decidability of the emptiness
problem for positive-frequency languages.
% (note that we will also be able to decide the emptiness of any
% $\mathcal{L}^{\sim c}(\A)$ if $\A$ is a single-clock timed
% automaton).
We also show that our technics do not extend to multi-clock timed
automata.

% recognized by timed automata with a positive-frequency acceptance
% condition.
We will also 
% In particular we
consider the universality problem and variants thereof (restriction to
Zeno or non-Zeno timed words).
% variants of the universality problem,
% % for timed automata under frequency-acceptance,
% focusing on universality for
% % finite\amelie{garder mots finis ou non?},
% Zeno or non-Zeno timed words. 
On the one hand, clearly enough, a (non-Zeno)-universal timed
automaton with a positive-frequency acceptance condition is
(non-Zeno)-universal for the classical B\"uchi-acceptance. The timed
automaton of Fig.~\ref{fig:univnZ} is a counterexample to the
converse. % reverse implication.
% \thomas{Convention pour le mot vide?}
On the other hand, a Zeno-universal timed automaton under the
classical semantics is necessarily Zeno-universal under the
positive-frequency acceptance condition, but the automaton depicted in
Fig.~\ref{fig:univZ} shows that the converse
% reverse implication 
does not hold.
%

%\section{Set of frequencies}
\section{Set of frequencies of runs in one-clock timed automata}
\label{section:set}
%- + formalisation de la contraction (dilatation) soit en ajoutant une sous-section soit???

% In this section, we establish key results on the set of frequencies
% of executions in single-clock timed automata.\pat{revoir la premiere
%   phrase}
In this section, we give a precise description of the set of
frequencies of runs in single-clock timed automata.
% that we can exactly characterize the set of frequencies of a given
% one-clock timed automaton.  \thomas{Tentative de reformulation de la
%   premiere phrase???}
To this aim, we use the \emph{corner-point
  abstraction}~\cite{BBL-fmsd08}, a refinement of the region
abstraction, and exploit the links between frequencies in the timed
automaton and ratios in its corner-point abstraction. We fix a
single-clock timed automaton $\A = (L,L_0,F,\Sigma,\{x\},E)$.

%Decidability consequences for emptiness and universality problems are
%then discussed.

\subsection{The corner-point abstraction} %(subsection 2 or 3?)}
Even though the corner-point abstraction can be defined for general
timed automata~\cite{BBL-fmsd08}, we focus on the case of single-clock
timed automata.
% For simplicity, in this section we consider only single-clock timed
% automata.
%
%The \textit{maximal constant} of a timed automaton is the largest constant which appears in its guards. 

%Given $\A$ a single-clock timed automaton, t
If $M$ is the largest constant appearing in the guards of $\A$, the
usual \textit{region abstraction} of $\A$ is the partition
$\mathit{Reg}_\A$ of the set of valuations $\mathbb{R}_+$ made of the
singletons $\{i\}$ for $0\le i\le M$, the open intervals $(i,i+1)$
with $0\le i\le M-1$ and the unbounded interval $(M,\infty)$
represented by $\bot$.
% , where $M$ is the largest constant appearing in the guards of $\A$.
A piece of this partition is called a \textit{region}.  The
corner-point abstraction refines the region abstraction by associating
\textit{corner-points} with regions. The singleton regions have a
single corner-point represented by $\bullet$ whereas the open
intervals $(i,i+1)$ have two corner-points $\bullet$-- (the left
end-point of the interval) and --$\bullet$ (the right end-point of the
interval). Finally, the region $\bot$ has a single corner-point
denoted $\alpha_\bot$. We write $(R,\alpha)$ for the region $R$
pointed by the corner $\alpha$ and $(R,\alpha) + 1$ denotes its
\textit{direct time successor} defined by:
\[ (R,\alpha) +1 = \left\{ \begin{array}{ll}((i,i+1),\bullet\text{--})& \quad \text{if}\ (R,\alpha)=(\{i\},\bullet)\ \text{with}\ i<M,\\
    ((i,i+1),\text{--}\bullet)& \quad \text{if}\ (R,\alpha)=((i,i+1),\bullet\text{--}),\\
    (\{i+1\},\bullet)& \quad \text{if}\ (R,\alpha)=((i,i+1),\text{--}\bullet),\\
    (\bot,\alpha_\bot)& \quad \text{if}\ (R,\alpha)=(\{M\},\bullet)\ \text{or}\ (\bot,\alpha_\bot).
\end{array}\right.
\]
%In a region of the form $(i,i+1)$, we have naturally $\bullet$-- $+1=$ --$\bullet$. In the same way, $\bullet +1$ is the corner-point $\bullet\text{--}$ of the region $(i,i+1)$ if $\bullet$ is associated to $\{i\}$ with $i<M$.
%Moreover, the corner-point --$\bullet$ associated to $(i,i+1)$ has $\bullet$ associated to $\{i+1\}$ for direct time successor. Finally, the corner-point $\alpha_\bot$ is the direct time successor of the corner-point $\bullet$ associated to the region $\{M\}$ and of itself.
Using these notions, we define the corner-point abstraction as follows.
%of a timed automaton $\A$.
\begin{definition}
  % Given $\A$ a single-clock timed automaton, t
  The \emph{(unweighted) corner-point abstraction} of $\A$ is the
  finite automaton
  $\A_{cp}=(L_{cp},L_{0,cp},F_{cp},\Sigma_{cp},E_{cp})$ where $L_{cp}
  = L\times \mathit{Reg}_\A \times
  \{\bullet,\bullet\text{--},\text{--}\bullet, \alpha_\bot\}$ is the
  set of states, $L_{0,cp}= L_0\times \{0\} \times \{\bullet\}$ is the
  set of initial states, $F_{cp}= F\times \mathit{Reg}_\A \times
  \{\bullet,\bullet\text{--},\text{--}\bullet, \alpha_\bot\}$ is the
  set of accepting states, $\Sigma_{cp}=\Sigma \cup \{\varepsilon\}$,
  and $E_{cp} \subseteq L_{cp}\times \Sigma_{cp} \times L_{cp}$ is the
  finite set of edges defined as the union of discrete transitions and
  idling transitions:
  \begin{itemize}
  \item \emph{discrete transitions:}
    $(\ell,R,\alpha)\xrightarrow{a}(\ell',R',\alpha')$ if $\alpha$ is
    a corner-point of $R$ and there exists a transition $\ell
    \xrightarrow{g,a,X'}\ell'$ in $\A$, such that $R\subseteq g$ and
    $(R',\alpha')= (R,\alpha)$ if $X'=\emptyset$, otherwise
    $(R',\alpha')= (\{0\},\bullet)$,
  \item \emph{idling transitions:}
    $(\ell,R,\alpha)\xrightarrow{\varepsilon}(\ell,R',\alpha')$ if
    $\alpha$ (resp. $\alpha'$) is a corner-point of $R$ (resp. $R'$)
    and $(R',\alpha')=(R,\alpha) +1$.
  \end{itemize}
\end{definition}
We decorate this finite automaton with two weights for
representing frequencies, one which we call the cost, and the other
which we call the reward (by analogy with double-priced timed automata
in~\cite{BBL-fmsd08}).  The \emph{(weighted) corner-point abstraction}
$\A_{cp}^F$ is obtained from $\A_{cp}$ 
% \pat{dire qu'on parlera de l'un ou l'autre indifferemment} Given $F
% \subseteq L$ a set of locations of $\A$, we define $\A_{cp}^F$, a
% double priced transition system,
by labeling idling transitions in $\A_{cp}$ as follows: transitions
$(\ell,R,\alpha) \xrightarrow{\varepsilon} (\ell,R,\alpha')$ with
$(R,\alpha')=(R,\alpha)+1$ ($\alpha'=\alpha+1$ for short) are assigned
cost $1$ (resp. cost $0$) and reward $1$ if $\ell \in F$
(resp. $\ell \notin F$), and all other transitions are assigned both
cost and reward $0$. 
% Next we will call corner-point abstraction either $\A_{cp}$ or
% $\A_{cp}^F$.
To illustrate this definition, the corner-point abstraction of the
timed automaton in Fig.~\ref{fig:exTA} is represented in
Fig.~\ref{fig:excorner}.

\begin{figure}
\begin{center}
\begin{tikzpicture}
\everymath{\scriptstyle}
\draw(-1.1,0) node (init) {};
\draw(0,0) node [rectangle,draw,inner sep=2.5pt,rounded corners=2pt] (l00) {$\ell_0,\{0\},\;\bullet\;$};
\draw(2.5,0) node [rectangle,draw,inner sep=2.5pt,rounded corners=2pt] (l01) {$\ell_0,(0,1),\;\bullet\text{---}\;$};
\draw(5,0) node [rectangle,draw,inner sep=2.5pt,rounded corners=2pt] (l02) {$\ell_0,(0,1),\;\text{---}\bullet\;$};
\draw(7.5,0) node [rectangle,draw,inner sep=2.5pt,rounded corners=2pt] (l03) {$\ell_0,\{1\},\;\bullet\;$};
\draw(9.5,0) node [rectangle,draw,inner sep=2.5pt,rounded corners=2pt] (l04) {$\ell_0,\bot,\;\bot\;$};
\draw(0,-1.7) node [rectangle,draw,inner sep=2.5pt,fill=black,rounded corners=2pt] (l10) {$\textcolor{white}{\ell_1,\{0\},\;\bullet\;}$};
\draw(2.5,-1.7) node [rectangle,draw,inner sep=2.5pt,fill=black,rounded corners=2pt] (l11) {$\textcolor{white}{\ell_1,(0,1),\;\bullet\text{---}\;}$};
\draw(5,-1.7) node [rectangle,draw,inner sep=2.5pt,fill=black,rounded corners=2pt] (l12) {$\textcolor{white}{\ell_1,(0,1),\;\text{---}\bullet\;}$};
\draw(7.5,-1.7) node [rectangle,draw,inner sep=2.5pt,fill=black,rounded corners=2pt] (l13) {$\textcolor{white}{\ell_1,\{1\},\;\bullet\;}$};
\draw(9.5,-1.7) node [rectangle,draw,inner sep=2.5pt,fill=black,rounded corners=2pt] (l14) {$\textcolor{white}{\ell_1,\bot,\;\bot\;}$};
\draw(0,-3.4) node [rectangle,draw,inner sep=2.5pt,rounded corners=2pt] (l20) {$\ell_2,\{0\},\;\bullet\;$};
\draw(2.5,-3.4) node [rectangle,draw,inner sep=2.5pt,rounded corners=2pt] (l21) {$\ell_2,(0,1),\;\bullet\text{---}\;$};
\draw(5,-3.4) node [rectangle,draw,inner sep=2.5pt,rounded corners=2pt] (l22) {$\ell_2,(0,1),\;\text{---}\bullet\;$};
\draw(7.5,-3.4) node [rectangle,draw,inner sep=2.5pt,rounded corners=2pt] (l23) {$\ell_2,\{1\},\;\bullet\;$};
\draw(9.5,-3.4) node [rectangle,draw,inner sep=2.5pt,rounded corners=2pt] (l24) {$\ell_2,\bot,\;\bot\;$};

\draw[-latex'] (l04) .. controls +(120:1cm) and +(60:1cm) .. (l04) node[pos=.5,above]{$\varepsilon,0/1$};
\draw[-latex'] (l14) .. controls +(120:1cm) and +(60:1cm) .. (l14) node[pos=.5,above]{$\varepsilon,1/1$};
\draw[-latex'] (l24) .. controls +(120:1cm) and +(60:1cm) .. (l24) node[pos=.5,above]{$\varepsilon,0/1$};
\draw[-latex'] (l03) -- (l04) node[pos=.5,above]{$\varepsilon,0/0$};
\draw[-latex'] (l12) -- (l13) node[pos=.5,above]{$\varepsilon,0/0$};
\draw[-latex'] (l13) -- (l14) node[pos=.5,above]{$\varepsilon,0/0$};
\draw[-latex'] (l22) -- (l23) node[pos=.5,above]{$\varepsilon,0/0$};
\draw[-latex'] (l23) -- (l24) node[pos=.5,above]{$\varepsilon,0/0$};
\draw[-latex'] (init) -- (l00) node[pos=.5,above]{};
\draw[-latex'] (l00) -- (l01) node[pos=.5,above]{$\varepsilon,0/0$};
\draw[-latex'] (l01) -- (l02) node[pos=.5,above]{$\varepsilon,0/1$};
\draw[-latex'] (l02) -- (l03) node[pos=.5,above]{$\varepsilon,0/0$};
\draw[-latex'] (l10) -- (l11) node[pos=.5,above]{$\varepsilon,0/0$};
\draw[-latex'] (l11) -- (l12) node[pos=.5,above]{$\varepsilon,1/1$};
\draw[-latex'] (l20) -- (l21) node[pos=.5,above]{$\varepsilon,0/0$};
\draw[-latex'] (l21) -- (l22) node[pos=.5,above]{$\varepsilon,0/1$};
\draw[-latex'] (l21) -- (l11) node[pos=.25,right]{$a,0/0$};
\draw[-latex'] (l22) -- (l12) node[pos=.25,right]{$a,0/0$};
\draw[-latex'] (l03) .. controls +(215:2.5cm) and +(50:2cm) .. (l10) node[pos=.475,above]{$a,0/0$};
\draw[-latex'] (l11) .. controls +(220:1.5cm) and +(60:1.5cm) .. (l20) node[pos=.5,above]{$a,0/0$};
\draw[-latex'] (l12) .. controls +(215:2cm) and +(50:1.5cm) .. (l20) node[pos=.25,above]{$a,0/0$};
\end{tikzpicture}\caption{The corner-point abstraction $\A_{cp}^F$ of
  $\A$ represented Fig.~\ref{fig:exTA}.}\label{fig:excorner}
\end{center}
\end{figure}
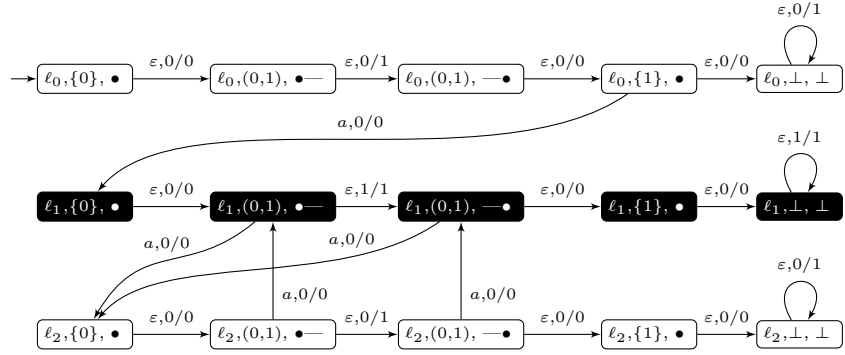
%\end{itemize}

There will be a correspondence between runs in $\A$ and runs in
$\A_{cp}$. As time is increasing in $\A$ we forbid runs in $\A_{cp}$
where two actions have to be made in $0$-delay (this is easy to do as
there should be no sequence $\dots \xrightarrow{\sigma}
(\ell,R,\alpha) \xrightarrow{\sigma'} \dots$, where both $\sigma$ and
$\sigma'$ are actions and $R$ is a punctual region).

% . A (finite or infinite) sequence of transitions
% $(\ell_0,R_0,\alpha_0) \xrightarrow{\sigma_0} (\ell_1,R_1,\alpha_1)
% \xrightarrow{\sigma_1} \dots$ (and thus a run) in $\A_{cp}$ is
% \emph{feasible} whenever $\sigma_i \in \Sigma$ and $R_{i+1}$ is
% punctual imply $\sigma_{i+1} = \varepsilon$. \pat{expliquer mieux ?}
% We will only be interested in feasible sequences of transitions,
% which is an easy property to check.

Given $\pi$ a run in $\A_{cp}^F$ the \textit{ratio} of $\pi$, denoted
$\Rat{\pi}$, is defined, provided it exists, as the $\limsup$ of the
ratio of accumulated costs divided by accumulated rewards for finite
prefixes. Run $\pi$ is said \emph{reward-converging}
(resp. \emph{reward-diverging}) if the accumulated reward along $\pi$
is bounded (resp. unbounded). Reward-converging runs in $\A_{cp}^F$
are meant to capture Zeno behaviours of $\A$.

\medskip Given $\rho$ a run in $\A$ we denote by
$\mathsf{Proj}_{cp}(\rho)$ the set of all runs in $\A^F_{cp}$
compatible with $\rho$ in the following sense.
% Formally,
We assume $\rho = (\ell_0,v_0) \xrightarrow{\tau_0,a_0} (\ell_1,v_1)
\xrightarrow{\tau_1,a_1} \cdots$, where move $(\ell_i,v_i)
\xrightarrow{\tau_i,a_i} (\ell_{i+1},v_{i+1})$ comes from an edge
$e_i$.  A run\footnote{For simplicity, we omit here the transitions
  labels} $\pi = (\ell_0,R_0^1,\alpha_0^1) \rightarrow
(\ell_0,R_0^2,\alpha_0^2) \rightarrow \cdots \rightarrow
(\ell_0,R_0^{k_0},\alpha_0^{k_0}) \rightarrow
(\ell_1,R_1^1,\alpha_1^1) \rightarrow \cdots \rightarrow
(\ell_1,R_1^{k_1},\alpha_1^{k_1}) \cdots$ of $\A^F_{cp}$ is in
$\mathsf{Proj}_{cp}(\rho)$ if for all indices $n \geq 0$:
% With $\rho = (\ell_0,v_0) \xrightarrow{\tau_0,a_0} (\ell_1,v_1)
% \xrightarrow{\tau_1,a_1} \cdots$ (assuming move $(\ell_i,v_i)
% \xrightarrow{\tau_i,a_i} (\ell_{i+1},v_{i+1})$ comes from an edge
% $e_i$) we associate the set of all runs\footnote{For simplicity, we
%   omit here the transitions labels} $\pi = (\ell_0,R_0^1,\alpha_0^1)
% \rightarrow (\ell_0,R_0^2,\alpha_0^2) \rightarrow \cdots \rightarrow
% (\ell_0,R_0^{k_0},\alpha_0^{k_0}) \rightarrow
% (\ell_1,R_1^1,\alpha_1^1) \rightarrow \cdots \rightarrow
% (\ell_1,R_1^{k_1},\alpha_1^{k_1}) \cdots $ in $\A^F_{cp}$ such that
% for all index $n \geq 0$:
\begin{itemize}
\item for all $i \leq k_n$, $\alpha_n^i$ is a corner-point of $R_n^i$,
\item for all $i \leq k_n-1$, $(R_n^{i+1},\alpha_n^{i+1}) =
  (R_n^i,\alpha_n^i) + 1$,
  % is the direct time successor of $(R_n^i,\alpha_n^i)$,
\item $(R_{n+1}^{1},\alpha_{n+1}^1)$ is the successor pointed-region
  of $(R^{k_n}_n,\alpha^{k_n}_{n})$ by transition $e_n$ (that is
  $(R_{n+1}^{1},\alpha_{n+1}^1) = (\{0\}, \bullet)$ if $e_n$ resets
  the clock $x$ and otherwise $(R_{n+1}^{1},\alpha_{n+1}^1)
  =(R^{k_n}_n,\alpha^{k_n}_{n})$),
\item $v_n \in R_n^1$ and if $R_n^{k_n} \neq \bot$, $v_n + \tau_n \in
  R_n^{k_n}$,
\item if $R_n^{k_n} = \bot$, the sum $\mu_n$ of the rewards since
  region $\{0\}$ has been visited for the last time has to be equal to
  $\lfloor v_n + \tau_n \rfloor$ or $\lceil v_n + \tau_n
  \rceil$.\footnote{Roughly, in the unbounded region $\bot$, the
    number of times an idling transition is taken should reflect how
    `big' the delay $\tau_n$ is.} Note that $\mu_n$ can be seen as the
  abstraction of the valuation $v_n$.
\end{itemize}
% Note that all runs in $\mathsf{Proj}_{cp}(\rho)$ are feasible.

\begin{remark}
  As defined above, the size of $\A_{cp}^F$ is exponential in the size
  of $\A$ because the number of regions is $2M$ (which is exponential
  in the binary encoding of $M$).  We could actually take a rougher
  version of the regions~\cite{LMS-concur04}, where only constants
  appearing in $\A$ should take part in the region partition. This
  partition, specific to single-clock timed automata is only
  polynomial in the size of $\A$.
  % This rougher version is only polynomial-size in the size of $\A$
  % and is correct in single-clock timed automata.
  We choose to simplify the presentation by considering the standard
  unit intervals.
\end{remark}

We will now see that the corner-point abstraction is a useful tool to
deduce properties of the set of frequencies of runs in the original
timed automata.

\subsection{From $\A$ to $\A_{cp}^F$, and \emph{vice-versa}}
\label{subsec:AotoAcp}
We first show that given a run $\rho$ of $\A$, there exists a run in
$\mathsf{Proj}_{cp}(\rho)$, whose ratio is smaller (resp. larger) than
the frequency of $\rho$. % Such two runs of $\A_{cp}^F$ can be
% effectively built from $\rho$, through the so-called
% \emph{contraction} (resp. \emph{dilatation}) operations. Intuitively
% it consists in minimizing (resp. maximizing) the time elapsed in
% $F$-locations. \fbox{suppression ref a l'annexe pour def formelle, ok???}
\newcounter{linkAAcp}
 \setcounter{linkAAcp}{\value{theorem}}
\begin{lemma}[From $\A$ to $\A_{cp}^F$]
  \label{prop:linkAAcp}
  For every run $\rho$ in $\A$, there exist $\pi$ and $\pi'$ in
  $\A_{cp}^F$ that can effectively be built and belong to
  $\mathsf{Proj}_{cp}(\rho)$ such that:
  % satisfy $\pi, \pi' \in \mathsf{Proj}_{cp}(\rho)$ and:
  \[
  \Rat{\pi} \leq \freq{\rho} \leq \Rat{\pi'}.
  \]
  % Moreover
  Run $\pi$ (resp. $\pi'$) minimizes (resp. maximizes) the ratio among
  runs in $\mathsf{Proj}_{cp}(\rho)$.
\end{lemma}
% \begin{lemma}[From $\A$ to $\A_{cp}^F$]
%   \label{prop:linkAAcp}
%   For every run $\rho$ in $\A$, its contraction $\pi$ and dilatation
%   $\pi'$ in $\A_{cp}^F$ can effectively be built, they belong to
%   $\mathsf{Proj}_{cp}(\rho)$ and satisfy:
%   % satisfy $\pi, \pi' \in \mathsf{Proj}_{cp}(\rho)$ and:
%   \[
%   \Rat{\pi} \leq \freq{\rho} \leq \Rat{\pi'}.
%   \]
%   % Moreover
%   Run $\pi$ (resp. $\pi'$) minimizes (resp. maximizes) the ratio among
%   runs in $\mathsf{Proj}_{cp}(\rho)$.
% %moreover, $\pi$ and $\pi'$ can be chosen in $\mathsf{Proj}_{cp}(\rho)$.
% \end{lemma}
Such two runs of $\A_{cp}^F$ can be effectively built from $\rho$,
through the so-called \emph{contraction} (resp. \emph{dilatation})
operations. Intuitively it consists in minimizing (resp. maximizing)
the time elapsed in $F$-locations.

Note that the notion of contraction cannot be adapted to the case of
timed automata with several clocks, as illustrated by the timed
automaton in Fig.~\ref{fic:cex2clockslem1}. Consider indeed the run
alternating delays $(\frac{1}{2}+\frac{1}{n})$ and
$1-(\frac{1}{2}+\frac{1}{n})$ for $n \in \mathbb{N}$, and switching
between the left-most cycle ($\ell_1-\ell_2-\ell_1$) and the
right-most cycle ($\ell_3-\ell_4-\ell_3$) following the rules: % for
% $k\in \mathbb{N}^*$,
%\pat{ai un peu modifie la phrase} 
in round $k$, take $2^{2k}$ times the cycle $\ell_1-\ell_2-\ell_1$,
then switch to $\ell_3$ and take $2^{2k+1}$ times the cycle
$\ell_3-\ell_4-\ell_3$ and return back to $\ell_1$ and continue with
round $k+1$. This run cannot have any contraction since its frequency
is $\frac{1}{2}$, whereas all its projections in the corner-point
abstraction have ratio $\frac{2}{3}$, the $\limsup$ of a
non-converging sequence. This strange behavior is due to the fact that
the delays in $\ell_1$ and $\ell_3$ need to be smaller and smaller,
and this converging phenomenon requires at least two clocks.
% Note that the notion of contraction cannot be adapted to the case of
% timed automata with several clocks, as illustrated by the timed
% automaton in Fig.~\ref{fic:cex2clockslem1}. Consider indeed the
% execution alternating delays $(\frac{1}{2}+\frac{1}{n})$ and
% $1-(\frac{1}{2}+\frac{1}{n})$ for $n \in \mathbb{N}$, which for $k\in
% \mathbb{N}^*$ takes $2^{2k}$ times the cycle $\ell_1-\ell_2-\ell_1$,
% then goes from $\ell_1$ to $\ell_3$ and takes $2^{2k+1}$ times the
% cycle $\ell_3-\ell_4-\ell_3$ and returns to $\ell_1$. This execution
% cannot have any contraction since its frequency is $\frac{1}{2}$,
% whereas all its projections in the corner point abstraction have ratio
% $\frac{2}{3}$, the $\limsup$ of a non-converging sequence. This
% strange behavior is due to the fact that the delays in $\ell_1$ and
% $\ell_3$ need to be smaller and smaller, and this converging
% phenomenom can only be obtained with two clocks or more.
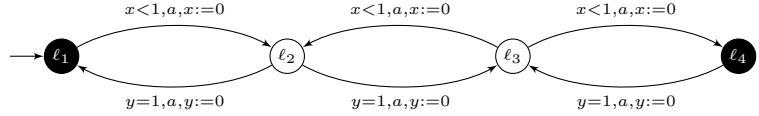
\begin{figure}[ht!]
\begin{center}
\begin{tikzpicture}
\everymath{\scriptstyle}
\draw(2.2,0) node (init) {};
\draw(3,0) node [circle,draw,inner sep=1.5pt,fill=black] (l1) {$\textcolor{white}{\ell_1}$};
\draw(6,0) node [circle,draw,inner sep=1.5pt] (l2) {$\ell_2$};
\draw(12,0) node [circle,draw,inner sep=1.5pt,fill=black] (l4) {$\textcolor{white}{\ell_4}$};
\draw(9,0) node [circle,draw,inner sep=1.5pt] (l3) {$\ell_3$};

\draw[-latex'] (init) -- (l1) node[pos=.5,above]{};
\draw[-latex'] (l1) .. controls +(30:1cm) and +(150:1cm) .. (l2) node[pos=.5,above]{$x<1,a,x:=0$};
\draw[-latex'] (l2) .. controls +(210:1cm) and +(330:1cm) .. (l1) node[pos=.5,below]{$y=1,a,y:=0$};
\draw[-latex'] (l3) .. controls +(150:1cm) and +(30:1cm) .. (l2) node[pos=.5,above]{$x<1,a,x:=0$};
\draw[-latex'] (l2) .. controls +(330:1cm) and +(210:1cm) .. (l3) node[pos=.5,below]{$y=1,a,y:=0$};
\draw[-latex'] (l3) .. controls +(30:1cm) and +(150:1cm) .. (l4) node[pos=.5,above]{$x<1,a,x:=0$};
\draw[-latex'] (l4) .. controls +(210:1cm) and +(330:1cm) .. (l3) node[pos=.5,below]{$y=1,a,y:=0$};
\end{tikzpicture}\caption{A counterexample with two clocks for
  Lemma~\ref{prop:linkAAcp}.}\label{fic:cex2clockslem1}
\end{center}
\end{figure}

We now want to know when and how runs in $\A_{cp}^F$ can be lifted to
$\A$. To that aim we distinguish between reward-diverging and
reward-converging runs.

\newcounter{linkAcpA1}
 \setcounter{linkAcpA1}{\value{theorem}}

\begin{lemma}[From $\A_{cp}^F$ to $\A$, reward-diverging case]
  \label{prop:linkAcpA1}
  For every reward-diverging 
  % feasible\pat{a mentionner dans les preuves}
  run $\pi$ in $\A_{cp}^F$, there exists a non-Zeno run $\rho$ in $\A$
  such that $\pi\in \mathsf{Proj}_{cp}(\rho)$ and $\freq{\rho} =
  \Rat{\pi}$.
\end{lemma}
\begin{proof}[Sketch]
  The key ingredient is that given a reward-diverging run $\pi$ in
  $\A_{cp}^F$, for every $\varepsilon >0$, one can build a non-Zeno
  run $\rho_\varepsilon$ of $\A$ with the following strong property:
  for all $n \in \mathbb{N}$, the valuation of the $n$-th state along
  $\rho_\epsilon$ is $\frac{\epsilon}{2^n}$-close to the abstract
  valuation in the corresponding state in $\pi$. The accumulated
  reward along $\pi$ diverges, hence $\freq{\rho_{\varepsilon}}$ is
  equal to $\Rat{\pi}$.
\end{proof}
%\thomaslong{New sketch, ai commente tout ce qui precedait}
  % The key ingredient %to prove this lemma
  % is that given a reward-diverging execution $\pi$ in $\A_{cp}^F$, for
  % every $\varepsilon >0$, one can build a non-Zeno execution
  % $\rho_\varepsilon$ of $\A$ such that, for all $n \in \mathbb{N}$,
  % $|\pi[n] - \rho_\epsilon[n]| \leq \frac{\epsilon}{2^n}$. \fbox{TB:
  %   Is this notation defined?}\amelie{je ne pense pas...c'est ce que je voulais Žviter :-)}
  %    The frequency of $\rho_\epsilon$
  % \fbox{NB: manque un indice $\varepsilon$, non?} \fbox{TB:Je crois
  %   aussi, OK maintenant?}  is thus equal to the ratio of $\pi$ since
  % the accumulated reward along $\pi$ diverges.
  % % This lemma is proved detailing why it is possible to construct
  % % $\rho$ an execution of $\A$ mimicking $\pi$ with more precision at
  % % each discrete transition. The frequency of $\rho$ is thus equal to
  % % the ration of $\pi$ thanks to the divergence of the reward of $\pi$.
  % \fbox{TB: ai change le sketch mais je suis peut etre trop
  %   technique??? NB: ok pour moi}  
  %   \fbox{AS:ok, mais les notations sont ˆ dŽfinies?}

The restriction to single-clock timed automata is crucial in
Lemma~\ref{prop:linkAcpA1}. Indeed, consider the two-clocks timed
automaton depicted in Fig.~\ref{fig:countlemnZ}.
%, where $F = \{\ell_0\}$.
In its corner-point abstraction there exists a reward-diverging run
$\pi$ with $\Rat{\pi} =0$, however every run $\rho$ satisfies
$\freq{\rho} >0$. %\thomas{TB: Add details in appendix.}
\begin{figure}
  \begin{center}
    \subfigure[A counterexample with two clocks.]{
      \begin{tikzpicture}
        \everymath{\scriptstyle}
        \draw(-3.8,0) node (init) {};
        \draw(-3,0) node [circle,draw,inner sep=1.5pt] (l2) {$\ell_0$};
        \draw(-.5,0) node [circle,draw,inner sep=1.5pt,fill=black] (l0) {$\textcolor{white}{\ell_1}$};
        \draw(1.5,0) node [circle,draw,inner sep=1.5pt] (l1) {$\ell_2$};
        \draw[-latex'] (init) -- (l2) node[pos=.5,above]{};
        \draw[-latex'] (l2) -- (l0) node[pos=.5,above]{$0<x<1,a,y:=0$};
        \draw[-latex'] (l0) .. controls +(30:1cm) and +(150:1cm)
        .. (l1) node[pos=.5,above]{$x>1,a,x:=0$}; 
        \draw[-latex'] (l1) .. controls +(210:1cm) and +(330:1cm)
        .. (l0) node[pos=.5,below]{$y=1,a,y:=0$};
      \end{tikzpicture}
      \label{fig:countlemnZ}}
    \subfigure[Zeno case.]{
        \begin{tikzpicture}
          \everymath{\scriptstyle} 
          \draw(-2.8,0) node (init) {}; 
                   \draw(-2,0) node [circle,draw,inner sep=1.5pt,fill=black] (l-1)
          {$\textcolor{white}{\ell_0}$}; 
          \draw(0,0) node [circle,draw,inner sep=1.5pt] (l0)
          {$\ell_1$}; 
\draw(2,0) node [circle,draw,inner sep=1.5pt,fill=black] (l1) {$\textcolor{white}{\ell_2}$};

\draw[-latex'] (init) -- (l-1) node[pos=.5,above]{};
\draw[-latex'] (l0) -- (l1) node[pos=.5,above]{$x=1,a,x:=0$};
\draw[-latex'] (l-1) -- (l0) node[pos=.5,above]{$x=1,a,x:=0$};
\draw[-latex'] (l1) .. controls +(120:1cm) and +(60:1cm) .. (l1) node[pos=.5,above]{$x>0,a,x:=0$};
\path (0,-.7) node {};
        \end{tikzpicture}
        \label{fig:cexlemZ}}
    \end{center}
\caption{Counterexamples to extensions of Lemma~\ref{prop:linkAcpA1}.}
  \end{figure}
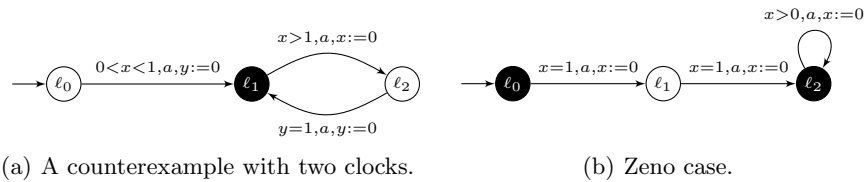

\newcounter{linkAcpA2}
 \setcounter{linkAcpA2}{\value{theorem}}

\begin{lemma}[From $\A_{cp}^F$ to $\A$, reward-converging case]
  \label{prop:linkAcpA2}
  For every reward-converging 
  % feasible
  run $\pi$ in $\A_{cp}^F$, if $\Rat{\pi}>0$, then for every
  $\varepsilon >0$, there exists a Zeno run $\rho_\varepsilon$ in $\A$
  such that $\pi \in\mathsf{Proj}_{cp}(\rho_\varepsilon)$ and
  $|\freq{\rho_\varepsilon} - \Rat{\pi}| < \varepsilon$.
\end{lemma} 
%For any Zeno execution $\pi$ in $\A_{cp}^F$, for all $\varepsilon
%>0$,
%\begin{itemize}
%\item if $\Rat{\pi}>0$, then there exists $\rho_\varepsilon$ in $\A$
%  with $|\freq{\rho_\varepsilon,F} - \Rat{\pi}| < \varepsilon$, and
%\item if $\Rat{\pi} =0$
%\begin{itemize}
%\item either for all $\rho$ with $\pi \in \mathsf{Proj}_{cp}(\rho)$,
%  $\freq{\rho,F} =0$,
%\item or for all $\rho$ with $\pi \in \mathsf{Proj}_{cp}(\rho)$,
%  $\freq{\rho,F} =1$,
%\item or for all $\rho$ with $\pi \in \mathsf{Proj}_{cp}(\rho)$,
%  $\freq{\rho,F} \in (0,1)$.
%\end{itemize}
%\end{itemize}
\begin{proof}[Sketch]
  A construction similar to the one used in the proof of
  Lemma~\ref{prop:linkAcpA1} is performed. % \thomas{Isn't it enough?}
%   \thomaslong{I would say that the following sentences are not
%     necessary if we keep the following example...} 
%     \amelielong{Comme vous voulez, c'etait pour clarifier pourquoi la preuve du cas non-Zeno ne marchait pas ici...}
  Note however that the result is slightly weaker, since in the
  reward-converging
  % Zeno
  case, one cannot neglect imprecisions (even the smallest) forced
  \emph{e.g.}, by the prohibition of the zero delays.
  % This result is proved using the same construction as the proof of
  % the Lemma~\ref{prop:linkAcpA1}.  Note that the result is less strong
  % as the Lemma~\ref{prop:linkAcpA1}, indeed the convergence of the
  % reward in $\pi$ does not allow to neglige the small imprecisions
  % forced, for example, by the interdiction of the zero delays.
\end{proof}
Note that Lemma~\ref{prop:linkAcpA2} does not hold in case $\Rat{\pi}
=0$, where we can only derive that the set of frequencies of runs
$\rho$ such that $\pi \in \mathsf{Proj}_{cp}(\rho)$ is either $\{0\}$
or $\{1\}$ or included in $(0,1)$.  Also an equivalent to
Lemma~\ref{prop:linkAcpA1} for Zeno runs (even in the single-clock
case!) is hopeless. The timed automaton $\A$ depicted in
Fig.~\ref{fig:cexlemZ}, where $F = \{\ell_0,\ell_2\}$ is a
counterexample. Indeed, in $\A_{cp}^F$ there is a reward-converging
run $\pi$ with $\Rat{\pi}=\frac{1}{2}$, whereas all Zeno runs in $\A$ have
frequency larger than $\frac{1}{2}$.%\amelie{exemple modifie, larger than = "$>$"?}
% \begin{figure}
% \begin{center}
% \begin{tikzpicture}
%  \everymath{\scriptstyle}
% \draw(-.8,0) node (init) {};
% \draw(0,0) node [circle,draw,inner sep=1.5pt] (l0) {$\ell_0$};
% \draw(3,0) node [circle,draw,inner sep=1.5pt,fill=gray] (l1) {$\ell_1$};

% \draw[-latex'] (init) -- (l0) node[pos=.5,above]{};
% \draw[-latex'] (l0) -- (l1) node[pos=.5,above]{$x=1,x:=0$};
% \draw[-latex'] (l1) .. controls +(30:25pt) and +(330:25pt) .. (l1) node[pos=.5,right]{$x>0,x:=0$};
% \end{tikzpicture}\caption{A counterexample in the Zeno case for
%   Lemma~\ref{prop:linkAcpA1}.\textcolor{red}{(in a subfigure of
%     Fig.~\ref{fig:countlemnZ}?)}}\label{fig:cexlemZ}
% \end{center}
% \end{figure}

\subsection{Set of frequencies of runs in $\A$}
\label{subsec:set}
We use the strong relation between frequencies in $\A$ and ratios in
$\A_{cp}^F$ proven in the previous subsection to establish key
properties of the set of frequencies. % of runs in $\A$.
% In this section, the synthesis of our knowledge about the set of
% frequencies in a timed automaton is made with a theorem using the
% results of the Section~\ref{subsec:AotoAcp}.
\newcounter{setfreq}
 \setcounter{setfreq}{\value{theorem}}
\begin{theorem}
  \label{th:setfreq}
  Let $\mathcal{F}_{\A} = \{\freq{\rho} \mid \rho\ \text{run of}\
  \A\}$ be the set of frequencies of runs in $\A$. We can compute
  $\inf \mathcal{F}_{\A}$ and $\sup \mathcal{F}_{\A}$. Moreover we can
  decide whether these bounds are reached or not. Everything can be
  done in NLOGSPACE.
  % polynomial time. \pat{complexite a verifier et affiner}
\end{theorem}
% \begin{theorem}
% \label{th:setfreq}
% Let $\A$ be a timed automaton and $F$ a set of locations. The infimum
% and supremum of the frequencies of $F$ for executions in $\A$ can be
% computed. Moreover, whether these bounds are reached is decidable.
% \end{theorem}
The above theorem is based on the two following lemmas dealing
respectively with the set of non-Zeno and Zeno runs in $\A$.
\newcounter{non-Zeno}
 \setcounter{non-Zeno}{\value{theorem}}
\begin{lemma}[non-Zeno case]\label{lm:non-Zeno}
  Let 
  % $\A$ be a timed automaton, $\A_{cp}$ its corner-point abstraction
  % and
  $\{C_1,\cdots,C_k\}$ be the set of reachable
  % feasible\pat{ok ?}  
  SCCs of $\A_{cp}^F$. The set of frequencies of non-Zeno runs of $\A$
  is then $\cup_{1\le i \le k} [m_i,M_i]$ where $m_i$ (resp. $M_i$) is
  the minimal (resp. maximal) ratio for a reward-diverging 
  % feasible
  cycle in $C_i$.
\end{lemma}
\begin{proof}[Sketch]
  First, the set of ratios of reward-diverging 
  % feasible 
  runs in $\A_{cp}^F$ is exactly $\cup_{1\le i \le k} [m_i,M_i]$.
  Indeed, given two extremal 
  % feasible
  cycles $c_m$ and $c_M$ of ratios $m$ and $M$ in an SCC $C$ of
  $\A_{cp}^F$, we show that every ratio $m \le r \le M$ can be
  obtained as the ratio of a run ending in $C$ by combining in a
  proper manner $c_m$ and $c_M$.  Then, using
  Lemmas~\ref{prop:linkAAcp} and \ref{prop:linkAcpA1} we derive that
  the set of frequencies of non-Zeno runs in $\A$ coincides with
  the set of ratios of reward-diverging 
  % feasible 
  runs in $\A_{cp}^F$.
\end{proof}
% \begin{lemma}
%\label{lem:sccAcp-nZ}
%Let $C_i$ be a SCC of $\A_{cp}^F$. If $\mathcal{R}_i$ denotes the set
%of ratios of reward-diverging simple cycles in $C_i$, then the set of
%ratios of reward-diverging executions of $\A_{cp}^F$ ending in $C$ is
%the interval $[m_i,M_i]$, where $m_i=\min(\mathcal{R}_i)$ and $M_i
%=\max(\mathcal{R}_i)$.
%%   $[m,M]$ where $m= \min(\mathcal{F}_C)$ and $M=
%%   \max(\mathcal{F}_C)$.
%\end{lemma}
\newcounter{Zeno}
 \setcounter{Zeno}{\value{theorem}}
\begin{lemma}[Zeno case]
  \label{lm:Zeno}
  Given $\pi$ a reward-converging 
  % feasible
  run in $\A_{cp}^F$, it is decidable whether there exists a Zeno run
  $\rho$ such that $\pi$ is the contraction of $\rho$ and
  $\freq{\rho} = \Rat{\pi}$.
\end{lemma}
\begin{proof}[Sketch]
  Observe that every fragment of $\pi$ between reset transitions can
  be considered independently, since compensations cannot occur in
  Zeno runs: even the smallest deviation (such as a delay
  $\varepsilon$ in $\A$ instead of a cost $0$ in $\pi$) will introduce
  a difference between the ratio and the frequency. A careful
  inspection of cases %(see Appendix for details)
  allows one to establish the result stated in the lemma.
 % This lemma is proved by considering
%   independently each fragments between resets.  Indeed, the Zenoness
%   of the executions ensures that no errors will be neglected
%   \fbox{TB:Not sure to understant the previous sentence}. The
%   minimization of the frequency of $\rho$ has therefore to be locally
%   optimal.  The proof is then a technical study of the possible cases.
\end{proof}

Using Lemmas~\ref{lm:non-Zeno} and \ref{lm:Zeno}, let us briefly
explain how we derive Theorem~\ref{th:setfreq}.
%(the details are in the appendix).
%\begin{proof}[Sketch of the proof of the Theorem~\ref{th:setfreq}]
For each 
% feasible
SCC $C$ of the corner-point abstraction $\A_{cp}^F$, the bounds of the
set of frequencies of runs whose contraction ends up in $C$ can be
computed thanks to the above lemmas. We can also furthermore decide
whether these bounds can be obtained by a real run in $\A$. The result
for the global automaton follows. % \pat{ai change le texte, verifier
%   que c'est ok}
% % are studied thanks to the above lemmas. 
% The global bounds are respectively the minimum and the maximum of the
% bounds of SCC of $\A_{cp}^F$. Moreover, they are realized for an
% execution in $\A$ if they are realized in at least one SCC of
% $\A_{cp}^F$.

\begin{remark}
  The link between $\A$ and $\A_{cp}^F$ differs in several aspects
  from~\cite{BBL-fmsd08}. First, a result similar to
  Lemma~\ref{prop:linkAAcp} was proven, but the runs $\pi$ and $\pi'$
  were not in $\mathsf{Proj}_{cp}(\rho)$, and more importantly it
  heavily relied on the reward-diverging hypothesis. Then the
  counter-part of Theorem~\ref{th:setfreq} was weaker
  in~\cite{BBL-fmsd08} as there was no way to decide whether the
  bounds were reachable or not. 
  % \pat{que veut-on dire de plus sur~\cite{BBL-fmsd08}~?} \thomas{Ca
  %   me semble suffisant}
\end{remark}

% \fbox{NB : comparaison avec BBL08 ici ? sinon en conclu}

%\fbox{TB: Ai legerement racourci, version orginiale commentee}
  % Assuming the above lemmas, the theorem is proved considering each
  % SCC of the corner-point abstraction $\A_{cp}^F$. For each SCC $C$,
  % the bounds of the set of frequencies of the execution whose
  % contraction ends in $C$ are studied thanks to both lemmas. The
  % global bounds are respectively the minimum and the maximum of the
  % bounds of SCC of $\A_{cp}^F$. Moreover, they are realized for an
  % execution in $\A$ if they are in at least one SCC of $\A_{cp}^F$.
%\end{proof}

%\input{sec-conseq}

\section{Emptiness and Universality Problems}
\label{section:univ}

\paragraph{The emptiness problem.}
In our context, the emptiness problem asks, given a timed automaton
$\A$ whether there is a 
% infinite
timed word which is accepted by $\A$ with positive frequency. We also
consider variants where we focus on
% its restriction to 
non-Zeno or Zeno timed words. As a consequence of
Theorem~\ref{th:setfreq}, we get the following
result. % \pat{v\'erifier la complexit\'e}
\begin{theorem}\label{th:empt}
  The emptiness  problem for infinite (resp. non-Zeno, Zeno) timed
  words 
  % with positive frequency 
  in single-clock timed automata is decidable. It is
  furthermore NLOGSPACE-Complete.
  % Given $\A$ a single-clock timed automaton, we can decide in ***
  % whether $\mathcal{L}^{>0}(\A)$ (resp. $\mathcal{L}^{>0}_{nZ}(\A)$,
  % $\mathcal{L}^{>0}_{Z}(\A)$)\pat{a verifier} is empty or not.
\end{theorem}
Note that the problem is open for timed automata with $2$ clocks or more.
\paragraph{The universality problem.}
We now focus on the universality problem, which asks, whether all timed
words are accepted with positive frequency in a given timed
automaton. We also consider variants thereof which distinguish between
Zeno and non-Zeno timed words. Note that these
% for timed automata with a positive-frequency acceptance condition,
% and also study its restriction to Zeno (resp. non-Zeno) timed
% words. Note that these
variants are incomparable: there are timed automata that, with
positive frequency,
% , under the positive-frequency acceptance condition,
recognize all Zeno timed words but not all non-Zeno timed words, and
\textit{vice-versa}.

\medskip A first obvious result concerns deterministic timed automata.
One can first check syntactically whether all infinite timed words can
be read (just locally check that the automaton is complete). Then we
notice that
% Moreover, for the restricted class of deterministic timed automata,
considering all timed words exactly amounts to considering all
runs. Thanks to Theorem~\ref{th:setfreq}, one can decide, in this
case, whether there is or not a run of frequency $0$. If not, the
automaton is universal, otherwise it is not universal.
\begin{theorem}\label{th:univdet}
  The universality problem for infinite (resp. non-Zeno, Zeno) timed
  words 
  % with positive frequency 
  in deterministic single-clock timed automata is decidable. It is
  furthermore NLOGSPACE-Complete.
%   Given a single-clock deterministic timed automaton $\A$, one can
%   decide whether $\mathcal{L}^{>0}(\A)$
%   (resp. $\mathcal{L}^{>0}_{nZ}(\A)$, $\mathcal{L}^{>0}_{Z}(\A)$) is
%   the set of all timed words.  Furthermore the problem is
%   NLOGSPACE-Complete.\nat{check also here}
  % The universality problem of frequency languages is
  % NLOGSPACE-Complete\nat{check also here} for single clock timed
  % automata.
\end{theorem}
\begin{remark}
Note that results similar to Theorems~\ref{th:empt} and~\ref{th:univdet} hold when considering languages defined with a threshold $\lambda$ on the frequency.
  %Note that a similar result holds when considering languages defined
 % with a threshold $\lambda$ on the frequency.
  % , for instance $\mathcal{L}^{>\lambda}$, $\mathcal{L}^{\ge
  %   \lambda}$, $\mathcal{L}^{<\lambda}$, and $\mathcal{L}^{\le
  %   \lambda}$. \pat{ok ?}
\end{remark}

If we relax the determinism assumption this becomes much harder!
% We now relax the determinism assumption and consider the general case.
\newcounter{ucomp}
 \setcounter{ucomp}{\value{theorem}}
\begin{theorem}\label{ucomp:infinite}
  % \nat{je couperais bien ce th en plusieurs pour mettre en valeur
  %   les r\'esultats}
  The universality problem for infinite (resp. non-Zeno, Zeno) timed
  words 
  % with positive frequency 
  in a single-clock timed automaton is non-primitive recursive. If two
  clocks are allowed, this problem is undecidable.
\end{theorem}

\begin{wrapfigure}{r}{4cm}
%\begin{figure}[ht]
  \centering
    \begin{tikzpicture}
      \everymath{\scriptstyle}
      \draw(1.3,.3) node (init) {$\textstyle{\A}$};
      \draw(-.6,0) node (init') {};
      \draw(1,0) node (init'') {};
      \draw(.8,-1) node (init''c) {};
      \draw(.75,0) node [ellipse,draw,thin,minimum height=1.2cm,minimum width=3.5cm] (A) {};
      \draw(.2,0) node [circle,draw,inner sep=3pt,double] (l1) {};
      \draw(1.4,-1) node [circle,draw,inner sep=3pt,fill=red,red] (l0c) {};
      \draw(0.2,-1) node [circle,draw,inner sep=3pt,red,fill=red] (l1c) {};
      
      \draw[-latex',dashed] (init') -- (l1) node[pos=.5,above]{};
      \draw[-latex',red] (init''c) -- (l0c) node[pos=.5,above]{};
      \draw[-latex',dashed] (l1) -- (init'') node[pos=.5,above]{};
      \draw[-latex',red] (l1) -- (l1c) node[pos=.3,left]{$c$};
      \draw[-latex',red] (l1c) .. controls +(150:1cm) and +(-150:1cm) .. (l1c) node[pos=.5,left]{$\Sigma \cup \{c\}$};
      \draw[-latex',red] (l0c) .. controls +(30:1cm) and +(-30:1cm) .. (l0c) node[pos=.5,right]{$\Sigma$};
    \end{tikzpicture}
\caption{}
%    \caption{Reduction for the proof of Theorem~\ref{ucomp:infinite}.}
\label{fig:redinfinite}
\end{wrapfigure}

\noindent \textit{Proof (Sketch).} % \begin{proof}[Sketch]
% We first sketch the hardness proof (in case of one-clock automata)
% and the undecidability proof (in the general case).
The proof is done by reduction to the universality problem for finite
words in timed automata (which is known to be undecidable for timed
automata with two clocks or more~\cite{AD-tcs94} and non-primitive
recursive for one-clock timed automata~\cite{OW05}).
Given a timed automaton $\A$ that accepts finite timed words, we
construct a timed automaton $\B$ with an extra letter $c$ which will
be interpreted with positive frequency.
% accept infinite timed words with positive frequency
% by adding an extra letter to the alphabet,
%\amelie{je ne comprends pas cette phrase :-/} 
% say $c$, 
% read this letter f
From all accepting locations of $\A$, we allow $\B$ to read $c$ and
then accept everything (with positive frequency). The construction is
illustrated on Fig.~\ref{fig:redinfinite}.
  % The idea is that if a timed automaton is universal for finite
  % words, then the way to have it universal for infinite words is to
  % add one letter (which is some kind of marker) and to read that
  % extra letter precisely when we reach a final location. Then after
  % having read the marker, everything should be accepted.
  It is easy to check that $\A$ is universal over $\Sigma$ iff $\B$ is
  universal over $\Sigma \cup \{c\}$.
  % \pat{ai commente les details}
  \hfill $\Box$

% if $w = (t_0,a_0)(t_1,a_1) \dots (t_n,a_n)$
%   is a finite timed word, then 
%   \begin{eqnarray*}
%     w\ \text{is accepted by}\ \A & \text{iff} & \forall t_n<t_{n+1}<t_{n+2} \dots\ \forall \alpha_{n+2} \alpha_{n+3} \dots \in (\Sigma \cup\{c\})^\omega,\\
%     & & w \cdot (t_{n+1},c) (t_{n+2},\alpha_{n+2}) (t_{n+3},\alpha_{n+3}) \dots\ \text{is accepted by}\ \B\ 
%   \end{eqnarray*}
%   By construction, $\B$ accepts all timed words on alphabet $\Sigma$. Thus,
%   \begin{eqnarray*}
%     \A\ \text{is universal over}\ \Sigma & \text{iff} & \B\ \text{is universal over}\ \Sigma \cup \{c\}
%   \end{eqnarray*}
% \hfill $\Box$
%\end{proof}

\begin{theorem}
  \label{theo:univ-zeno-decidable}
  The universality problem for Zeno timed words with positive
  frequency in a one-clock timed automaton is decidable.
\end{theorem}

\begin{proof}[Sketch]
  This decidability result is rather involved and requires some
  technical developments for which there is no room here.
%  All details are given in Appendix~\ref{app:univ}.
  It is based on the idea that for a Zeno timed word to be accepted
  with positive frequency it is (necessary and) sufficient to visit an
  accepting location once. Furthermore the sequence of timestamps
  associated with a Zeno timed word is converging, and we can prove
  that from some point on, in the automaton, all guards will be
  trivially either verified or denied: for instance if the value of
  the clock is $1.4$ after having read a prefix of the word, and if
  the word then converges in no more than $0.3$ time units,
  % we know that after having read a prefix of the word the value of
  % the clock is $1.4$ and that the word converges in no more than
  % $0.3$ time units,
  then only the constraint $1<x<2$ will be satisfied while reading the
  suffix of the word, unless the clock is reset, in which case only
  the constraint $0<x<1$ will be satisfied. Hence the algorithm is
  composed of two phases: first we read the prefix of the word (and we
  use a now standard abstract transition system to do so,
  see~\cite{OW05}), and then for the tail of the Zeno words, the
  behaviour of the automaton can be reduced to that of a finite
  automaton (using the above argument on tails of Zeno words).
  % It uses an abstraction into a well-structured transition which
  % records an abstraction of, inspired by~\cite{}. \pat{mettre la
  %   bonne ref}
  % 
  % \patlong{bon, je commence par finir la preuve complete dans
  %   npr.tex et je reviens en faire un resume ici ensuite}
\end{proof}

%
%\pat{je trouve que la table n'apporte pas grand chose, je la commente
%  donc}
%
% Our results are summarized in the following table:
% \begin{center}
%   \begin{tabular}{|l||c|c|c|}
%     \hline
%     Universality~& $\mathcal{L}^{>0}_{\infty}$ & $\mathcal{L}^{>0}_{nZ}$ & $\mathcal{L}^{>0}_{Z}$\\
%     \hline
%     \hline
%     ~one-clock~ & ~NPR~ & ~NPR~ & ~Decidable \& NPR~ \\
%     \hline
%     ~several clocks~ & ~undecidable~ & ~undecidable~ & ~undecidable~ \\
%     \hline
%   \end{tabular}
% \end{center}

%\input{sec-relWork}

\section{Conclusion}\label{section:conclusion}

In this paper we introduced a notion of (positive-)frequency
acceptance for timed automata and studied the related emptiness and
universality problems. This semantics is not comparable to the
classical B\"uchi semantics. For deterministic single-clock timed
automata, emptiness and universality are decidable by investigating
the set of possible frequencies based on the corner-point
abstraction. For (non-deterministic) single-clock timed automata, the
universality problem restricted to Zeno timed words is decidable but
non-primitive recursive.
% for nondeterministic single-clock timed automata, as in the
%classical setting. 
The restriction to single-clock timed automata is justified on the one
hand by the undecidability of the universality problem in the general
case. On the other hand, the techniques we employ to study the set of
possible frequencies do not extend to timed automata with several
clocks. A remaining open question is the decidability status of the
universality problem for non-Zeno timed words, which is only known to
be non-primitive recursive.  Further investigations include a deeper
study of frequencies in timed automata with multiple clocks, and also
the extension of this work to languages accepted with some frequency
larger than a given threshold. 
%\pat{relire conclusion}

\putbib[ICALP11-2]
\end{bibunit}

%\nocite{*}
% \bibliographystyle{abbrv}
% \bibliography{ICALP11-2}
\begin{bibunit}[myalpha]
\newpage

\appendix
\newgeometry{textwidth=14cm}

\section*{Technical appendix}

\noindent In this appendix, we present all proofs omitted in the core
of the paper.

\subsection*{Proofs for Section~\ref{subsec:AotoAcp}}
\paragraph{Definition of $F$-dilatation.} Let $\A$ be a timed
automaton, $F \subseteq L$ a set of locations, and $\rho = (\ell_0,0)
\xrightarrow{\tau_0,a_0} (\ell_1,v_1) \cdots$ an initial run in
$\A$. We note $e_0, e_1, \cdots$ the edges fired along $\rho$. We
define the $F$-dilatation (or simply dilatation) of $\rho$ as the
run $\pi = (\ell_0,R_0^1=\{0\},\alpha_0^1=\bullet) \rightarrow
(\ell_0,R_0^2,\alpha_0^2) \rightarrow \cdots \rightarrow
(\ell_0,R_0^{k_0},\alpha_0^{k_0}) \rightarrow
(\ell_1,R_1^1,\alpha_1^1) \rightarrow \cdots \rightarrow
(\ell_1,R_1^{k_1},\alpha_1^{k_1}) \cdots \in \mathsf{Proj}_{cp}(\rho)$
in $\A_{cp}^F$ defined inductively as follows.
%\begin{itemize}
% \item $\pi$ starts in $(\ell_0,R^1_0,\alpha^1_{0})$ with $R^1_0 = \{0\}$
% the region of $v_0$ which is the zero valuation and $\alpha^1_{0}$ the single corner-point of $\{0\}$.
%\item 
Assume $n$ transitions of $\rho$ are reflected in $\pi$: $\pi$ starts
with $(\ell_0,R_0^1,\alpha_0^1) \rightarrow \cdots \rightarrow
(\ell_0,R_0^{k_0},\alpha_0^{k_0}) \rightarrow
(\ell_1,R_1^1,\alpha_1^1) \rightarrow \cdots \rightarrow
(\ell_n,R_n^{1},\alpha_n^{1})$ with $v_n \in R^1_n$ and $\alpha^1_{n}$
corner-point of $R^1_n$.
  \begin{itemize}
  \item if $v_n + \tau_n \le M$:
\begin{itemize}
\item if $v_n + \tau_n \in R^1_n = (c,c+1)$, $\alpha^1_{n} = \bullet$--
  and $\ell_n \in F$, then we let time elapse as much as possible and
  choose in $\A_{cp}^F$ the portion of path
  $(\ell_n,R^1_n,\bullet\text{--}) \rightarrow
  (\ell_n,R^1_n,\text{--}\bullet) \rightarrow
  (\ell_{n+1},R^1_{n+1},\alpha^1_{n+1})$ where
  $(R^1_{n+1},\alpha^1_{n+1})$ is the successor pointed region of
  $(R^1_n,\text{--}\bullet)$ by transition $e_n$.
\item if $v_n + \tau_n \in R^1_n = (c,c+1)$, $\alpha^1_{n}= \bullet$--
  and $\ell_n \notin F$, we choose to fire $e_n$ as soon as possible
  by selecting the following portion of path:
  $(\ell_n,R^1_n,\bullet\text{--}) \rightarrow
  (\ell_{n+1},R^1_{n+1},\alpha^1_{n+1})$ where
  $(R^1_{n+1},\alpha^1_{n+1})$ is the successor pointed region of
  $(R^1_n,\bullet\text{--})$ by transition $e_n$.
\item if $v_n + \tau_n \in R^1_n =(c,c+1)$ and
  $\alpha^1_{n}=\text{--}\bullet$ is the last corner of $R^1_n$ (that
  is the second one), we need to immediately fire $e_n$ in $\A^F_{cp}$
  and thus choose $(\ell_n,R^1_n,\text{--}\bullet) \rightarrow
  (\ell_{n+1},R^1_{n+1},\alpha^1_{n+1})$ where
  $(R^1_{n+1},\alpha^1_{n+1})$ is the successor pointed region of
  $(R^1_n,\text{--}\bullet)$ by transition $e_n$.
\item if $v_n +\tau_n \notin R^1_n$ and $\ell_n \notin F$, we fire $e_n$ as
  soon as possible, that is, we let time elapse until region $R^{k_n}_n$ with
  $v_n + \tau_n \in R^{k_n}_n$ and its first corner-point $\alpha^{k_n}_n$, and then fire
  $e_n$: $(\ell_n,R^1_n,\alpha^1_{n}) \rightarrow \cdots \rightarrow
  (\ell_n,R^{k_n}_n,\alpha^{k_n}_n) \rightarrow (\ell_{n+1},R^1_{n+1},\alpha^1_{n+1})$
  where $(R^1_{n+1},\alpha^1_{n+1})$ is the successor pointed region of
  $(R^{k_n}_n,\alpha^{k_n}_n)$ by transition $e_n$.
\item if $v_n + \tau_n \notin R^1_n$ and $\ell_n \in F$, we fire $e_n$ as
  late as possible, that is, we let time elapse until region $R^{k_n}_n$ with
  $v_n + \tau_n \in R^{k_n}_n$ and its last corner-point $\alpha^{k_n}_n$, and then
  fire $e_n$: $(\ell_n,R^1_n,\alpha^1_{n}) \rightarrow \cdots
  \rightarrow (\ell_n,R^{k_n}_n,\alpha^{k_n}_n) \rightarrow
  (\ell_{n+1},R^1_{n+1},\alpha^1_{n+1})$ where
  $(R^1_{n+1},\alpha^1_{n+1})$ is the successor pointed region of $(R^{k_n}_n,\alpha^{k_n}_n)$
  by transition $e_n$.
  \end{itemize}
 \item if $v_n + \tau_n > M$:
 \begin{itemize}
 \item if $R^1_n \neq \bot$, we let time elapse until region $\bot$
   and add a delay to respect the definition of the projection in
   $\A_{cp}^F$ which depends on $\ell_n$ and then fire $e_n$:
   $(\ell_n,R^1_n,\alpha^1_{n}) \rightarrow \cdots \rightarrow
   (\ell_n,R^{i}_n,\alpha^{i}_n) \rightarrow (\ell_{n},\bot,\bot)\big(
   \rightarrow (\ell_{n},\bot,\bot)\big)^{\nu_n} \rightarrow
   (\ell_{n+1},R^1_{n+1},\alpha^1_{n+1})$ where $\nu_n =
   \left\{\begin{array}{ll} \lceil v_n + \tau_n \rceil - M & \mbox{ if
       }\ell_n\in F\\ \lfloor v_n + \tau_n \rfloor - M& \mbox{ if
       }\ell_n\notin F \end{array}\right.$ and
   $(R^1_{n+1},\alpha^1_{n+1})$ is the successor pointed region of
   $(\bot,\bot)$ by transition $e_n$.
 \item if $R^1_n = \bot$, respecting the definition of the projection
   give two possible delays, our choice depends on $\ell_n$, then we
   fire $e_n$: $(\ell_n,\bot,\bot) \big( \rightarrow
   (\ell_{n},\bot,\bot)\big)^{\nu_n} \rightarrow
   (\ell_{n+1},R^1_{n+1},\alpha^1_{n+1})$ where $\nu_n =
   \left\{\begin{array}{ll} \lceil v_n + \tau_n \rceil - \nu_{n-1} &
       \mbox{ if }\ell_n\in F\\ \lfloor v_n + \tau_n \rfloor - \nu_{n-1}&
       \mbox{ if }\ell_n\notin F \end{array}\right.$ and
   $(R^1_{n+1},\alpha^1_{n+1})$ is the successor pointed region of
   $(\bot,\bot)$ by transition $e_n$.

 \end{itemize}
\end{itemize}
%\fbox{\textcolor{red}{expliquer l'idee + pk on a tout considere + remplacer c par 1?}}
%\end{itemize}
Similarly, we define the $F$-contraction of $\rho$ as its
$\bar{F}$-dilatation, \emph{i.e.} the run $\pi \in
\mathsf{Proj}_{cp}(\rho)$ of $\A_{cp}^F$ which fires transition $e_n$
as soon as possible when $\ell_n \in F$ and as late as possible when
$\ell_n \notin F$.

\noindent\rule{\linewidth}{.5pt}

\setcounter{theorem}{\value{linkAAcp}}
\begin{lemma}[From $\A$ to $\A_{cp}^F$]
  \label{prop:linkAAcp}
  For every run $\rho$ in $\A$, its contraction $\pi$ and dilatation
  $\pi'$ in $\A_{cp}^F$ can effectively be built, they are in
  $\mathsf{Proj}_{cp}(\rho)$ and they satisfy:
  % satisfy $\pi, \pi' \in \mathsf{Proj}_{cp}(\rho)$ and:
  \[
  \Rat{\pi} \leq \freq{\rho} \leq \Rat{\pi'}.
  \]
  % Moreover
  Run $\pi$ (resp. $\pi'$) minimizes (resp. maximizes) the ratio among
  runs in $\mathsf{Proj}_{cp}(\rho)$.
%moreover, $\pi$ and $\pi'$ can be chosen in $\mathsf{Proj}_{cp}(\rho)$.
\end{lemma}

\begin{proof}%[of Lemma~\ref{prop:linkAAcp}]
  The proof is based on the following intuitive lemma, whose proof is
  tedious but not difficult. Given a run $\rho = (\ell_0,v_0)
  \xrightarrow{\tau_0,a_0} (\ell_1,v_1) \cdots $, in the sequel we
  abusevely denote by $\freq{\rho_n}$ the quantity given by $ (\sum_{i
    \leq n | \ell_i \in F} \tau_i)/(\sum_{i \leq n} \tau_i)$. In the
  same spirit, given $\pi$ a run in $\A_{cp}^F$, we abusively denote
  by $\Rat{\pi_n}$, the ratio of accumulated costs divided by
  accumulated rewards for the finite prefix of length $n$.%\thomas{OK???}

\begin{lemmastyleS}
\label{lem-cases}
% Let $\A$ be a timed automaton, 
Let $\rho$ be a run of $\A$, and $\pi_n$ be the dilatation of $\rho_n$
(for $n \in \IN$). For all $n \in \IN$, if $\freq{\rho_n} =
\frac{c_n}{r_n}$ and $\Rat{\pi_n} = \frac{C_n}{R_n}$ then $C_n \ge
c_n$ and $(R_n-C_n) \le (r_{n} - c_{n})$.
\end{lemmastyleS}

Assuming the latter lemma, it is easy to conclude that $\freq{\rho}
\leq \Rat{\pi}$ for $\pi$ the dilatation of $\rho$. Indeed, given $n
\in \mathbb{N}$, $R_n - C_n \leq r_n -c_n$ and $c_n >0$ (the case $c_n
=0$ is straightforward) imply $\frac{R_n - C_n}{c_n} \leq \frac{r_n -
  c_n}{c_n}$. Moreover, $C_n \geq c_n$. Hence $\frac{R_n - C_n}{C_n}
\leq \frac{R_n - C_n}{c_n}$. All together, this yields
$\frac{R_n}{C_n} -1 \leq \frac{r_n}{c_n} -1$ which is equivalent to
$\frac{C_n}{R_n} \geq \frac{c_n}{r_n}$. When $n$ tends to infinity, we
obtain $\Rat{\pi} \geq \freq{\rho}$.

Using the fact that the $F$-contraction of $\rho$ is the
$\bar{F}$-dilatation of $\rho$, one obtains $\Rat{\pi'} \le
\freq{\rho}$ for $\pi'$ the contraction of $\rho$.
\end{proof}

\begin{proof}[of Lemma~\ref{lem-cases}]
  The proof is by induction on $n$. The base case, for $n=0$ is
  trivial, since the $R_0$, $C_0$, $r_0$, $c_0$ are all set to $0$ by
  convention. Note that an initialization at step $n=1$ would also be
  possible using cases 1 to 3 in the following cases enumeration.

  Assume now that the lemma holds for $n \in \mathbb{N}$, and let us
  prove it for $n+1$. Consider the prefix of length $n+1$ of $\rho$:
  $\rho_{n+1} = (\ell_0,0) \xrightarrow{\tau_0,a_0} (\ell_1,v_1) \cdots
  (\ell_n,v_n) \xrightarrow{\tau_n,a_n} (\ell_{n+1},v_{n+1})$. We note $e_0,e_1,... $ the edges
  fired along $\rho$. Let us
  detail a careful inspection of cases, depending on the value of $v_n
  + \tau_n$ and whether $\ell_n \in F$.

\begin{description}
\item[Case 1] $\mathsf{frac}(v_n)=0$\footnote{$\mathsf{frac}(v)$
    denotes its fractional part}, $\ell_n \in F$, and $v_n +\tau_n
  \notin \mathbb{N}$.\\
  In this case, $\tau_n = \mathcal{T}_n - \tau'$ with $\mathcal{T}_n \in \mathbb{N}$ and $\tau' \in
  (0,1)$. By definition of the dilatation, $\pi_{n+1}$ is built from
  $\pi_n$ by firing $\mathcal{T}_n$ idling transitions weighted $1/1$ in
  $\A_{cp}$ (possibly interleaved with idling transitions weighted
  $0/0$) followed by the discrete transition weighted $0/0$ that
  corresponds to $e_n$. Thus, $C_{n+1} = C_n + \mathcal{T}_n$, $R_{n+1} = R_n +
  \mathcal{T}_n$, whereas $c_{n+1} = c_n + \tau_n $ and $r_{n+1} = r_n + \tau_n$. In
  particular, $C_{n} \geq c_n$ (induction hypothesis) and $\tau_n < \mathcal{T}_n$
  imply $C_{n+1} \geq c_{n+1}$. Moreover, $R_{n+1} - C_{n+1} = R_n
  -C_n$ and $r_{n+1} -c_{n+1} = r_n -c_n$, and by induction hypothesis
  $R_n - C_n \leq r_n -c_n$. Hence $R_{n+1} - C_{n+1} \leq r_{n+1} -
  c_{n+1}$.
  \begin{figure}[h!]
\begin{center}
\subfigure[Case 1 ($\mathcal{T}_n > \tau_n$).]{
      \begin{tikzpicture}
        \everymath{\scriptstyle}
        \draw(-.5,0) -- (1.5,0);
        \draw[dotted] (1.5,0) -- (3.5,0);
        \draw(3.5,0) -- (5.5,0);
        
        \draw(0,-.1) -- (0,.1);
        \draw(1,-.1) -- (1,.1);
        \draw(4,-.1) -- (4,.1);
        \draw(5,-.1) -- (5,.1);
        
        \draw [snake=brace] (5,-.2) -- (0,-.2);
        \draw(2.5,-.4) node [below] (q) {$\mathcal{T}_n$};
        
        \draw[fill=black] (0,0) circle (1.5pt);
        \draw(0,0.1) node [above] (q) {$v_n$};
        
        \draw[fill=black] (4.85,0) circle (1.5pt);
        \draw(4.85,0.1) node [above] (q) {$v_n+\tau_n$};
      \end{tikzpicture}
     \label{fig-cas1}  }  
    %  \caption{Case 1 ($T_n > t_n$).}
\quad\quad\quad
 \subfigure[Case 2 ($\mathcal{T}_n<\tau_n$).]{
        \begin{tikzpicture}
        \everymath{\scriptstyle}
        \draw(-.5,0) -- (1.5,0);
        \draw[dotted] (1.5,0) -- (3.5,0);
        \draw(3.5,0) -- (5.5,0);
        
        \draw(0,-.1) -- (0,.1);
        \draw(1,-.1) -- (1,.1);
        \draw(4,-.1) -- (4,.1);
        \draw(5,-.1) -- (5,.1);
        
        \draw [snake=brace] (4,-.2) -- (0,-.2);
        \draw(2,-.4) node [below] (q) {$\mathcal{T}_n$};
        
        \draw[fill=black] (0,0) circle (1.5pt);
        \draw(0,0.1) node [above] (q) {$v_n$};
        
        \draw[fill=black] (4.15,0) circle (1.5pt);
        \draw(4.15,0.1) node [above] (q) {$v_n+\tau_n$};
      \end{tikzpicture}
      \label{fig-cas2}}
%      \caption{}
\caption{Cases 1 and 2.}
    \end{center}
  \end{figure}
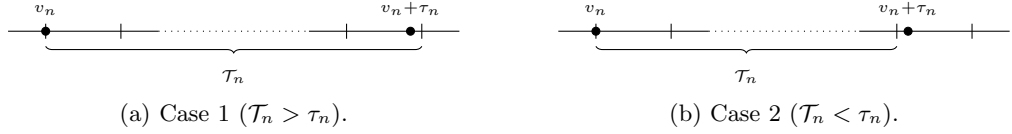
\item[Case 2] $\mathsf{frac}(v_n)=0$, $\ell_n \notin F$, and $v_n +\tau_n
  \notin \mathbb{N}$.\\
  Here, $\tau_n = \mathcal{T}_n + \tau'$ with $\mathcal{T}_n \in \mathbb{N}$ and $\tau' \in (0,1)$. In
  the dilatation, $\mathcal{T}_n$ transitions weighted $0/1$ will be fired before
  taking the transition corresponding to $e_n$. Thus $R_{n+1} = R_n +
  \mathcal{T}_n$, $C_{n+1} = C_n$, whereas $c_{n+1} = c_n $ and $r_{n+1} = r_n + \mathcal{T}_n
  + \tau'$. We immediately deduce that $C_{n+1} \geq c_{n+1}$. Morevoer
  $R_{n+1} - C_{n+1} = R_n + \mathcal{T}_n -C_n \leq r_n - c_n + \mathcal{T}_n < r_n + \mathcal{T}_n +\tau'
  -c_n = r_{n+1} - c_{n+1}$, where the second step uses the induction
  hypothesis.
\item[Case 3] $\mathsf{frac}(v_n)=0$, and $v_n +\tau_n \in \mathbb{N}$.\\
  In this case, $\tau_n= \mathcal{T}_n \in \mathbb{N}$ and exactly $\mathcal{T}_n$
  transitions with reward $1$ will be taken in $\A_{cp}$ before firing
  the transition that corresponds to $e_n$. In other words, the costs
  and rewards are exactly matched in the corner-point abstraction:
  $C_{n+1} - C_n = c_{n+1} - c_n$ and $R_{n+1} - R_n = r_{n+1}
  -r_n$. Notice that the last equalities hold regardless whether
  $\ell_n \in F$. Using the induction hypothesis ($C_n \geq c_n$ and
  $R_n - c_n \leq r_n -c_n$) we easily conclude: $R_{n+1} - C_{n+1} =
  R_n -C_n + r_{n+1} - c_{n+1} + c_n -r_n \leq r_{n+1} - c_{n+1}$, and
  $C_{n+1} = c_{n+1} + C_n -c_n \geq c_{n+1}$.
\begin{figure}[h!]
\begin{center}
\subfigure[Case 3 ($\mathcal{T}_n = \tau_n$).]{
      \begin{tikzpicture}
        \everymath{\scriptstyle}
        \draw(-.5,0) -- (1.5,0);
        \draw[dotted] (1.5,0) -- (3.5,0);
        \draw(3.5,0) -- (5.5,0);
        
        \draw(0,-.1) -- (0,.1);
        \draw(1,-.1) -- (1,.1);
        \draw(4,-.1) -- (4,.1);
        \draw(5,-.1) -- (5,.1);
        
        \draw [snake=brace] (5,-.2) -- (0,-.2);
        \draw(2.5,-.4) node [below] (q) {$\mathcal{T}_n$};
        
        \draw[fill=black] (0,0) circle (1.5pt);
        \draw(0,0.1) node [above] (q) {$v_n$};
        
        \draw[fill=black] (5,0) circle (1.5pt);
        \draw(5,0.1) node [above] (q) {$v_n+\tau_n$};
      \end{tikzpicture}
      \label{fig-cas3}
     }\quad \quad \quad 
\subfigure[Case 4.1 ($\mathcal{T}_n > \tau_n$).]{
       \begin{tikzpicture}
        \everymath{\scriptstyle}
        \draw(-.5,0) -- (1.5,0);
        \draw[dotted] (1.5,0) -- (3.5,0);
        \draw(3.5,0) -- (5.5,0);
        
        \draw(0,-.1) -- (0,.1);
        \draw(1,-.1) -- (1,.1);
        \draw(4,-.1) -- (4,.1);
        \draw(5,-.1) -- (5,.1);
        
        \draw [snake=brace] (5,-.2) -- (0,-.2);
        \draw(2.5,-.4) node [below] (q) {$\mathcal{T}_n$};
        
        \draw[fill=black] (.15,0) circle (1.5pt);
        \draw(.15,0.1) node [above] (q) {$v_n$};
        
        \draw[fill=black] (4.85,0) circle (1.5pt);
        \draw(4.85,0.1) node [above] (q) {$v_n+\tau_n$};
      \end{tikzpicture}
 \label{fig-cas41}}
\caption{Cases 3 and 4.1.}
 \end{center}
  \end{figure}
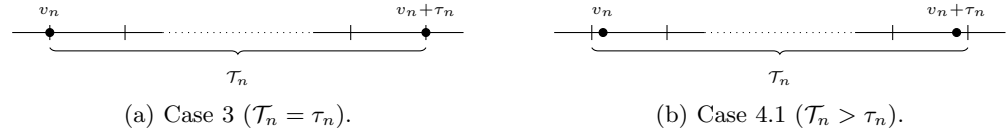  
\item[Case 4] $\mathsf{frac}(v_n) \neq 0$ and $\ell_n \in F$
% and $v_n +t_n \in \mathbb{N}$.\\
\begin{description}
\item[Case 4.1]
  Assume first that the corner in the last state of $\pi_n$ is
  $\bullet$--. Then letting $\mathcal{T}_n = \lceil v_n + \tau_n - \lfloor v_n
  \rfloor\rceil$, in the dilatation, $\mathcal{T}_n$ idling transitions weighted
  $1/1$ will be fired in $\A_{cp}$ before firing the discrete
  transition corresponding to $e_n$. Thus $C_{n+1} = C_n + \mathcal{T}_n$,
  $R_{n+1} = R_n + \mathcal{T}_n$, whereas $c_{n+1} = c_n + \tau_n$ and $r_{n+1} = r_n
  +\tau_n$. We immediately obtain $C_{n+1} \geq c_{n+1}$ using the
  induction hypothesis and the fact that $\mathcal{T}_n > \tau_n$. Moreover, $R_{n+1}
  - C_{n+1} = R_n -C_n \leq r_n - c_n = r_{n+1} - c_{n+1}$.

  Note that the picture on Fig.~\ref{fig-cas41} represents the case
  $v_n + \tau_n \notin \mathbb{N}$, but the reasoning is valid for $v_n +
  \tau_n \in \mathbb{N}$ as well.
\item[Case 4.2] Assume now that the corner in the last state of
  $\pi_n$ is --$\bullet$. In this situation, we cannot conclude
  immediately, since $C_{n+1} = C_n + \mathcal{T}_n$ and $c_{n+1} = c_n + \tau_n$,
  with $\mathcal{T}_n= \lceil v_n + \tau_n - \lceil v_n \rceil \rceil$ is
  incomparable to $\tau_n$ in general. Instead, we need to reason in a
  more global way, taking into account some previous steps in $\rho$
  and $\pi$. Let us consider the least index $i$ such that the corner
  of the last state in $\pi_{n-i}$ is not --$\bullet$. For this index,
  $\pi_{n-i}$ ends either with the pointed region $((\lfloor v_{n-i}
  \rfloor,\lceil v_{n-i} \rceil),\bullet$--$)$ or with
  $(\{v_{n-i}\},\bullet)$.  We then consider the suffix of path
  $\pi_{n+1}$ after $\pi_{n-i}$. Notice that the clock $x$ was not
  reset along this suffix (since no pointed region of the form
  $(\cdot,\cdot,\bullet)$ was reached). For this part, the accumulated
  reward is $\mathcal{T}_{n,i}=\lceil v_n + \tau_n - \lfloor v_{n-i} \rfloor
  \rceil$. %, denoted by $T_{n,i}$
  The corresponding part in $\rho$ has an accumulated delay
  $\tau_{n,i}=v_n + \tau_n - v_{n-i} = \sum_{j=n-i}^n \tau_j$ (since the clock
  has not been reset). %, denoted $t_{n,i}$.
  Note that $\mathcal{T}_{n,i} \ge \tau_{n,i}$.
  
  Let us now discuss the cost accumulated along
  the suffix of path $\pi_{n+1}$ after $\pi_{n-i}$.  By definition of
  the dilatation,
  %, there are only transitions weighted $1/1$ or $0/0$ along this suffix.
  %we know that along this suffix, we only meet
  %transitions weighted by $1/1$ or $0/0$. 
  no idling transition can be fired from a state with an $F$-location
  along this suffix, else the last state of $\pi_{n-i+1}$ has not the
  corner --$\bullet$.
  % In particular, if we cross a
  % location $\ell$ such that $\ell \not\in F$, the corresponding
  % transition will be of the form $0/0$. \fbox{TB: Clear enough?}
  Thus, the accumulated cost along the suffix of path
  $\pi_{n+1}$ after $\pi_{n-i}$ is equal to $\mathcal{T}_{n,i}$. However, the
  corresponding part in $\rho$ has an accumulated cost $c_{n,i}$
  %of $c_{n,i}$
  %\fbox{TB: notation boaf???}  with $c_{n,i} \le t_{n,i}$ 
  smaller than $\tau_{n,i}$ (due to the potential time spent in locations
  not in $F$).
  
  The above discussion
  can be summarised as follows: $ R_{n+1} = R_{n-i} + \mathcal{T}_{n,i}$,
  $C_{n+1} = C_{n-i} + \mathcal{T}_{n,i}$, $r_{n+1} = r_{n-i} + \tau_{n,i}$ and
  $c_{n+1} = c_{n-i} + c_{n,i}$ with $\mathcal{T}_{n,i} \ge \tau_{n,i} \ge
  c_{n,i}$.  We can thus derive that $C_{n+1} \ge c_{n+1}$, using both
  the induction hypothesis stating that $C_{n-i} \ge c_{n-i}$ and the
  fact that $\mathcal{T}_{n,i} \ge c_{n,i}$.  It remains to prove that $R_{n+1}
  - C_{n+1} \le r_{n+1} - c_{n+1}$.  By the above equalities, this is
  equivalent to prove that $R_{n+i} - C_{n+i} \le (r_{n+i} - c_{n+i})
  + (\tau_{n,i} - c_{n,i})$ which is true by the induction hypothesis
  stating that $(R_{n+i} - C_{n+i}) \le (r_{n+i} - c_{n+i})$ and the
  fact that $\tau_{n,i} \ge c_{n,i}$.
  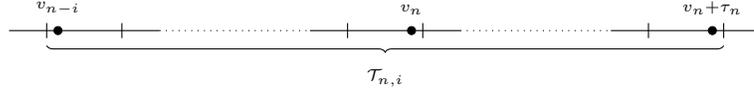
\begin{figure}
\begin{center}
      \begin{tikzpicture}
        \everymath{\scriptstyle}
        \draw(-.5,0) -- (1.5,0);
        \draw[dotted] (1.5,0) -- (3.5,0);
        \draw(3.5,0) -- (5.5,0);
        \draw[dotted] (5.5,0) -- (7.5,0);
        \draw(7.5,0) -- (9.5,0);
        
        \draw(0,-.1) -- (0,.1);
        \draw(1,-.1) -- (1,.1);
        \draw(4,-.1) -- (4,.1);
        \draw(5,-.1) -- (5,.1);
        \draw(8,-.1) -- (8,.1);
        \draw(9,-.1) -- (9,.1);
        
        \draw [snake=brace] (9,-.2) -- (0,-.2); 

        \draw(4.5,-.4) node [below] (q) {$\mathcal{T}_{n,i}$};
        
        \draw[fill=black] (.15,0) circle (1.5pt);
        \draw(.15,0.1) node [above] (q) {$v_{n-i}$};
        
        \draw[fill=black] (4.85,0) circle (1.5pt);
        \draw(4.85,0.1) node [above] (q) {$v_n$};

        \draw[fill=black] (8.85,0) circle (1.5pt);
        \draw(8.85,0.1) node [above] (q) {$v_n+\tau_n$};
      \end{tikzpicture}
      \label{fig-cas42}\caption{Case 4.2}
    \end{center}
  \end{figure}
\end{description}
\item[Case 5] $\mathsf{frac}(v_n) \neq 0$ and $\ell_n \notin
  F$% and $v_n +t_n \in \mathbb{N}$.\\
\begin{description}
\item[Case 5.1]
  Symmetrically to what precedes, the easy case is when the corner in
  the last state of $\pi_n$ is --$\bullet$. Then, letting $\mathcal{T}_n = \lfloor
  v_n + \tau_n\rfloor - \lceil v_n \rceil < \tau_n$, we can write $R_{n+1} =
  R_n + \mathcal{T}_n$ and $C_{n+1} = C_n$. Since $c_{n+1} = c_n$ and
  $r_{n+1} = r_n + \tau_n$, we deduce the desired inequalities.
  \begin{figure}[h!]
    \begin{center}
      \begin{tikzpicture}
        \everymath{\scriptstyle}
        \draw(-.5,0) -- (1.5,0);
        \draw[dotted] (1.5,0) -- (3.5,0);
        \draw(3.5,0) -- (5.5,0);
        
        \draw(0,-.1) -- (0,.1);
        \draw(1,-.1) -- (1,.1);
        \draw(4,-.1) -- (4,.1);
        \draw(5,-.1) -- (5,.1);
        
        \draw [snake=brace] (4,-.2) -- (1,-.2);
        \draw(2.5,-.4) node [below] (q) {$\mathcal{T}_n$};
        
        \draw[fill=black] (.85,0) circle (1.5pt);
        \draw(.85,0.1) node [above] (q) {$v_n$};
        
        \draw[fill=black] (4.15,0) circle (1.5pt);
        \draw(4.15,0.1) node [above] (q) {$v_n+\tau_n$};
      \end{tikzpicture}
      \label{fig-cas51}
      \caption{Case 5.1 ($\mathcal{T}_n<\tau_n$).}
    \end{center}
  \end{figure}
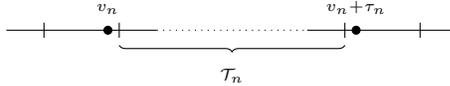
\item[Case 5.2] Assume now that the corner in the last state of
  $\pi_n$ is $\bullet$--. Therefore, the last pointed region in
  $\pi_n$ is $((\lfloor v_{n} \rfloor,\lceil v_{n}
  \rceil),\bullet$--$)$, and we let $\mathcal{T}_n = \lfloor v_n + \tau_n - \lfloor
  v_n \rfloor \rfloor$. By definition of the dilatation, this can only
  happen if $\ell_{n-1} \notin F$.  We then consider the least index
  $i$ such that the last corner in $\pi_{n-i}$ is not $\bullet$--. For
  this index, the last pointed region in $\pi_{n-i}$ is either
  $((\lfloor v_{n-i} \rfloor,\lceil v_{n-i} \rceil),$--$\bullet)$ or
  $(\{v_{n-i}\},\bullet)$. Moreover, all locations $\ell_j$ for $n-i
  \leq j \leq n$ are not in $F$. We define $\mathcal{T}_{n,i} = \lfloor v_n +\tau_n
  \rfloor - \lceil v_{n-i} \rceil$. Note that $\mathcal{T}_n \le \sum_{j=n-i}^n
  \tau_j$. Using this notation, $R_{n+1} = R_{n-i} + \mathcal{T}_{n,i}$, and
  $C_{n+1} = C_{n-j}$. In $\A$, $r_{n+1} = r_{n-j} + \sum_{j=n-i}^n
  \tau_j$ and $c_{n+1} = c_{n-i}$. We trivially derive $C_{n+1} \geq
  c_{n+1}$ using the analogous induction hypothesis at rank
  $n-i$. Moreover, $R_{n+1} - C_{n+1} = R_{n-i} + \mathcal{T}_{n,i} - C_{n-i} <
  R_{n_i} + \sum_{j=n-i}^n \tau_j - C_{n-i} \leq r_{n-i} + \sum_{j=n-i}^n
  \tau_j - c_{n-i} = r_{n+1} - c_{n+1}$, using the induction hypothesis
  at rank $n-i$ in the next to last step.
%  maximal suffix of $\pi_{n+1}$ such
%   that for each $j$, $\ell_j \notin F$ and the last pointed region in
%   $\pi_{n-j}$ is $((\lfloor v_{n-j} \rfloor,\lceil v_{n-j} \rceil),
%   \bullet$--$)$. Similarly to the previous case, one can argue that
%   the accumulated reward in this suffix agrees with the accumulated
%   delay in the corresponding part of $\rho$. The accumulated cost and
%   time spent in $F$ are both zero. Using the induction hypothesis at
%   rank $n-i$, one easily obtains the desired inequalities.
% \item[Case 6] $\{v_n\} \neq 0$, $\ell_n \in F$, and $v_n +t_n \notin
%   \mathbb{N}$.\\
%   This case is very similar to case 4. Assume first that the corner in
%   the last state of $\pi_n$ is $\bullet$--. Then letting $t = \lceil
%   v_n + t_n - \lceil v_n \rceil \rceil$, in the dilatation, $t$ idling
%   transitions weighted with $1/1$ will be fired before the discrete
%   transitions that corresponds to $e_n$. Therefore, $C_{n+1} = C_n +
%   t$ and $R_{n+1} = R_n +t$ whereas $c_{n+1} = c_n + t_n$ and $r_{n+1}
%   = r_n + t_n$.
% \item[Case 7] $\{v_n\} \neq 0$, $\ell_n \notin F$, and $v_n +t_n
%   \notin \mathbb{N}$.\\
\end{description}

  \begin{figure}[h!]
    \begin{center}
      \begin{tikzpicture}
        \everymath{\scriptstyle}
        \draw(-.5,0) -- (1.5,0);
        \draw[dotted] (1.5,0) -- (3.5,0);
        \draw(3.5,0) -- (5.5,0);
        \draw[dotted] (5.5,0) -- (7.5,0);
        \draw(7.5,0) -- (9.5,0);
        
        \draw(0,-.1) -- (0,.1);
        \draw(1,-.1) -- (1,.1);
        \draw(4,-.1) -- (4,.1);
        \draw(5,-.1) -- (5,.1);
        \draw(8,-.1) -- (8,.1);
        \draw(9,-.1) -- (9,.1);
        
        \draw [snake=brace] (8,-.2) -- (1,-.2); 

        \draw(4.5,-.4) node [below] (q) {$T_{n,i}$};
        
        \draw[fill=black] (.85,0) circle (1.5pt);
        \draw(.85,0.1) node [above] (q) {$v_{n-i}$};
        
        \draw[fill=black] (4.15,0) circle (1.5pt);
        \draw(4.15,0.1) node [above] (q) {$v_n$};

        \draw[fill=black] (8.15,0) circle (1.5pt);
        \draw(8.15,0.1) node [above] (q) {$v_n+\tau_n$};
      \end{tikzpicture}
      \label{fig-cas52}
      \caption{Case 5.2.}
    \end{center}
  \end{figure}
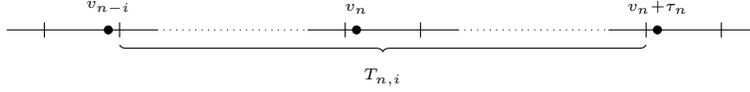
\end{description}
Let us notice that we ignored the unbounded region through the all
proof. However it can be treated exactly in the same way. Indeed, we
can consider the accumulated reward since the last reset and its
difference with the valuation in $\rho$ instead of the corner-point.

Note that in cases 4.2 and 5.2, the induction relies on other cases
(4.1, 5.1, and 1, 2, 3). However, the induction is well-founded since
those cases are treated independently. %\nat{NB: pas tr\`es bien dit}.
\end{proof}
\rule{\linewidth}{.5pt}
\setcounter{theorem}{\value{linkAcpA1}}
\begin{lemma}[From $\A_{cp}^F$ to $\A$, reward-diverging case]
  \label{prop:linkAcpA1}
  For every reward-diverging 
  % feasible\pat{a mentionner dans les preuves}
  run $\pi$ in $\A_{cp}^F$, there exists a non-Zeno run $\rho$ in $\A$
  such that $\pi\in \mathsf{Proj}_{cp}(\rho)$ and $\freq{\rho} =
  \Rat{\pi}$.
\end{lemma}

\begin{proof}%[of Lemma~\ref{prop:linkAcpA1}]
  Given $\rho$ a run and $n \in \mathbb{N}$, we denote by
  $\rho[n]$ the valuation of the $n$-th state along the
  run. Similarly, if $\pi$ belongs to
  $\mathsf{Proj}_{cp}(\rho)$, we consider the states of $\pi$ which
  correspond with a state of $\rho$ (those which are just before a
  discrete transition) and we note $\pi[n]$ the valuation of the
  corner of the $n$-th state if the region is bounded. Otherwise,
  $\pi[n]$ is the sum of all the rewards since the last region
  $\{0\}$.  Lemma~\ref{prop:linkAcpA1} relies on the following lemma:
\begin{lemmastyleS}
\label{lem:pi2rho}
For every reward-diverging run $\pi$ in $\A_{cp}^F$, for all
$\varepsilon >0$, there exists a run $\rho_\varepsilon$ of $\A$
such that, for all $n \in \mathbb{N}$, $|\pi[n] - \rho_\epsilon[n]|
\leq \frac{\epsilon}{2^n}$.
\end{lemmastyleS}
  Let us assume the Lemma~\ref{lem:pi2rho} and consider apart the cases where
  $\Rat{\pi} = 0$ and $\Rat{\pi}>0$.

  Assume first that $\pi$ is a reward-diverging run in
  $\A_{cp}^F$ with $\Rat{\pi} =0$. Given $\varepsilon>0$, let $\rho$
  be a run of $\A$ such that, for all $n \in \mathbb{N}$,
  $|\pi[n]-\rho[n]| < \frac{\varepsilon}{2^n}$. If $C_n$ and $R_n$ are
  the accumulated costs and rewards in the first $n$ steps of $\pi$ in
  $\A_{cp}^F$, then $\Rat{\pi} = \limsup_{n \rightarrow \infty}
  \frac{C_n}{R_n}$ and $\freq{\rho} = \limsup_{n \rightarrow \infty}
  \frac{C_n + \sum_{i \leq n} \alpha_i \varepsilon/{2^i}}{R_n +
    \sum_{i \leq n} \beta_i \varepsilon/{2^i}}$ where for every $i$,
  $\alpha_i \in \{-1,0,1\}$ and $\beta_i \in \{-1,1\}$. Hence
  $\freq{\rho} \leq \limsup_{n \rightarrow \infty} \frac{C_n +
    \varepsilon}{R_n - \varepsilon}$ (because $R_n > \varepsilon$ for
  $n$ large enough). Since $\lim_{n \rightarrow \infty}\frac{C_n}{R_n}
  = 0$ and $\lim_{n \rightarrow \infty} R_n = \infty$, we deduce
  $\limsup_{n \rightarrow \infty} \frac{C_n + \varepsilon}{R_n -
    \varepsilon} =0$ which means $\freq{\rho} =0 = \Rat{\pi}$.

  Assume now that $\pi$ is a reward-diverging run in $\A_{cp}^F$
  with $\Rat{\pi} >0$. Using the same notations as in the previous
  case, $|\frac{C_n}{R_n} - \frac{C_n + \sum_{i \leq n} \alpha_i
    \varepsilon/{2^i}}{R_n + \sum_{i \leq n} \beta_i
    \varepsilon/{2^i}}| \leq \frac{C_n \varepsilon + R_n
    \varepsilon}{R_n (R_n -\varepsilon)}$.
  % \fbox{detail this step?}.
  The latter term tends to $0$ as $n$ tends to infinity. As a
  consequence $\freq{\rho} = \Rat{\pi}$.
%\fbox{single proof for both cases?}
\end{proof}
\begin{proof}[of Lemma~\ref{lem:pi2rho}]
  We show that given a reward-diverging run $\pi$ in
  $\A_{cp}^F$, we can build a run $\rho$ such that $\pi \in
  \mathsf{Proj}_{cp}(\rho)$ and the $\rho[i]$ are as close as we want
  of the $\pi[i]$. More precisely, we show that we can choose suitable
  delays.  In the case where $\pi[i]$ is different than $\pi[i+1]$,
  the choice of the delay allows to be as close as wanted of
  $\pi[i+1]$.  If $\pi[i]$ and $\pi[i+1]$ are equal but an upper bound
  of a region, we can move nearer to $\pi[i+1]=\pi[i]$ by the new
  delay.  If the region is unbounded, and $\pi[i+1]$ larger than the
  maximal constant, it is again a good case.  The only difficulty is
  the case where the new delay force us to move further than
  $\pi[i+1]=\pi[i]$. The solution is to consider globally the sequence
  of the delays in the same corner together with the delay leading to
  it. Thanks to the non-Zenoness, this sequence is necessarily
  finite. Therefore, we can effectively choose suitable delays to
  respect the condition at the end of the sequence and thus all along
  the sequence.  Note this lemma is a simpler version of the Lemma~3
  in \cite{BBL-fmsd08}.
\end{proof}
\rule{\linewidth}{.5pt}
\paragraph{Details on the counterexample Fig.\ref{fig:countlemnZ}.}
We explicit here a reward-diverging run $\pi$ in $\A_{cp}^F$ of zero
ration and explain why every run $\rho$ in $\A$ has a positive
frequency. First, $\pi$ consists (omitting idling transitions weighted
$0/0$) in the following sequence of transitions :
 \begin{multline*}
   (\ell_0,\{0\}^2, \bullet) \xrightarrow{\varepsilon,0/1}(\ell_0,-,-)
   \xrightarrow{a,0/0} \\ \big((\ell_1,-,-)
   \xrightarrow{a,0/0}(\ell_2,-,-)\xrightarrow{\varepsilon,0/1}(\ell_2,-,-) \xrightarrow{a,0/0}\big)^\omega.
   % (\ell_0,\{0\}^2, \bullet) \xrightarrow{\varepsilon,0/1}(\ell_0,(0,1)^2\,|\,x=y, \mbox{right}) \xrightarrow{a,0/0} \big((\ell_1,(0,1)\times\{0\},\text{--}\bullet)\\
   % \xrightarrow{a,0/0}(\ell_2,\{0\}\times(0,1),\mbox{bot.})\xrightarrow{\varepsilon,0/1}(\ell_2,(0,1)^2\,|\,x\le
   % y,\mbox{top-left}) \xrightarrow{a,0/0}\big)^\omega
 \end{multline*}
 The ratio of $\pi$ is thus zero because the accumulated cost of $\pi$
 is zero whereas the reward diverges.  On the other hand, let us
 consider a run $\rho$ of $\A$ and prove that its frequency is
 positive.  Indeed, $\rho$ reads necessarily a word of the form
 $(1-\tau_0,a).\big((\tau_i,a).(1-\tau_i)\big)_{1\le i}$ where $\tau_0
 \in (0,1)$ and $\tau_{i+1}>\tau_i$ for all $0\le i$. The frequency of
 $F = \{\ell_1\}$ in $\rho$ is thus given by:
\[
\freq{\rho} = \limsup_{n\rightarrow \infty} \frac{\sum_{i \leq n}
  \tau_i}{\sum_{i \leq n} 1} > \limsup_{n\rightarrow \infty}
\frac{\sum_{i \leq n} \tau_0}{\sum_{i \leq n} 1}.
\]
Hence, $\freq{\rho} > \tau_0>0$.

\noindent\rule{\linewidth}{.5pt}
\setcounter{theorem}{\value{linkAcpA2}}
\begin{lemma}[From $\A_{cp}^F$ to $\A$, reward-converging case]
  \label{prop:linkAcpA2}
  For every reward-converging 
  % feasible
  run $\pi$ in $\A_{cp}^F$, if $\Rat{\pi}>0$, then for every
  $\varepsilon >0$, there exists a Zeno run $\rho_\varepsilon$ in $\A$
  such that $\pi \in\mathsf{Proj}_{cp}(\rho_\varepsilon)$ and
  $|\freq{\rho_\varepsilon} - \Rat{\pi}| < \varepsilon$.
\end{lemma} 
\begin{proof}%[of Lemma~\ref{prop:linkAcpA2}]
Lemma~\ref{prop:linkAcpA2} uses the following lemma:
\begin{lemmastyleS}
\label{lem:piZeno2rho}
For every reward-converging run $\pi$ in $\A_{cp}^F$, for all
$\varepsilon >0$, there exists a Zeno run $\rho_\varepsilon$ in $\A$
such that $\pi \in \mathsf{proj}_{cp}(\rho_\varepsilon)$ and for all
$n \in \mathbb{N}$, $|\pi[n] - \rho_\varepsilon[n]| < \varepsilon$.
\end{lemmastyleS}
Assuming Lemma~\ref{lem:piZeno2rho} and that $\Rat{\pi}>0$, let
$n_\pi$ be the length of the smallest prefix of $\pi$ such that there
is no transition with non-Zero reward after. Thanks to the convergence
of the reward of $\pi$, $n_\pi$ is necessarily finite. Given
$\varepsilon>0$, the run $\rho_{\varepsilon'}$ given by the
Lemma~\ref{lem:piZeno2rho} with $\varepsilon'=\frac{\varepsilon}{n_\pi
  + 1}$ satisfies the desired
property. % \thomas{La fin de la preuve est ici, je me trompe??? Je
 % suis OK jusque la...}
\end{proof}

\begin{remarkstyleS}
  %\amelie{rq importante pour le lemme9}
  Note that if $\pi$ is a contraction, then $\pi$ is the contraction
  of $\rho_\varepsilon$ defined in the proof of
  Lemma~\ref{lem:piZeno2rho}.
\end{remarkstyleS}

Remark that, if $\pi$ is of ratio $0$, three cases are possible:
  \begin{itemize}
  \item only $F$-locations are along $\pi$ and the reward of $\pi$ is
    $0$, then if $\pi$ is the contraction of $\rho$, $\freq{\rho}=1$,
  \item only $\bar{F}$-locations are along $\pi$, the frequency of
    each run $\rho$ whose contraction is $\pi$, is $0$
  \item otherwise, neither $1$ nor $0$ can be the frequency of an
    run $\rho$ of contraction $\pi$.
  \end{itemize}  
  These results follow immediately from the prohibition of
  zero-delays.

\begin{proof}[of Lemma~\ref{lem:piZeno2rho}]
  Let $\pi$ be a reward-converging run in $\A_{cp}^F$, and
  $\varepsilon \in (0,1)$. As $\pi$ is reward-converging, it ends with
  transitions weighted $0/0$ and its longest prefix
  $\pi'$ not ending with a transition weighted $0/0$
  exists. To prefix $\pi'$, one can associate a finite run $\rho'$ of
  $\A$, as we did for reward-diverging runs (see proof of
  Lemma~\ref{lem:pi2rho}): for all indices $i$ less than the length of
  $\pi'$, $|\pi'[i] - \rho'[i]| < \frac{\varepsilon}{2^i}$. For the
  suffix of $\pi$, composed only of transitions weighted
  $0/0$, we define a corresponding run in $\A$
  with total duration less than $\varepsilon$. This can, \emph{e.g.},
  be achieved by taking successive delays of $\frac{\varepsilon}{2^k}$
  for $k \geq 1$. Concatenating $\rho'$ and the run defined
  above yields a run in $\A$ always $\varepsilon$-close to
  $\pi$.
\end{proof}

\subsection*{Proofs for Section~\ref{subsec:set}}
 \setcounter{theorem}{\value{non-Zeno}}
 \begin{lemma}[non-Zeno case]\label{lm:non-Zeno}
  Let 
  % $\A$ be a timed automaton, $\A_{cp}$ its corner-point abstraction
  % and
  $\{C_1,\cdots,C_k\}$ the set of reachable 
  % feasible\pat{ok ?}  
  SCCs of $\A_{cp}^F$. The set of frequencies of non-Zeno runs of $\A$
  is then $\cup_{1\le i \le k} [m_i,M_i]$ where $m_i$ (resp. $M_i$) is
  the minimal (resp. maximal) ratio for a reward-diverging 
  % feasible
  cycle in $C_i$.
\end{lemma}
\begin{proof}%[of Lemma~\ref{lm:non-Zeno}]
  The lemma is based on the following lemma which expresses the set of
  ratios in $\A_{cp}^F$ for reward-diverging runs ending up in a
  given SCC.
  \begin{lemmastyleS}
    \label{lem:sccAcp-nZ}
    Let $C_i$ be an SCC of $\A_{cp}^F$. If $\mathcal{R}_i$ denotes the
    set of ratios of reward-diverging simple cycles in $C_i$, then the
    set of ratios of reward-diverging runs of $\A_{cp}^F$ ending
    in $C_i$ is the interval $[m_i,M_i]$, where
    $m_i=\min(\mathcal{R}_i)$ and $M_i=\max(\mathcal{R}_i)$.
  \end{lemmastyleS}
  Admitting the lemma for now, we conclude as follows. The set of
  ratios for the reward-diverging runs in $\A_{cp}^F$ is thus
  $\cup_{1\le i \le k} [m_i,M_i]$, where $m_i$ is the minimal
  frequency for a simple reward-diverging cycle in SCC $C_i$, and
  $M_i$ the maximal one. By the Lemma~\ref{prop:linkAcpA1}, we know
  that $\cup_{1\le i \le k} [m_i,M_i]$ is included in
  $\mathcal{F}_{nZ}$, the set of frequencies of non-Zeno runs in
  $\A$.  Moreover, thanks to the Lemma~\ref{prop:linkAAcp} and the
  convexity of the intervals $[m_i,M_i]$, we can show the other
  inclusion $\mathcal{F}_{nZ}\subseteq \cup_{1\le i \le k} [m_i,M_i]$
  as follows. Let $\rho$ be a non-Zeno run in $\A$. We
  distinguish between two cases:
  \begin{itemize}
  \item if the contraction and the dilatation of $\rho$ are both
    reward-diverging, then either the clock is reset infinitely often
    along $\rho$ or from some point on, the value of the clock along
    $\rho$ lies in the unbounded region forever. In the first case,
    there is some state of the form $(\ell,\{0\}, \bullet)$ in
    $\A_{cp}$ which is visited infinitely often by both the
    contraction and the dilatation. In the second case, from some
    point on, they will follow the same transitions between states of
    the form $(\ell,\bot,\alpha_\bot)$ (within the unbounded region).
    In both cases, the contraction and the dilatation both end up in
    the same SCC, say $C_i$. Their frequencies, and that of $\rho$
    (thanks to Lemma~\ref{prop:linkAAcp}) thus lie in the interval
    $[m_i,M_i]$.
  \item if the contraction (resp. dilatation) of $\rho$ is
    reward-converging, the frequency of $\rho$ is $1$ (resp. $0$), in
    this case, the dilatation (resp. contraction) is reward-diverging
    and of ratio $1$ (resp. $0$), therefore $1$ (resp. $0$) is in
    $\cup_{1\le i \le k} [m_i,M_i]$.
  \end{itemize}
  % On the one hand, if the contraction and the dilatation of some
  % run $\rho$ are reward-diverging, they belong a common
  % interval $[m_i,M_i]$ and $\rho$'s frequency is in this
  % interval. On the other hand, if the contraction (resp. dilatation)
  % is reward-converging, the frequency of $\rho$ is $1$ (resp. $0$),
  % in this case, the dilatation (resp. contraction) is
  % reward-diverging and of ratio $1$ (resp. $0$), therefore $1$
  % (resp. $0$) is in $\cup_{1\le i \le k} [m_i,M_i]$.
  As a consequence, the set $\mathcal{F}_{nZ}$ of frequencies of
  non-Zeno runs of $\A$ is equal to the set $\cup_{1\le i \le k}
  [m_i,M_i]$ of frequencies of the reward-diverging runs of
  $\A_{cp}^F$.
\end{proof}

\begin{proof}[of Lemma~\ref{lem:sccAcp-nZ}]
  Let $\pi$ be a reward-diverging run of $\A_{cp}^F$.  To $\pi$
  we associate the SCC $C_\pi$ of $\A_{cp}^F$ where $\pi$ ends up
  in. First observe that the influence of the prefix leading $C_{\pi}$
  is negligible in the computation of the ratio because $\pi$ is
  reward-diverging. Precisely, the ratio of the prefix of length $n$
  (for $n$ large enough) is:
  \[
  \Rat{\pi_{|n}} = \frac{p_{\mathit{pref}} + P_n}{q_{\mathit{pref}} +
    Q_n}
  \]
  where $p_{\mathit{pref}}/q_{\mathit{pref}}$ is the ratio of the
  shortest prefix of $\pi$ leading to $C_{\pi}$. The sequence $Q_n$
  diverges when $n$ tends to infinity because $\pi$ is
  reward-diverging.  Hence $\lim_{n \rightarrow \infty} \Rat{\pi_{|n}}
  = \lim_{n \rightarrow \infty} \frac{P_n}{Q_n}$. As a consequence,
  without loss of generality, we assume that $\A_{cp}^F$ is restricted
  to $C_\pi$ and $\pi$ starts in some state of $C_{\pi}$.

  Observe now that reward-converging cycles in $\A_{cp}^F$ necessarily
  have reward (and hence cost) $0$, and thus do not contribute to the
  ratio $\Rat{\pi}$. Hence we can assume w.l.o.g. that $\pi$ does not
  pass through reward-converging cycles. Following the proof
  of~\cite[Prop.~4]{BBL-fmsd08}, we can decompose $\pi$ into
  (reward-converging) cycles and prove that $\Rat\pi$ lies between $m
  = \min (\mathcal{R}_{C_\pi})$ and $M=\max
  (\mathcal{R}_{C_\pi})$. Note that the extremal values ($m$ and $M$)
  are obtained by a run reaching a cycle with extremal ratio,
  and iterating it forever. % \pat{ai change la preuve a cet endroit car
%     je ne suis pas ok avec l'argument de decomposer en cycles simples}

  % For similar reasons, $\Rat{\pi}$ only depends on the
  % reward-diverging simple cycles that are taken infinitely often
  % along $\pi$. Observe that reward-converging cycles in $\A_{cp}^F$
  % necessarily have reward (and hence cost) $0$, and thus do not
  % contribute to the ratio. The value $\Rat{\pi}$ then clearly lies
  % between $m = \min (\mathcal{R}_{C_\pi})$ and $M=\max
  % (\mathcal{R}_{C_\pi})$. Note that the extremal values ($m$ and
  % $M$) are obtained by a run reaching a cycle with extremal
  % ratio, and iterating it forever.

  Let us now show that any value in the interval $[m,M]$ is the ratio
  of some run in $\A_{cp}^F$ which ends up in the SCC
  $C_{\pi}$. The arguments are inspired by
  \cite{CDEHR-concur10}. Given $\lambda \in (0,1)$, we explain how to
  build a run with ratio $r_\lambda = (1 - \lambda) m + \lambda
  M$. To do so, for $(a_n) \in (\mathbb{Q} \cap (m,M))^{\mathbb{N}}$ a
  sequence of rational numbers converging to $\lambda$, we build an
  run $\pi$ such that $|\Rat{\pi_{|f(n)}} - a_n|<\frac{1}{n}$
  for some increasing function $f \in \mathbb{N}^{\mathbb{N}}$.
  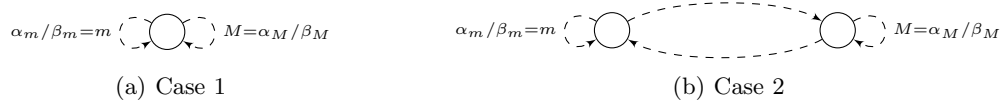
\begin{figure}
    \begin{center}
      \subfigure[Case 1]{
        \begin{tikzpicture}
          \everymath{\scriptstyle}
          \draw(0,0) node [circle,draw,inner sep=1.5pt] (A)
          {\textcolor{white}{le}};
          \draw[-latex',dashed] (A) .. controls +(30:25pt) and
          +(330:25pt) .. (A) node[pos=.5,right]{$M=\alpha_M/\beta_M$};
          \draw[-latex',dashed] (A) .. controls +(150:25pt) and
          +(210:25pt) .. (A) node[pos=.5,left]{$\alpha_m/\beta_m=m$};
        \end{tikzpicture}
        \label{fig:case1}}
      \hspace{1cm}
      \subfigure[Case 2]{
        \begin{tikzpicture}
          \everymath{\scriptstyle}
          \draw(0,0) node [circle,draw,inner sep=1.5pt] (A)
          {\textcolor{white}{le}};
          \draw(3,0) node [circle,draw,inner sep=1.5pt] (A')
          {\textcolor{white}{le}};
          \draw[-latex',dashed] (A') .. controls +(30:25pt) and
          +(330:25pt) .. (A') node[pos=.5,right]{$M=\alpha_M/\beta_M$};
          \draw[-latex',dashed] (A) .. controls +(150:25pt) and
          +(210:25pt) .. (A) node[pos=.5,left]{$\alpha_m/\beta_m=m$};
          \draw[-latex',dashed] (A) .. controls +(30:25pt) and
          +(150:25pt) .. (A') node[pos=.5,right]{};
          \draw[-latex',dashed] (A') .. controls +(210:25pt) and
          +(330:25pt) .. (A) node[pos=.5,left]{};
        \end{tikzpicture}
        \label{fig:case2}}
      \caption{The two possible cases.}
    \end{center}
  \end{figure}

  \paragraph{Case 1.} We first assume for simplicity that in $C_{\pi}$
  two cycles of respective ratio $m$ and $M$ share a state, as
  depicted in Fig.~\ref{fig:case1}, and prove a stronger result: we
  build a run $\pi$ such that $\Rat{\pi_{|f(n)}} = a_n$.  Since
  two cycles, one of minimal ratio, and the other of maximal ratio
  share a common state, it suffices to explain how to combine these
  two cycles to obtain ratio $r_\lambda$. 

  Assume $a_0 = p_0/q_0$ with $(p_0,q_0) \in \mathbb{N}^2$.
  % For any pair $(p_0,q_0) \in \mathbb{N}^2$, with $p_0/q_0 \in (0,1)$,
  Let us show how to build a finite run $\pi$ of ratio $r_{a_0}
  = (1-p_0/q_0)m+(p_0/q_0)M$. Assume $m = \alpha_m/\beta_m$ where
  $\alpha_m$ is the cost of the cycle, and $\beta_m$ its reward, and
  similarly $M = \alpha_M/\beta_M$. Taking $(q_0-p_0)\beta_M$ times
  the cycle of ratio $m$ and then $p_0 \beta_m$ times the cycle of
  ratio $M$ yields an finite run $\pi_0$ with the desired
  property (this will be $\pi_{|f(0)}$). Indeed:
  \[
  \frac{\bigl((q_0-p_0) \beta_M\bigr) \alpha_m + \bigl(p_0
    \beta_m\bigr) \alpha_M}{\bigl((q_0-p_0) \beta_M\bigr) \beta_m +
    \bigl(p_0 \beta_m\bigr) \beta_M} = \frac{(q_0-p_0) \beta_M
    \alpha_m + p_0 \beta_m \alpha_M}{q_0 \beta_M \beta_m} =
  \frac{q_0-p_0}{q_0} m + \frac{p_0}{q_0} M = r_{a_0}.
  \]
  To build an infinite run with ratio $r_\lambda$, we
  incrementally build prefixes $\pi_n$ (which will be $\pi_{|f(n)}$)
  of ratio $r_{a_n}$, starting with $\pi_0$,
  % of ratio $r_{a_0} = r_{p_0/q_0}$.
  as depicted in the picture below.
  \begin{center}
    \scalebox{.9}
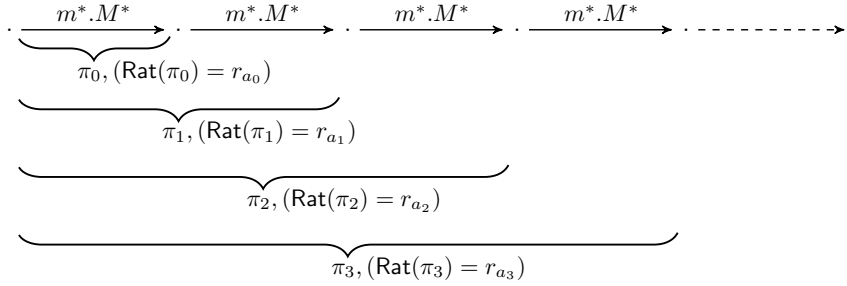
{
      \begin{tikzpicture}[->,>=stealth',shorten >=1pt,auto,node distance=2cm,
        semithick]
        
        \tikzstyle{every state}=[text=black]
        
        \node (init) at (0,0) {.};
        \node (a0) at (2.5,0) {.};
        \node (a1) at (5,0) {.};
        \node (a2) at (7.5,0) {.};
        \node (a3) at (10,0) {.};
        \node (init3) at (0,-3) {};
        \node (init1) at (0,-1) {};
        \node (init0) at (0,-0.1) {};
        \node (init2) at (0,-2) {};
        \node (a1') at (5,-1) {};
        \node (a0') at (2.5,-0.1) {};
        \node (a2') at (7.5,-2) {};
        \node (a3') at (10,-3) {};
        \node (a4) at (12.5,0) {};
        
        \path (init) edge node [above] {$m^*.M^*$} (a0)
        (a0) edge node [above] {$m^*.M^*$} (a1)     
        (a1) edge node [above] {$m^*.M^*$} (a2)     
        (a2) edge node [above] {$m^*.M^*$} (a3)     
        (a3) edge [dashed] node [above] {} (a4)     
        ;
        \draw[decoration={brace,amplitude=10},decorate,thick,>=] (a1') -- (init1) node[pos=.5,below=6pt, xshift=34pt]{$\pi_1,(\Rat{\pi_1}=r_{a_1})$};
        \draw[decoration={brace,amplitude=10},decorate,thick,>=] (a0') -- (init0) node[pos=.5,below=6pt, xshift=34pt]{$\pi_0,(\Rat{\pi_0}=r_{a_0})$};
        \draw[decoration={brace,amplitude=10},decorate,thick,>=] (a2') -- (init2) node[pos=.5,below=6pt, xshift=34pt]{$\pi_2,(\Rat{\pi_2}=r_{a_2})$};
        \draw[decoration={brace,amplitude=10},decorate,thick,>=] (a3') -- (init3) node[pos=.5,below=6pt, xshift=34pt]{$\pi_3,(\Rat{\pi_3}=r_{a_3})$};
      \end{tikzpicture}}
  \end{center}

  Run $\pi_{n+1}$ has $\pi_n$ as prefix, then iterates the cycle
  of minimal ratio, and finally iterates the cycle of maximal ratio in
  order to compensate $r_{a_n}$ and reach ratio $r_{a_{n+1}}$.  We
  assume $a_n = p_n/q_n$ with $(p_n,q_n) \in \mathbb{N}^2$. In
  $\pi_{n+1}$ the number of iterations of the cycle of ratio $m$
  (resp. the cycle of ratio $M$) is globally $b_{n+1} (q_{n+1} -
  p_{n+1})\beta_M$ (resp. $b_{n+1} p_{n+1}\beta_m$) for some $b_{n+1}
  \in \mathbb{N}_{>0}$.
  % if $a_{n+1} = p_{a_{n+1}}/q_{a_{n+1}}$. 
  This construction ensures that $r_\lambda$ is an accumulation point
  of the set of ratios for the prefixes $\pi_n$. Moreover, since each
  path fragment starts with iterations of the cycle of minimal ratio
  first, $r_\lambda$ is the largest accumulation point of the sequence
  of the ratios of prefixes after each cycle. The sequence of the
  prefixes' ratios is schematized below. The oscillations during a
  cycle become negligible when the length of the run
  increases. In the picture below, they are represented by shorter and
  shorter dashes. %\fbox{am\'eliorer dessin}
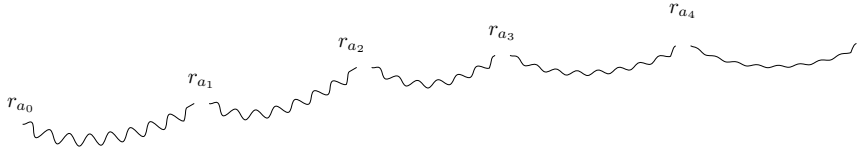
\begin{center}
\scalebox{.8}{
\begin{tikzpicture}

  \tikzstyle{every state}=[text=black]

  \node [inner sep =0] (init) at (-1,0) {};
  \node (init') at (-1,0.3) {$r_{a_0}$};
  \node (a0) at (2.,0.4) {};
  \node (a1) at (4.7,1) {};
  \node (a2) at (7,1.2) {};
  \node (a3) at (10,1.36) {};
  \node (a0') at (2,0.7) {$r_{a_1}$};
  \node (a1') at (4.5,1.3) {$r_{a_{2}}$};
  \node (a2') at (7,1.5) {$r_{a_{3}}$};
  \node (a3') at (10,1.86) {$r_{a_{4}}$};
  \node (a4) at (13,1.4) {};

\draw [decorate,decoration={snake,amplitude=1mm}] (init) .. controls +(330:1cm) and +(210:1cm) .. (a0)  ;

\draw [decorate,decoration={snake,amplitude=.8mm}] (a0) .. controls +(330:1cm) and +(210:1cm) .. (a1)  ;

\draw [decorate,decoration={snake,amplitude=.6mm}] (a1) .. controls +(330:1cm) and +(210:1cm) .. (a2)  ;

\draw [decorate,decoration={snake,amplitude=.4mm}] (a2) .. controls +(330:1cm) and +(210:1cm) .. (a3)  ;

\draw [decorate,decoration={snake,amplitude=.2mm}] (a3) .. controls +(330:1cm) and +(210:1cm) .. (a4)  ;

\end{tikzpicture}}
\end{center}
\paragraph{Case 2} In the general case, in the SCC $C_{\pi}$ of
$\A_{cp}^F$ the cycles with minimal and maximal ratios do not
necessarily share a common state: two finite runs connect the
two cycles, as represented on Fig~\ref{fig:case2}. We fix two cycles
of minimal and maximal ratios, and two finite paths $\pi_{mM}$ and
$\pi_{Mm}$ that connect those cycles in $C_\pi$. Similarly to the
first case, we show how to build a sequence of finite runs
$(\pi_n)$ with $|\Rat{\pi_n} - r_{a_n}| <\frac{1}{n}$, and prove that
the influence of the finite paths linking the cycles is negligible
when $n$ tends to infinity. The run $\pi_{n+1}$ is defined as
the concatenation of $\pi_n$ with $\pi_{Mm}$ then iterations of the
cycle of minimal ratio $m$ then $\pi_{mM}$ and ending with iterations
of the cycle with maximal ratio $M$. If $\tilde{p}$ and $\tilde{q}$
are respectively the cost and the reward of $\pi_{mM}$ and $\pi_{Mm}$
together, then the ratio of $\pi_{n+1}$ is:
\[
\Rat{\pi_{n+1}} = \frac{b_{n+1} (q_{n+1} - p_{n+1}) \beta_M \alpha_m +
  b_{n+1} p_{a_{n+1}} \beta_m \alpha_M + (n+1) \tilde{p}}{b_{n+1}
  (q_{n+1} - p_{n+1}) \beta_M \beta_m + b_{n+1} p_{a_{n+1}} \beta_m
  \beta_M + (n+1) \tilde{q}}
\]
% \[
% \Rat{\pi_{n+1}} = \frac{b_{n+1} (q_{a_{n+1}} - p_{a_{n+1}}) p_m + b_{n+1} p_{a_{n+1}}
%   p_M + (n+1) \tilde{p}}{b_{n+1} \bigl((q_{a_{n+1}} - p_{a_{n+1}}) +
%   p_{a_{n+1}}\bigr) q + (n+1) \tilde{q}}.
% \]
Since this value tends to $r_{a_{n+1}}$ when $b_{n+1}$ tends to
infinity, $b_{n+1}$ can be chosen such that $|\Rat{\pi_{n+1}} -
r_{a_{n+1}}| < 1/(n+1)$. This way, $\lim_{n \to\infty} \Rat{\pi_n}$
agrees with $\lim_{n \to\infty} r_{a_n}$, that is $\lim_{n \to\infty}
\Rat{\pi_n} = r_\lambda$. The function $f$ is defined by `$f(n)$ is
the length of $\pi_n$'.
% the limit of $(\Rat{\pi_n})_n$ agrees with the limit of
% $(r_{a_n})_n$ when $n$ tends to $+\infty$, that is $r_\lambda$.
\end{proof}
\rule{\linewidth}{.5pt}
 \setcounter{theorem}{\value{Zeno}}
\begin{lemma}[Zeno case]
  \label{lm:Zeno}
  Given $\pi$ a reward-converging 
  % feasible
  run in $\A_{cp}^F$, it is decidable whether there exists a Zeno run
  $\rho$ such that $\pi$ is the contraction of $\rho$ and
  $\freq{\rho} = \Rat{\pi}$.
\end{lemma}
\begin{proof}%[of Lemma~\ref{lm:Zeno}]
  This proof is composed of two parts. First, we study how to detect
  if the reward-converging execution $\pi$ in $\A_{cp}^F$ is a
  contraction, that is if there exists an execution in $\A$ whose
  contraction is $\pi$. If $\pi$ is a contraction, then by
  Lemma~\ref{prop:linkAcpA2}, we can construct a Zeno execution $\rho$
  in $\A$ whose contraction is $\pi$ (in particular
  $\freq{\rho}\ge\Rat{\pi}$) and such that $\freq{\rho}$ is as near as
  we want from $\Rat{\pi}$. The second step is to decide if we can
  construct an \emph{optimal} $\rho$ in $\A$, that is a Zeno execution
  $\rho$ whose contraction is $\pi$ and such that
  $\freq{\rho}=\Rat{\pi}$.
  %In the second step of the proof, we give
  %necessary and sufficient conditions on $\pi$ for the existence of
  %such a $\rho$. 
  To do so, we see that we can study $\pi$
  independently on each fragment between its resets. For each of these
  fragments of $\pi$, we provide necessary and sufficient conditions
  for them to be exactly reflected in $\A$. Thus, there
   exists an optimal $\rho$ if and only if all of these fragments
   respect these conditions. Indeed, a tiny difference
  between the ratio of a fragment and the corresponding frequency on
  $\rho$ is never neglected because on the one hand $\rho$ is Zeno and
  on the other hand the contraction minimizes the frequency.

  % This proof is composed of two parts. First, we study how to detect
  % if the reward-converging execution $\pi$ in $\A_{cp}^F$ is a
  % contraction, that is if there exists an execution in $\A$ whose
  % contraction is $\pi$. If $\pi$ is a contraction, then by
  % Lemma~\ref{prop:linkAcpA2}, we can construct a Zeno execution $\rho$
  % in $\A$ whose contraction is $\pi$ and of frequency ($\ge\Rat{\pi}$)
  % as near as we want of $\Rat{\pi}$. We want now decide if we can
  % construct an optimal $\rho$ in $\A$, that is a Zeno execution $\rho$
  % whose contraction is $\pi$ and such that $\freq{\rho}=\Rat{\pi}$. In
  % the second step of the proof, we give the conditions on $\pi$ for
  % the existence of a such $\rho$. To do so, we see that we can study
  % $\pi$ independently on each fragment between resets. Thus, there
  % exists an optimal $\rho$ if and only if all of these fragments
  % respect the conditions to be mimiked exactly in $\A$. Indeed, a tiny
  % difference between the ratio of a fragment and the corresponding
  % frequency on $\rho$ is never neglected because on the one hand
  % $\rho$ is Zeno and on the other hand the contraction minimizes the
  % frequency.

  Let $\pi$ be a reward-converging run in $\A_{cp}^F$.  By
  definition of contractions, we easily verify whether $\pi$ is the
  contraction of some run in $\A$. It is the case if and only if the
  two following conditions are satisfied by $\pi$:
  \begin{enumerate}[$(i)$]
  \item from each state of $\pi$ of the form $(\ell, (i,i+1),
    \bullet\text{--})$ where $\ell \notin F$, $\pi$ follows an idling
    transition to $(\ell, (i,i+1), \text{--}\bullet)$;
  \item after each move $(\ell, (i,i+1),
    \bullet\text{--})\xrightarrow{1/1}(\ell, (i,i+1),
    \text{--}\bullet)$ where $\ell \in F$, $\pi$ goes to
    $(\ell,\{i+1\}, \bullet)$ by an idling transition.
  \end{enumerate}
  % on the one hand, from each state of $\pi$ of the form $(\ell,
  % (i,i+1), \bullet\text{--})$ where $\ell \notin F$, $\pi$ fires the
  % idling transition to $(\ell, (i,i+1), \text{--}\bullet)$ and on
  % the other hand, after each move $(\ell, (i,i+1),
  % \bullet\text{--})\xrightarrow{1/1}(\ell, (i,i+1),
  % \text{--}\bullet)$ where $\ell \in F$, $\pi$ goes to
  % $(\ell,\{i+1\}, \bullet)$ by an idling transition.
  We first consider two simple cases:
  \begin{itemize}
  \item If $\Rat{\pi}=0$, then there exists a Zeno run $\rho$ in
    $\A$ whose contraction is $\pi$, such that
    $\freq{\rho}=\Rat{\pi}=0$ if and only if there are only (non
    $F$)-locations along $\pi$. Otherwise, because of the Zenoness of
    $\rho$ and the positivity of all the delays, $\freq{\rho}>0$.
  \item If $\Rat{\pi}=1$ and $\pi$ is the contraction of some run
    $\rho$, then $\freq{\rho}\le\Rat{\pi}$ hence by definition of the
    contraction $\freq{\rho}=\Rat{\pi}$.
  \end{itemize}

  We now assume that $0<\Rat{\pi}<1$. For any run $\rho$ such that
  $\pi$ is the contraction of $\rho$, $\Rat{\pi} \leq
  \freq{\rho}$. Thus in order to have equality we will have to
  minimize delays spent in $F$-locations when building $\rho$.
%
%   , and that there exists a
%   \fbox{Zeno}\pat{ce serait bien d'enlever le Zeno ici} run
%   $\rho$ in $\A$ whose contraction is $\pi$.  Since $\pi$ is the
%   contraction of $\rho$, $\Rat{\pi} \leq \freq{\rho,F}$.
%
  % Let us assume there exists a Zeno run $\rho$ in $\A$ whose
  % contraction is $\pi$. Since $\pi$ is the contraction of $\rho$,
  % $\Rat{\pi} \leq \freq{\rho,F}$. We first treat two simple cases:
  % \begin{itemize}
  % \item If $\Rat{\pi}=0$ there exists a Zeno run $\rho_0$ in
  %   $\A$ whose contraction is $\pi$, such that
  %   $\freq{\rho_0}=\Rat{\pi}=0$ if and only if there are only (non
  %   $F$)-locations along $\pi$. Otherwise, because of the Zenoness
  %   of $\rho_0$ and the positivity of all the delays,
  %   $\freq{\rho_0}>0$.
  % \item If $\Rat{\pi}=1$, then $\freq{\rho}\le\Rat{\pi}$ hence by
  %   definition of the contraction $\freq{\rho}=\Rat{\pi}$.
  % \end{itemize}
  %
%   In the other cases, a necessary condition for the
%   equality to hold is to minimize the frequency when building
%   $\rho$. 
%   Since $\pi$ is reward-converging (resp. $\rho$ Zeno), the
%   cost along $\pi$ (resp. the total time elapsed in $F$-locations
%   along $\rho$) is finite, and the time elapsed in $F$-locations along
%   $\rho$ should be minimized on each fragment of the run separated by
%   a reset of the clock.
%

%  \fbox{Begin here} 

%  \fbox{End here}

  Note that resets of the clock are not reflected in the corner-point
  abstraction, but could easily be. Therefore, in the sequel, we
  abusively speak of resets in $\pi$. In the rest of the proof, we
  will work independently on the reset-free parts of $\pi$, let us
  shortly argue why this reasoning holds in this
  context.  Let $\rho$ be a path containing finitely
  many resets and such that
  $\rho=\rho^1\xrightarrow{a_1}\rho^2\xrightarrow{a_2}\cdots\rho^n$
  where all the $\rho^i$'s are reset-free. Let $\pi$ be the
  contraction of $\rho$, let us notice that $\pi$ can be written as
  $\pi^1\xrightarrow{a_1}\pi^2\xrightarrow{a_2}\cdots\pi^n$ where
  $\pi^i$ corresponds to the contraction of $\rho^i$ (for $1 \le i \le
  n$).
  % Assume that $\freq{\rho_i}=\frac{a_i}{b_i}$ and
  % $\Rat{\pi_i}=\frac{A_i}{B_i}$. We thus have that $\freq{\rho} =
  % \frac{\sum_i{a_i}}{\sum_i{b_i}}$ and $\freq{\rho} =
  % \frac{\sum_i{A_i}}{\sum_i{B_i}}$.
  By definition of the contraction, we know that $\Rat{\pi^i} \le
  \freq{\rho^i}$ for each $1 \le i \le n$.  In particular, if there
  exists $i$ such that $\Rat{\pi^i} \ne \freq{\rho^i}$, it is
  necessarily the case that $\Rat{\pi^i} < \freq{\rho^i}$. In this
  situation, it is clearly impossible to have that $\Rat{\pi} =
  \freq{\rho}$.

  Let us now detail the different cases that can arise:

  \begin{itemize}
  \item Assume there is an unbounded number of resets along $\pi$, and
    assume $\rho$ is a run such that $\pi$ is the contraction of
    $\rho$ and that $\Rat\pi = \freq{\rho}$. Because $\pi$ is
    reward-converging, after some point, all rewards are $0$. Hence,
    only the prefix of $\pi$ before this point contributes to the
    ratio. By construction of the contraction, the ratio of $\pi$ up
    to this point (say it is $a/b$, which is by assumption $<1$) is
    smaller or equal to the frequency of $\rho$ up to that point,
    there must be equality as $\Rat\pi = \freq{\rho}$.  We will
    prove now that from this point on, $\pi$ only visits
    $F$-states. Due to condition~$(i)$ above, if a state
    $(\ell,(i,i+1),\bullet\text{--})$ is visited, then $\ell \in F$
    (otherwise there should be a reward of $1$ on the next edge). If
    $(\ell,\{i\},\bullet)$ is visited, then this is because we have
    just seen a reset, thus $i=0$, and the next move should be an
    idling move (time is strictly increasing), leading to
    $(\ell,(0,1),\bullet\text{--})$. For the same reason we also get
    that $\ell \in F$.
    % Assuming $\pi$ is a contraction, necessarily $\pi$ ends in
    % $F$-states.
    Since zero-delays are forbidden in $\A$, the run $\rho$ 
    % any run $\rho$ in $\A$ with contraction $\pi$
    necessarily spends some positive delay in the $F$-locations. As a
    consequence the equality $\Rat{\pi} = \freq{\rho}$ cannot hold
    because $\Rat\pi = a/b < (a+c)/(b+c) = \freq{\rho}$ for some
    $c>0$ (the duration of the tail of $\rho$). There is no run
    $\rho$ such that $\pi$ is the contraction of $\rho$ and $\Rat{\pi}
    = \freq{\rho}$.
  \item Assume now that the number of resets along $\pi$ is
    finite. The run $\pi$ can be split into fragments between resets.
    As explained at the beginning of the proof, each fragment has to
    be fairly reflected in $\rho$ to not hinder the optimality of the
    construction.  There is finite fragments and an infinite one. The
    finite fragments of $\pi$ are treated as follows.  The two
    conditions that they have to satisfy in order to not block the
    construction of an optimal $\rho$ are: 
    % \pat{je suis perdue a partir de la}
    \begin{itemize}
    \item a discrete transition of $\pi$ going from an $F$-location
      (resp. $\overline{F}$-location) to an $\overline{F}$-location
      (resp. $F$-location) has to be fired from a punctual region
      ($\{i\}$) or from the region $\bot$, moreover this discrete
      transition has to occur after a sub-fragment in $F$-location
      (resp. $\overline{F}$-location) whose reward is positive.
    \item the same way, the end of the fragment has to go from a
      punctual region or from the region $\bot$.
  \end{itemize}
  These conditions are necessary and sufficient, the proof is based on
  the disjunction of cases of the proof of Lemma~\ref{lem-cases}.  In
  the disjunction, we seen that for every $F$-fragment
  (resp. $\bar{F}$-fragment) of $\pi$, the ratio is smaller or equal
  to the frequency of the corresponding fragment in $\rho$.
  Furthermore, the equality holds only for the case $3$ and the cases
  $4.2$ and $5.2$ whether $v_{n-i}$ and $v_n+\tau_n$ belong to
  $\mathbb{N}$. Then every sub-fragment ($F$-fragment or
  $\bar{F}$-fragment) of a finite fragment (separated by resets) has
  to correspond to one of these cases. This is clearly equivalent to
  the two above conditions.
%  we observe separately the fragments with only
%  $F$-locations and with only $\overline{F}$-locations. For the
%  construction of an optimal $\rho$, the reward of each $F$-fragment
%  has to be positive, hence lead to a punctual region or an unbounded
%  one because $\pi$ is a contraction, moreover the exact reward of
%  each $\overline{F}$-fragment has to be possibly elapsed.
  % These fragments have to be minimized independently.
  % In particular, to respect the region along $\pi$,
%  Remembering the proof of Lemma~\ref{lem-cases}, these conditions
%  correspond to the cases in which we can preserve the equality
%  between the ratio of the prefix of $\pi$ and the corresponding
%  frequency.
 % \fbox{reformuler? d\'etailler...}
  
  If no finite fragment hinders the minimization, the infinite suffix of
  $\pi$ without resets has to be considered.  The conditions for this
  infinite fragment are:
  \begin{itemize}
  \item as above, a discrete transition of $\pi$ going from an
    $F$-location (resp. $\bar{F}$-location) to an $\bar{F}$-location
    (resp. $F$-location) has to be fired from a punctual region or
    from the region $\bot$ after a sub-fragment in $F$-location
    (resp. $\bar{F}$-location) whose reward is positive.
  \item moreover, the fragment has to end in $\bar{F}$-locations.
  \end{itemize}
  First, in this infinite fragment the reward is finite. Hence, if
  there is an unbounded alternation of $F$- and $\bar{F}$-locations,
  then there is (at least) an $F$-fragment of reward $0$
  and there is no optimal $\rho$.  Else, as above, the first condition
  is necessary and ensures the good behaviors before stabilization of
  $\pi$ in $F$ or $\bar{F}$. Moreover, if $\pi$ ends in $\bar{F}$,
  then $\rho$ can be constructed choosing delays to minimize the
  frequency by tending to the reward corresponding in
  $\pi$. Otherwise, $\pi$ ends in $F$ with a pointed region
  $((i,i+1),\alpha)$ or $(\bot,\alpha_\bot)$:
  \begin{itemize}
  \item $((i,i+1),\bullet\text{--})$: the sum of the delays can at
    best tend to $i+\varepsilon$,
  \item $((i,i+1),\text{--}\bullet)$: either the last fragment in
    $\bar{F}$ does not end in a suitable region, or $\pi$ is not a
    contraction,
  \item $(\bot,\alpha_\bot)$: the sum of the delays can at best tend
    to the reward of the fragment $+\varepsilon$ because this reward
    has to be always smaller or equal to the sum of the delay to $\pi$
    be the contraction of the built $\rho$.
  \end{itemize}
\end{itemize}
To conclude, a careful inspection of the corner-point abstraction
allows us to decide whether there exists a run $\rho$ in $\A$
whose contraction is $\pi$ and such that $\Rat{\pi} = \freq{\rho}$.
\end{proof}

A similar lemma holds for dilatations:
\begin{lemmastyleS}
  \label{lm:Zeno-bis}
  Given $\pi$ a reward-converging run of $\A_{cp}^F$, it is
  decidable whether there exists a Zeno run $\rho$ such that
  $\pi$ is the dilatation of $\rho$ and $\freq{\rho} = \Rat{\pi}$.
\end{lemmastyleS}
\rule{\linewidth}{.5pt}
 \setcounter{theorem}{\value{setfreq}}
\begin{theorem}
  \label{th:setfreq}
  Let $\mathcal{F}_{\A} = \{\freq{\rho} \mid \rho\ \text{run of}\
  \A\}$ be the set of frequencies of runs in $\A$. We can compute
  $\inf \mathcal{F}_{\A}$ and $\sup \mathcal{F}_{\A}$. Moreover we can
  decide whether these bounds are reached or not. Everything can be
  done in NLOGSPACE.
  % polynomial time. \pat{complexite a verifier et affiner}
\end{theorem}
\begin{proof}%[of Theorem~\ref{th:setfreq}]
  With each run $\pi$ of the corner-point abstraction
  $\A_{cp}^F$ is associated the SCC it ends up in. Let us argue that
  given $C$ an SCC of $\A_{cp}^F$:
  \begin{enumerate}
  \item the infimum of the frequencies of runs of $\A$ whose
    contraction ends up in $C$ can be computed, and
  \item it is decidable whether the bound is reached.
  \end{enumerate}
  First of all, if there is no reward-converging (simple) cycle in
  $C$, then all runs in $\A_{cp}^F$ ending up in $C$ are
  reward-diverging. For each such run $\pi$, there exists a
  non-Zeno run $\rho$ in $\A$ with $\Rat{\pi} = \freq{\rho}$,
  thanks to Lemma~\ref{prop:linkAcpA1}.  In this case,
  Lemma~\ref{lm:non-Zeno} allows us to conclude.

  Assume now that $C$ contains a reward-converging cycle, and let
  $S_{rc}$ be the set of states in $\A_{cp}^F$ that belong to a
  reward-converging cycle. The set $\mathsf{S}$ of cycle-free finite
  runs in $\A_{cp}^F$ ending up in a state of $S_{rc}$ is finite and
  therefore contains a run $\pi_{\min}$ with minimal ratio $r_{\min}$.
  The infimum $r^*$ of the ratios of runs of $\A_{cp}^F$ ending up in
  $C$ is thus $\min(r_{\min},m)$ where $m$ is the minimal ratio of
  reward-diverging simple cycles co-reachable from $C$. Moreover, it
  is also the infimum of the frequencies of runs of $\A$ whose
  contraction ends up staying in $C$.  This bound is reached by a
  non-Zeno run of $\A$ whose contraction ends up staying in $C$ if and
  only if $m=r^*$ and there is a reward-diverging cycle of ratio $m$
  in $C$ (Lemma~\ref{lm:non-Zeno}). We do not need to consider
  non-Zeno runs whose contraction is reward-converging because their
  frequency is necessarily equal to $1$.  On the other hand, the
  infimum may be reached by a Zeno run whose contraction ends up
  staying in $C$. This contraction may be reward-converging or
  diverging. We distinguish three cases: $r_{\min} > m$, $r_{\min} <m$
  and $r_{\min} = m$.
\begin{description}
\item[Case $r_{\min} > m$] The infimum of ratios of runs in
  $\A_{cp}^F$ ending up in $C$ is then $m$: iterating the cycle of
  minimal ratio $m$ many times and then going to a 
    reward-converging cycle in $C$ yields a ratio which tends to
  this infimum. However, the infimum of the ratios is not reached by
  runs of $\A_{cp}^F$ ending up in $C$. A fortiori, frequency
  $m$ cannot be realized by a run in $\A$ whose contraction is
  reward-converging and ends up in $\A_{cp}^F$.
\item[Case $r_{\min} < m$] The infimum of ratios of runs in
  $\A_{cp}^F$ ending up in $C$ is $r_{\min}$. Note that considering
  only runs whose contraction is reward-converging suffices.  Indeed,
  a contraction reward-diverging of a Zeno run is necessarily of ratio
  $0$, therefore its existence would imply that $m=0$.  Using the
  proof of Lemma~\ref{lm:Zeno}, we can decide if there exists a
  reward-converging run $\pi$ of $\A_{cp}^F$ ending up in $C$ such
  that there exists a Zeno run $\rho$ whose contraction is $\pi$ and
  such that $\freq{\rho} = \Rat{\pi}=r_{\min}$.
  
 In this case, $\pi$ can be decomposed either as follows:
 \begin{multline*}
%\[
1):\;\pi=\pi_0.\pi_p.(\ell^{\bar{F}},\{i\},\bullet) \xrightarrow{0/0} (\ell^{\bar{F}},(i,i+1),\bullet\text{---})\\
\xrightarrow{0/1} (\ell^{\bar{F}},(i,i+1),\text{---}\bullet) \bigl( \xrightarrow{0/0}(\ell_i^{\bar{F}},(0,1),\text{---}\bullet)\bigr)_{i\in \mathbb{N}}
%\]
\end{multline*}
or as follows :
\[
2):\;\pi=\pi_0.\pi_p.(\ell^{\bar{F}},\bot,\alpha_\bot) \bigl( \xrightarrow{0/0}(\ell_i^{\bar{F}},\bot,\alpha_\bot)\bigr)_{i\in \mathbb{N}}
\]
where, first $\pi_0$ ends up resetting the clock and its fragments
between resets satisfy the good conditions that is the ones to be a
contraction and the ones of the proof of the Lemma~\ref{lm:Zeno} over
finite fragments, second the suffix contains only $\bar{F}$-locations
and no resets and finally the factor $\pi_p$ satisfies the good
conditions over the beginning of a fragment. Note that $\pi_p$ can
ends in an $F$-location or an $\bar{F}$-location.  If there exists
such a run $\pi$ taking several cycles in the suffix, the same
run $\tilde{\pi}$ whose suffix simply consists in the infinite
iteration of the first cycle taken by $\pi$ satisfies the required
properties. Therefore, the set of states of $\A_{cp}^F$ from which a
suitable suffix runs is computable. The set of cycle-free prefixes
satisfying the conditions and ending up in this set is also
computable. Moreover, adding cycle iterations to the prefix cannot
help to meet the conditions.  As a consequence, the existence of such
a $\pi$ is decidable.
\item[Case $r_{\min} =
  m$] %\amelie{lemme trivial, ecrit dans draft, pas dans ICALP... mais pas besoin ˆ mon avis}\fbox{Thanks to Lemma~\ref{lem:cas-part} (first item),}
  The infimum of the ratios can be the frequency of some Zeno
  run whose contraction is reward-diverging only if $r_{\min} =
  m=0$.  In this case, the infimum is reached by a Zeno run
  whose contraction ending up in $C$ is reward-diverging if and only
  if there is a reward-diverging contraction ending up in $C$ visiting
  only locations of $\bar{F}$, resetting infinitely often the clock
  and visiting finitely often the region $\{1\}$. On the other hand,
  the Zeno runs with reward-converging contraction can be
  treated similarly to the case $r_{\min} < m$ since the introduction
  of a finite number of minimal or reward-converging
  cycles in the minimal prefix cannot help to reach the bound.
\end{description}
We can thus decide whether the infimum is reached for a run of $\A$ by considering 
the question in subsets of runs (sorted with respect to the SCC in which the contraction ends up) forming a partition of the set of the runs of $\A$.

Similarly given $C$ an SCC of $\A_{cp}^F$:
\begin{enumerate}
\item the supremum of the frequencies of runs of $\A$ whose
  dilatation ends up in $C$ can be computed, and
\item it is decidable whether the bound is reached.
\end{enumerate}

Let us now conclude the proof. If the infimum (resp. supremum) of the
frequencies of runs in $\A$ is reached by some run, then
the ratio of its contraction (resp. dilatation) is equal to this
infimum (resp. supremum) frequency. By considering the bounds of the
sets of frequencies for runs of $\A$ whose contraction
(resp. dilatation) end up in each SCC, the bounds of the set of
frequencies of runs of $\A$ are respectively the minimum and the
maximum of these latter. Moreover, the two bounds are reached if and
only if they are reached for some SCC.
\end{proof}

%\subsection*{Proofs for Section~\ref{section:decid}}

\subsection*{Proofs for Section~\ref{section:univ}}
\label{app:univ}

The different variants of the universality problems for timed automata
under frequency-acceptance are incomparable, as illustrated by the
following examples:
\begin{figure}
\begin{center}
\subfigure[$\A_1$.]{
\begin{tikzpicture}
\everymath{\scriptstyle}
\draw(-3.8,0) node (init) {};
\draw(-2,1) node (init') {};
\draw(-3,0) node [circle,draw,inner sep=1.5pt] (l2) {$\ell_0$};
\draw(-2,0) node [circle,draw,inner sep=1.5pt,fill=black] (l0) {$\textcolor{white}{\ell_1}$};
\draw(-1,0) node [circle,draw,inner sep=1.5pt] (l1) {$\ell_2$};

\draw[-latex'] (init) -- (l2) node[pos=.5,above]{};
\draw[-latex'] (init') -- (l0) node[pos=.5,above]{};
\draw[-latex'] (l2) -- (l0) node[pos=.5,above]{$\Sigma$};
\draw[-latex'] (l2) .. controls +(120:1cm) and +(60:1cm) .. (l2) node[pos=.5,above]{$\Sigma$};
\draw[-latex'] (l0) -- (l1) node[pos=.5,above]{$\Sigma$};
\end{tikzpicture}
\label{fig:Afini}}
\hspace{1cm}
\subfigure[$\A_2$.]{
\begin{tikzpicture}
\everymath{\scriptstyle}
\draw(-3.8,0) node (init) {};
\draw(-3,0) node [circle,draw,inner sep=1.5pt] (l2) {$\ell_0$};
\draw(-1.5,0) node [circle,draw,inner sep=1.5pt,fill=black] (l0) {$\textcolor{white}{\ell_1}$};

\draw[-latex'] (init) -- (l2) node[pos=.5,above]{};
\draw[-latex'] (l2) -- (l0) node[pos=.5,above]{$\Sigma,\{x\}$};
\draw[-latex'] (l0) .. controls +(120:1cm) and +(60:1cm) .. (l0) node[pos=.5,above]{$x\le1,\Sigma$};
\draw[-latex'] (l2) .. controls +(120:1cm) and +(60:1cm) .. (l2) node[pos=.5,above]{$\Sigma$};
\end{tikzpicture}
\label{fig:AZ}}
\hspace{1cm}
\subfigure[$\A_3$.]{
\begin{tikzpicture}
\everymath{\scriptstyle}
\draw(-3.8,0) node (init) {};
\draw(-3,0) node [circle,draw,inner sep=1.5pt] (l2) {$\ell_0$};
\draw(-1.5,0) node [circle,draw,inner sep=1.5pt,fill=black] (l0) {$\textcolor{white}{\ell_1}$};

\draw[-latex'] (init) -- (l2) node[pos=.5,above]{};
\draw[-latex'] (l2) -- (l0) node[pos=.5,above]{$x\ge1,\Sigma$};
\draw[-latex'] (l0) .. controls +(120:1cm) and +(60:1cm) .. (l0) node[pos=.5,above]{$\Sigma$};
\draw[-latex'] (l2) .. controls +(120:1cm) and +(60:1cm) .. (l2) node[pos=.5,above]{$\Sigma$};
\end{tikzpicture}
\label{fig:AnZ}}
\caption{Counterexamples for the comparison between universality problems.}\label{fig:As}
\end{center}
\end{figure}
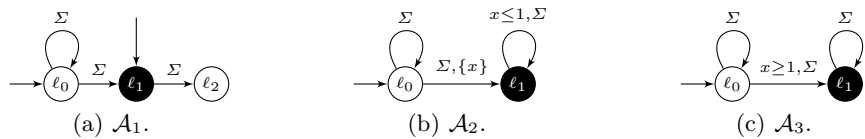

%\thomas{ai reduit. OK?} 
Let us explain why the universality problems
with positive-frequency acceptance are not comparable when considering
respectively finite timed words, Zeno timed words or non-Zeno timed
words.  The three timed automata of Fig.~\ref{fig:As} illustrate this. The timed automaton
$\A_1$ is universal for finite timed words but neither for the Zeno
ones or the non-Zeno ones.  In the same way, $\A_2$ and $\A_3$ are
universal respectively for Zeno timed words and non-Zeno timed words
but not for the other types of words.

\noindent\rule{\linewidth}{.5pt}
 \setcounter{theorem}{\value{ucomp}}
 \begin{theorem}\label{ucomp:infinite}
  % \nat{je couperais bien ce th en plusieurs pour mettre en valeur
  %   les r\'esultats}
  The universality problem for infinite (resp. non-Zeno, Zeno) timed
  words 
  % with positive frequency 
  in a one-clock timed automaton is non-primitive recursive. If two
  clocks are allowed, this problem is undecidable.
\end{theorem}
\begin{proof}%[of Theorem~\ref{ucomp:infinite}]
  We want to check whether $\A$ is universal for Zeno timed words with
  positive frequency. We first check that every Zeno timed word can be
  read in $\A$: this is equivalent to checking that all finite timed
  word can be read in $\A$, and this can be done~\cite{OW04}.
  % , this can be done in ***\pat{utile~?  reference + complexite}.
  Thus, w.l.o.g. we assume that $\A$ reads all Zeno timed words, and
  we now only need to take care of the accepting condition.

  From $\A$ we build the timed automaton $\B$ composed of two copies
  of $\A$, one with a tag $\mathsf{w}$ and one with a tag
  $\mathsf{b}$. The symbol $\mathsf{w}$ (resp. $\mathsf{b}$) is for
  white (resp. black).  If $\ell' \in F$, then for all transitions
  leading to $\ell'$, we will have transitions from the
  $\mathsf{w}$-copy to the $\mathsf{b}$-copy. All others are
  maintained.  In $\B$ once the $\mathsf{b}$-copy is entered, it is
  never left, and for Zeno timed words, the B\"uchi condition can be
  reduced to ``enter the $\mathsf{b}$-copy and read the rest of the
  word''.

  More formally, we define $\B = (L',L'_0,F',\Sigma,\{x\},E')$ as
  follows:
  \begin{itemize}
  \item $L' = L \times \{\mathsf{w},\mathsf{b}\}$, $L'_0 = L_0 \times
    \{\mathsf{w}\}$ and $F' = L \times \{\mathsf{b}\}$;
\item if $\ell \xrightarrow{g,a,X'} \ell'$ is in $E$, then the
  following edges are in $E'$:
  \begin{itemize}
  \item $(\ell, \mathsf{b}) \xrightarrow{g,a,X'} (\ell', \mathsf{b})$
    and $(\ell,\mathsf{w}) \xrightarrow{g,a,X'} (\ell', \mathsf{w})$;
  \item $(\ell,\mathsf{w}) \xrightarrow{g,a,X'} (\ell', \mathsf{b})$
    if $\ell' \in F$;
    % \item $(\ell,\mathsf{w}) \xrightarrow{g,a,X'} (\ell',
    %   \mathsf{w})$ if $\ell' \not\in F$;
    % \item $(\ell, \mathsf{b}) \xrightarrow{g,a,X'} (\ell',
    %   \mathsf{b})$.
  \end{itemize}
\end{itemize}
% A configuration of $\B$ is formally a pair $((\ell,\mathsf{tag},v))$
% where $\ell \in L$, $\mathsf{tag} \in \{\mathsf{w},\mathsf{b}\}$ and
% $v$ is a value for clock $x$, but we will simply write
% $(\ell,\mathsf{tag},v)$ in the following.

The correctness of the transformation is stated in the following
straightforward lemma.

\begin{lemmastyleS}
  $\A$ is universal with positive frequency for Zeno timed words iff
  $\B$ is universal with positive frequency for Zeno timed words.
\end{lemmastyleS}

% \begin{proof}
%   Assume that $w$ is a Zeno timed word accepted by $\A$. Then there is
%   a run $\varrho = (\ell_0,v_0) \rightarrow (\ell_1,v_1) \rightarrow
%   \dots$ which reads $w$ and visits $F$ (that is, there exists some
%   index $i$ such that $\ell_i \in F$). Then the run $\varrho' =
%   (\ell_0,\mathsf{w},v_0) \rightarrow \dots \rightarrow
%   (\ell_{i-1},\mathsf{w},v_{i-1}) \rightarrow (\ell_i,\mathsf{b},v_i)
%   \rightarrow (\ell_{i+1},\mathsf{b},v_{i+1}) \rightarrow \dots$ is an
%   accepting run for word $w$ in $\B$.  The converse also holds, and
%   can be proven similarly.
% \end{proof}

% \fbox{TB: Do we want to write the proof of the above lemma?}

In the following we will write configurations of $\B$ as triples
$(\ell,\mathsf{tag},v)$ where $\ell$ is a location of $\A$,
$\mathsf{tag} \in \{\mathsf{w},\mathsf{b}\}$ and $v$ is a value for
the clock. Furthermore we let $M$ be the maximal constant the clock is
compared with in $\B$, and if $0 \le c<M$, $I_c$ denotes the interval
$(c;c+1)$ whereas $I_M$ denotes the interval $(M;+\infty)$.

% We assume $M$ is the maximal constant the clock $x$ is compared to in
% $\mathcal{A}$.

%\pat{Thomas, je garde ma formulation, l'autre ne convient pas} 
An infinite execution $(\ell_0,\mathsf{tag}_0,v_0) \rightarrow
(\ell_1,\mathsf{tag}_1,v_1) \dots \rightarrow
(\ell_p,\mathsf{tag}_p,v_p) \rightarrow \dots$ in $\B$
\emph{stabilizes} after $n_0$ steps whenever there exists some integer
$c$ such that for every $n \ge n_0$, either $v_n \in I_c$, or $v_n=0$,
or $v_n \in I_0$. Note that every infinite run which reads a Zeno word
stabilizes.

\medskip Given a Zeno timed word $w$, our aim is to analyze all runs
that read $w$, so that we will be able to detect whether $w$ is
accepted or not. Therefore we need to be able to compute a uniform
bound after which all runs which read $w$ stabilizes. This is the aim
of the next lemma.

% Our aim is to check for non-universality of $\mathcal{A}$. Therefore
% we need to be able to analyze all executions that read a given Zeno
% timed word $w$. Hence we need to compute a uniform bound after which
% any execution that reads $w$ stabilizes.

\begin{lemmastyleS}[Uniform stabilization]
  \label{lemma-stabilize}
  Let $w$ be a Zeno timed word (which can be read in $\B$ by
  assumption).  Then, there exists some integer $n_0$ such that every
  execution that reads $w$ stabilizes after $n_0$ steps.
\end{lemmastyleS}

To prove this lemma we need the notion of duration of a timed word. If
$w = (t_0,a_0) \dots (t_k,a_k)$ is a finite timed word, the duration
of $w$ is $t_k$ and is denoted $\mathit{duration}(w)$. If $w$ is an
infinite timed word, we let $w_{\le n}$ be the $n$-th prefix of $w$,
then the duration of $w$ is $\mathit{duration}(w) = \lim_{n\to\infty}
\mathit{duration}(w_{\le n})$. Note that this is a finite value if and
only if $w$ is a Zeno timed word.

\begin{proof}
%   \fbox{\`a prouver, mais j'y crois :-)}
%
%   \noindent \fbox{C'est le fait que $n_0$ soit uniforme pour tous les
%     chemins qui n'est pas \'evident} 
%
%   Let $n$ be an integer such that the length (in terms of duration) of
%   $w$ after $n$ steps (limit duration of $w_{\ge n}$ \fbox{ou $w_{>
%       n}$}) is smaller than $1$, say $d$. Let $\rho$ be an execution
%   that reads $w$. Let $\{\alpha_1,\dots,\alpha_p\}$ be the various
%   values of clock $x$ along executions that read $w$ after $n$ steps
%   (there are finitely many such values). Assume w.l.o.g. that
%   $\textsf{frac}(\alpha_1) \le \textsf{frac}(\alpha_2) \le \dots$. Let
%   $i_0$ be the smallest integer such that $\alpha_i+d$ has not the
%   same integral part as $\alpha_i$, and let $t$ \fbox{finir, ca semble
%     ok, mais c'est moche pour le moment} 
%   \bigskip
%
  Let $D$ be the duration of $w$, and for every $n$, $d_n$ be the
  duration of $w_{\le n}$ (the prefix of length $n$ of $w$). We fix
  some integer $N$ such that $D-d_N<1$, and we write $V_N$ for the set
  of possible valuations for clock $x$ after having read prefix
  $w_{\le n}$ in $\B$. The set $V_N$ is finite.

  Let $\varrho = (\ell_0,\mathsf{tag}_0,v_0) \rightarrow
  (\ell_1,\mathsf{tag}_1,v_1) \dots$ be a run that reads $w$ in
  $\B$. Then $v_N \in V_N$, and either $\varrho$ stabilizes after $N$
  steps, or $\varrho$ stabilizes after $N+k$ steps, where $k$ is the
  smallest integer such that $\lfloor v_N \rfloor + 1 = \lfloor
  v_N+d_{N+k}-d_N \rfloor$. This property does not depend on the
  choice of the execution, but only on the value $v_N$. Since there
  are finitely many $v_N$, as $V_N$ is finite, we can find a maximal
  $k$, denoted $k_0$ which will work for all the $v_N$. We choose
  $n_0$ either as $N$ (in case $\lfloor v_N \rfloor = \lfloor
  v_N+d_{N+k}-d_N \rfloor$ for every $k$), or as $N+k_0$ where $k_0$
  is defined as above.  By the previous analysis, we are done, every
  run $\varrho$ which reads $w$ stabilizes after $n_0$ steps.
  % \thomas{Ai legerement modifie par rapport au $k_0$. OK?}
  % \fbox{TB: On peut juste dire qu'on prend $n_0=N+k_0$ sans
  %   distinguer les cas?}
  % 
  % Let $\varrho = (\ell_0,\mathsf{tag}_0,v_0) \rightarrow
  % (\ell_1,\mathsf{tag}_1,v_1) \dots$ be a run that reads $w$ in
  % $\B$. Then $v_N \in V_N$, and either $\varrho$ stabilizes after $N$
  % steps, or $\varrho$ stabilizes after $N+k$ steps, where $k$ is the
  % smallest integer such that $\mathit{integral}(v_N) + 1 =
  % \mathit{integral}(v_N+d_{N+k}-d_N)$. This property does not depend
  % on the choice of the execution, but only on the value $v_N$.  We
  % choose $n_0$ either as $N$ (in case $\mathit{integral}(v_N) =
  % \mathit{integral}(v_N+d_{N+k}-d_N)$ for every $k$), or as $N+k_0$
  % where $k_0$ is the smallest integer such that
  % $\mathit{integral}(v_N) + 1 = \mathit{integral}(v_N+d_{N+k}-d_N)$.
  % By the previous analysis, we are done, every run $\varrho$ which
  % reads $w$ stabilizes after $n_0$ steps.
\end{proof}

\begin{example}
  Take for instance the following timed automaton:
  
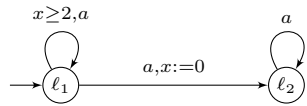
\begin{center}
    \begin{tikzpicture}
      \everymath{\scriptstyle}

      \draw(2.2,0) node (init) {}; 

      \draw(3,0) node [circle,draw,inner sep=1.5pt] (l1)
      {$\ell_1$}; 

      \draw(6,0) node [circle,draw,inner sep=1.5pt] (l2) {$\ell_2$};
      
      \draw[-latex'] (init) -- (l1) node[pos=.5,above]{};
      \draw[-latex'] (l1) -- (l2) node[pos=.5,above]{$a,x:=0$};
      \draw [-latex'] (l1) .. controls +(120:1cm) and +(60:1cm) .. (l1) node [midway,above] {$x \ge 2,a$};
      \draw [-latex'] (l2) .. controls +(120:1cm) and +(60:1cm) .. (l2) node [midway,above] {$a$};
    \end{tikzpicture}
  \end{center}
  and the Zeno timed word $w = (a,2)(a,2+1/2)(1,2+3/4)... $ whose
  duration is $3$. There are infinitely many runs that read $w$, each
  one depends on the time where it takes the transition to
  $\ell_2$. We have that all runs that read $w$ stabilize after $2$
  steps.
\end{example}

Hence a Zeno word that is read in $\B$ will have a prefix (up to $n_0$
steps) and a Zeno tail that will satisfy clock constraints in a
``straightforward'' manner. 
% Furthermore the following property holds.
% \fbox{TB: We assume time strictly increasing? Otherwise guards $x=0$
%   could be problematic...}
%
% \fbox{TB: I imagine we can not bound $n_0$, but this is not important,
%   since we will use non determinism to guess it. Am I right?}
%
% \begin{lemmastyleS}
%   Assume that $w$ is a Zeno timed word which can be read from a
%   configuration $(\ell,v)$ along a path $\pi$ in $\A$. Assume
%   furthermore that $v$ and $v+\mathit{duration}(w)$ lie in the same
%   unit interval $(c;c+1)$. \pat{ajouter que $w$ peut etre egal a $c$}
%   Then for every $0<\lambda<1$, the timed word $w_\lambda$ defined by
%   contracting $w$ with $\lambda$ is also read in $\A$ along the path
%   $\pi$.
% \end{lemmastyleS}
%
% \begin{proof}
%   \patlong{needs to be done, in case it is useful :-)}
% \end{proof}
We will take advantage of this structure to draw an algorithm for
deciding whether there is a Zeno word that is not accepted by $\B$.
% We deal with words that are not read in $\mathcal{A}$ in a different
% widget.  \fbox{\`a faire, mais \c{c}a ne doit pas \^etre le plus
%   dur, si on a le droit \`a NPR}

\paragraph{\textbf{Tail of Zeno words.}}
We construct a finite automaton $\B_f$ that will somehow recognize the
tail of Zeno behaviours. We build the automaton as follows: the set of
states is $Q=L \times \{\mathsf{w},\mathsf{b}\} \times \left(\{0,1,
  \dots, M-1,M\}\right)$.
% A state $(\ell,c)$ is said to be sane for letter $a$ whenever either
% $c=0$ or $c>0$ and there is no transition in $\mathcal{A}$ of the
% form $\ell \xrightarrow{g,a,\{x\}} \ell'$ (\textit{i.e.} which
% resets clock $x$!)  with ``$(c<x<c+1) \subseteq g$''.
\begin{itemize}
\item There is a transition $(\ell,\mathsf{tag},c) \xrightarrow{a}
  (\ell',\mathsf{tag}',0)$ whenever there is a transition
  $(\ell,\mathsf{tag}) \xrightarrow{g,a,\{x\}} (\ell',\mathsf{tag}')$
  in $\B$ with $\val{x \in I_c} \subseteq \val{g}$;
\item There is a transition $(\ell,\mathsf{tag},c) \xrightarrow{a}
  (\ell',\mathsf{tag}',c)$ whenever there is a transition
  $(\ell,\mathsf{tag}) \xrightarrow{g,a,\emptyset}
  (\ell',\mathsf{tag}')$ in $\B$ with $\val{x \in I_c} \subseteq
  \val{g}$
\end{itemize}
% In both cases, $\mathsf{tag}' = \mathsf{b}$ whenever $\mathsf{tag} =
% \mathsf{b}$ or $\ell' \in F$, and $\mathsf{tag}' = \mathsf{w}$
% otherwise (same tag as in $\B$...).
A state $(\ell,\mathsf{tag},c)$ is accepting if $\mathsf{tag} =
\mathsf{b}$, and we assume a B\"uchi condition. We parameterize $\B_f$
with the set of initial states $Q_0 \subseteq Q$, and we then write
$\B_f^{Q_0}$.

% It is easy to see that the language of $\A_f$ is
% suffix-closed.\pat{pas vraiment ca, je veux dire que une fois qu'on
%   est acceptant, on accepte}

This abstraction is a region abstraction for tails of Zeno behaviours
in the following sense:

\begin{lemmastyleS}
  \label{lemma2}
  Let $q=(\ell,\mathsf{tag},c)$ be a state of $\B_f$, and $u$ be an
  infinite (untimed) word. There is an equivalence between the two
  following properties:
  \begin{enumerate}
  \item $u$ can be read along some path $\varpi$ from $q$ in $\B_f$;
  \item for every $v \in \{c\} \cup I_c$, for every increasing
    timestamps sequence $\tau$
    % \pat{increasing!}\thomas{to be sure, increasing means strictly
    %   increasing?}
    which is convergent, and such that $v+\mathit{duration}(\tau) \in
    \overline{I_c} \cap (v;v+1)$, the timed word $w = (u,\tau)$ can be
    read along some path $\pi$ in $\B$.
%     there exists a Zeno timed word $w$ which can be read from
%     $(\ell,v)$ (along some path $\pi$) in $\A$, s.t.
%     \begin{itemize}
%     \item if $v$ has positive integral part (which is always the case
%       if $c<M$), then $v$ and $v+\mathit{duration}(w)$ lie in the same
%       unit interval;
%     \item if $c=M$ and $v$ is an integer, then
%       $v+\mathit{duration}(w)$ is in $(v;v+1)$.
%     \end{itemize}
%     It is furthermore the case that for every infinite timed word $w'$
%     such that $\mathsf{untime}(w) = \mathsf{untime}(w')$ and
%     $\mathit{duration}(w') \le \mathit{duration}(w)$, $w'$ can be read
%     from $(\ell,v)$ in $\A$ along the same path $\pi$.
%     % $v$ and $v+\mathit{duration}(w)$ lie in the same unit interval
%     % (in case $v$ has positive fractional part), and if $c=M$ and $v$
%     % is an integer, then $v+\mathit{duration}(w)$ is in $(v;v+1)$.
  \end{enumerate}
  In this equivalence we can furthermore assume the sequence of
  locations and tags encountered along $\varpi$ coincide with the
  sequence of locations and tags encountered along $\pi$.
  % In particular the sequence of tags along $\varpi$ coincide with
  % the sequence of tags along $\pi$.  only depends on the initial
  % value $\mathsf{tag}$ and on the
  % encountered locations.
\end{lemmastyleS}

\begin{proof}
  We first prove the implication $2 \Rightarrow 1$. Assume $w$ is a
  Zeno timed word read along the path $\pi = (\ell,\mathsf{tag})
  \xrightarrow{g_0,a_0,Y_0} (\ell_1,\mathsf{tag}_1)
  \xrightarrow{g_1,a_1,Y_1} \dots$ from configuration
  $(\ell,\mathsf{tag},v)$. The corresponding run has the form
  $(\ell,\mathsf{tag},v) \xrightarrow{\tau_0,a_0}
  (\ell_1,\mathsf{tag}_1,v_1) \xrightarrow{\tau_1,a_1} \dots$ and by
  assumption for every $j \ge 0$, $v+\sum_{i=0}^j \tau_i \in I_c$.
  % lies in the same unit interval as $v$.
  \begin{itemize}
  \item Assume the clock $x$ is never reset along $\pi$ (all $Y_k$'s
    are empty), then for every $k \ge 1$, $v_k = v+\sum_{0 \le i < k}
    \tau_i$ and thus $v_k \in I_c$, which implies $\val{x \in I_c}
    \subseteq \val{g_k}$. In that case, by construction of $\B_f$, we
    get that
    % if $\mathsf{tag} \in \{\mathsf{b},\mathsf{w}\}$, then
    there is path $(\ell,\mathsf{tag},c) \xrightarrow{a_0}
    (\ell_1,\mathsf{tag}_1,c) \xrightarrow{a_1} \dots$ in $\B_f$.
    % the finite automaton where each $\mathsf{tag}_k$ is uniquely
    % determined by $\mathsf{tag}$ and the encountered locations
    % $\ell_j$ (for $1 \le j \le k$).
  \item Assume the clock $x$ is reset along $\pi$, and that $Y_k =
    \{x\}$ is the first time $x$ is reset along $\pi$. Then the same
    argument as before applies to the prefix of $\pi$ up to
    $\ell_{k-1}$. Then we have that for every $j>k$, $v_j$ is either
    $0$ (in case $Y_{j-1} = \{x\}$) or lies in $(0;1)$, and that
    $\val{x \in(0;1)} \subseteq \val{g_j}$ in any case (time is
    supposed to be strictly monotonic). Thus we can build a path
    $\varpi$ in the finite automaton which reads the untiming of $w$.
  \end{itemize}
  
  We now prove the implication $1 \Rightarrow 2$. Assume that $u = a_0
  a_1 \dots$ is an infinite (untimed) word which is read along some
  path $\varpi = (\ell_0,\mathsf{tag}_0,c_0) \xrightarrow{a_0}
  (\ell_1,\mathsf{tag}_1,c_1) \xrightarrow{a_1} \dots$ in $\B_f$ with
  $(\ell_0,\mathsf{tag}_0,c_0) = (\ell,\mathsf{tag},c)$. Take now a
  value $v \in \{c\} \cup I_c$ for clock $x$ and take an increasing
  timestamps sequence $(t_i)_{i \ge 0}$ that is convergent and such
  that $v+ \sup_i t_i \in \overline{I_c} \cap (v;v+1)$.

  By construction of $\B_f$ there is a path $\pi =
  (\ell_0,\mathsf{tag}_0,c_0) \xrightarrow{g_0,a_0,Y_0}
  (\ell_1,\mathsf{tag}_1,c_1) \xrightarrow{g_1,a_1,Y_1} \dots$ that
  corresponds to $\varpi$. In particular, if $Y_i = \emptyset$, then
  $c_{i+1} = c_i$, and if $Y_i = \{x\}$, then $c_{i+1} = 0$. We
  distinguish between two cases:
  \begin{itemize}
  \item the clock $x$ is never reset (for all $i \ge 0$, $Y_i =
    \emptyset$), in which case for all $i$, $c_i = c$. We define for
    every $i \ge 1$, $v_i = v+t_{i-1}$. By assumption on $(t_i)_i$, it
    holds that $v_i \in I_c = I_{c_i}$. Thus, the following run reads
    the timed word $w=((t_i)_{i \ge 0},u)$:
    \[
    (\ell,\mathsf{tag},v) = (\ell_0,\mathsf{tag}_0,v_0)
    \xrightarrow{t_0,a_0} (\ell_1,\mathsf{tag}_1,v_1)
    \xrightarrow{t_1-t_0,a_1} \dots
    \]
  \item the clock $x$ is reset along $\pi$, and $i_0$ is the smallest
    index such that $Y_{i_0} = \{x\}$. We then define $I =
    \{i_0<i_1<\dots\}$ the set of index of sets $Y_i$'s where $Y_i =
    \{x\}$. We then define
    \[
    v_{j+1} = \left\{\begin{array}{ll}
      v+t_j & \text{if}\ j<i_0 \\
      0 & \text{if}\ j \in I \\
      t_j-t_{i_k} & \text{if}\ i_k<j<i_{k+1}
    \end{array}\right.
    \]
    Then, the following run reads the timed word $w=((t_i)_{i \ge
      0},u)$:
    \[
    (\ell,\mathsf{tag},v) = (\ell_0,\mathsf{tag}_0,v_0)
    \xrightarrow{t_0,a_0} (\ell_1,\mathsf{tag}_1,v_1)
    \xrightarrow{t_1-t_0} \dots
    \]
    because the values of the clock never exceeds $1$ after having
    reset the clock for the first time (and time is increasing, hence
    the constraint $x \in I_0$ is then always satisfied when firing a
    transition).
  \end{itemize}
  This concludes the proof of the second implication.
\end{proof}

\begin{lemmastyleS}
  \label{lemma4}
  Let $Q_0 \subseteq Q$ be a set of initial states for $\B_f$. Then we
  can decide in polynomial space whether $\B_f^{Q_0}$ is universal.
\end{lemmastyleS}

\begin{proof}
  $\B_f$ is a B\"uchi automaton, this result is thus standard,
  see~\cite{MH84}.
\end{proof}

% \begin{lemmastyleS}
%   First condition of the previous lemma can be decided (easily :-)
% \end{lemmastyleS}

% \begin{proof}
%   Some clever subset construction, that needs to be handled carefully
%   with all aborting executions with a black state. For instance, this
%   automaton is a bit problematic:
%   \begin{center}
%     \begin{tikzpicture}
%       \draw (0,0) node [draw,circle] (A) {};
%       \draw (0,-2) node [draw,circle,fill=black] (B) {};
%       \draw [-latex'] (A) -- (B);
%       \draw [-latex'] (A) .. controls +(-45:1cm) and +(45:1cm) .. (A);
%     \end{tikzpicture}
%   \end{center}
% \end{proof}

% \begin{lemmastyleS}
%   Assume that $\B_f^{Q_0}$ is not universal and that $u$ is an
%   infinite (untimed) word not accepted by $\B_f^{Q_0}$. Define
%   inductively $Q_i$ as the set of states reachable from $Q_0$ after
%   having read the length-$i$ prefix of $u$. Then $\B_f^{Q_i}$ is not
%   universal, and the suffix of $u$ starting from letter $i+1$ is not
%   accepted by $\B_f^{Q_i}$.
% \end{lemmastyleS}

% \begin{proof}
%   \fbox{should be easy to prove this}
% \end{proof}

\paragraph{\textbf{Handling the prefix.}}

We use an abstraction which is now standard in the context of
single-clock timed automata,
see~\cite{ADM04,ADMN04,OW04,LW05,OW05,LW07,OW07}
% We refine this standard abstraction for our context.
This is a symbolic transition system associated with $\B$, which is
denoted $\mathsf{Abst}_{\B}$ and defined as follows. We let $\Gamma$
be the finite set $2^{L \times \{\mathsf{b},\mathsf{w}\} \times
  \{0,1,\dots,M\}}$. The states of $\mathsf{Abst}_{\B}$ are tuples
$(\gamma,h,\gamma')$ where $\gamma,\gamma' \in \Gamma$ and $h \in
\Gamma^*$ is a finite word over alphabet $\Gamma$.  Informally an
abstract state $(\gamma,h,\gamma')$ represents a set of configurations
of $\B$, where $\gamma$ are those states where the value of $x$ is an
integer, $\gamma'$ are those states where the value of $x$ is larger
than $M$, and $h$ encodes the order on the fractional part of the
other states.  More precisely, an abstract state $(\gamma,h=\gamma_1
\dots \gamma_m,\gamma')$ represents a set of states of $\B$
$S=\{(\ell_j,\mathsf{tag}_j,v_j) \mid j \in J\}$ such that:
\begin{itemize}
\item $\gamma = \{(\ell_j,\mathsf{tag}_j,v_j) \mid v_j \in
  \{0,1,\dots,M\}\}$
\item $\bigcup_{l=1}^m \gamma_l = \{(\ell_j,\mathsf{tag}_j,v_j) \mid j
  \in J,\ v_j \in I_{c_j},\ c_j<M\}$
\item $\gamma' = \{(\ell_j,\mathsf{tag}_j,M) \mid v_j>M\}$
\item if $(\ell_i,\mathsf{tag}_i,v_i) \in \gamma_i$ and
  $(\ell_j,\mathsf{tag}_j,v_j) \in \gamma_j$, $i \le j$ iff
  $\mathsf{frac}(v_i) \le \mathsf{frac}(v_j)$
\end{itemize}
We then write $\mathit{abstr}(S) = (\gamma,h,\gamma')$, and $S \in
\mathit{concr}((\gamma,h,\gamma'))$.

We omit the definition of the abstract transitions in
$\mathsf{Abst}_{\B}$, which is rather tedious and can be found for
instance in~\cite{OW05}.

% There is an abstract transition $(\gamma_1,h_1,\gamma'_1)
% \xrightarrow{a} (\gamma_2,h_2,\gamma'_2)$ iff there exist \fbox{...}

This abstraction `computes' the set of executions that can read
finite timed words in $\B$ in the following sense.
\begin{lemmastyleS}
  \label{lemma-abstract}
  Let $\mathit{Conf} = (\gamma,h,\gamma')$ be an abstract
  configuration. Then $\mathit{Conf}$ is reachable in
  $\mathsf{Abst}_{\B}$ iff there exists a finite timed word $w$ such
  that $\mathit{Conf} = \mathit{abstr}(S_w)$, where $S_w$ is the set
  of configurations that are reached after reading $w$ in $\B$.
  % Let $w$ be a finite timed word, and let $R$ be the set of
  % executions that read $w$ in $\B$. We denote by $\mathit{Conf} =
  % \{(\ell_j,\mathsf{tag}_j,v_j) \mid j \in J\}$ the set of
  % configurations at the end of an execution in $R$.
  %   % , with the additional information $\mathsf{tag}_j =
  %   % \mathsf{b}$
  %   % iff a location in $F$ has been visited.
  % Then $\mathit{abstr}(\mathit{Conf})$ is reachable in the abstract
  % transition system iff there is a finite timed word $w$ such that
  % \fbox{blabla}
\end{lemmastyleS}

\begin{proof}
  This is proven in close terms for instance in~\cite{OW05} (where the
  point-of-view of alternating timed automata is taken).
\end{proof}

\paragraph{\textbf{Gluing everything.}}
We define $\mathcal{Z}$ the set of all sets of states $Q_0 \subseteq
Q$ such that $\A_f^{Q_0}$ is universal. This set can be computed
thanks to Lemma~\ref{lemma4}. If $(\gamma,h,\gamma')$ is an abstract
state of the above symbolic transition system, we write
$\mathit{set}((\gamma,h,\gamma'))$ for the set $\gamma \cup \gamma'
\cup \bigcup_{l=1}^m \gamma_l$, assuming $h = \gamma_1 \dots
\gamma_m$.

\begin{propositionstyleS}
  \label{prop:correctness}
  There is a Zeno timed word not accepted by $\B$ iff in
  $\mathsf{Abst}_{\B}$, it is possible to reach a configuration
  $\mathit{Conf}$ such that $\mathit{set}(\mathit{Conf}) \not\in
  \mathcal{Z}$.
  % when we take the set of pairs $(\ell,c)$ involved in $\gamma$ (we
  % forget about the order), the first condition of Lemma~\ref{lemma2}
  % is satisfied.
\end{propositionstyleS}

% \fbox{en fait, on pourrait pr\'ecalculer les \'etats desquels on
%   peut avoir des mots Zeno rejet\'es}

\begin{proof}
  Assume that $\B$ is not universal, and consider a Zeno timed word
  $w$ which is not accepted by $\B$. By assumption it can be read in
  $\B$. Also there exists some integer $n_0$ such that any run in $\B$
  which reads $w$ stabilizes after $n_0$ steps
  (Lemma~\ref{lemma-stabilize}). Assume $w_1$ is the $n_0$-th prefix
  of $w$ and that $w = w_1 \cdot w_2$. Let
  $\mathit{Conf}=\mathit{abstr}(S_{w_1})$. We will prove that writing
  $Q_0$ for $\mathit{set}(\mathit{Conf})$, it is the case that
  $\B_f^{Q_0}$ does not accept $\mathsf{untime}(w_2)$. Towards a
  contradiction assume it is not the case. Then let $\pi_2$ be a path
  that accepts $\mathsf{untime}(w_2)$ in $\B_f^{Q_0}$. It starts from
  some $q_0 \in Q_0$. We have that $q_0 \in
  \mathit{set}(\mathit{abstr}(S_{w_1}))$, and thus there is some
  configuration $(\ell,\mathsf{tag},v) \in S_{w_1}$ which corresponds
  to $q_0$, and there is some run in $\B$ which reads $w_1$ and
  reaches $(\ell,\mathsf{tag},v)$. We can then apply
  Lemma~\ref{lemma2} and lift $\pi_2$ into a run $\varpi_2$ that
  accepts $w_2$ (it visits the same locations and the same tags as
  $\pi_2$, which is accepting). This is the expected contradiction.

  \medskip Assume now that in $\mathsf{Abst}_{\B}$, it is possible to
  reach a configuration $\mathit{Conf}$ such that $Q_0 =
  \mathit{set}(\mathit{Conf}) \not\in \mathcal{Z}$. Let $w_1$ be a
  finite timed word such that $\mathit{Conf} =
  \mathit{abstr}(S_{w_1})$ (Lemma~\ref{lemma-abstract}). As
  $\B_f^{Q_0}$ is not universal, there is an infinite (untimed) word
  $u_2$ which is not accepted by $\B_f^{Q_0}$. Let $\alpha$ be the
  largest fractional part involved in $S_{w_1}$, and $T$ be the
  largest timestamp of $w_1$. We define the (Zeno) timed word $w = w_1
  \cdot w_2$ where the untiming of $w_2$ is $u_2$ and the timestamps
  of $w_2$ are $T+\epsilon \cdot \left(\frac{1}{2}\right)$,
  $T+\epsilon \cdot \left(\frac{1}{2}+\frac{1}{4}\right)$, ...,
  $T+\epsilon \cdot
  \left(\frac{1}{2}+\frac{1}{4}+\dots+\frac{1}{2^k}\right)$, ... where
  $\epsilon$ is chosen such that $0<\epsilon<1-\alpha$. It is then
  obvious that any infinite run in $\B$ that reads $w$ stabilizes
  after $n_0 = |u_1|$ steps (because any infinite run which reads $w$
  starts with a prefix reading $w_1$, hence after $|u_1|$ steps, it is
  in a configuration of $S_{w_1}$).
  
  % The timed word $w$ is read along several runs in $\B$.
  Towards a contradiction, consider a run $\varrho$ in $\B$ which
  accepts $w$.
  % , and take $n_0$ an integer such that $\varrho$ stabilizes after
  % $n_0$ steps. W.l.o.g. we can assume that $n_0 \ge |u_1|$.
  The suffix $w_2$ of $w$ satisfies the second condition of
  Lemma~\ref{lemma2}. Applying the equivalence stated in this lemma,
  the first condition is also satisfied, and this corresponds to 
  the untimed word $u_2$.
  % an untimed word $u'$ which is a suffix of $u_2$ (because $n_0 \ge
  % |u_1|$).
  The corresponding path in $\B_f$, which starts from some $q \in
  Q_0$, is accepting because $\varrho$ is accepting, and tags and
  locations are preserved. Hence it is the case that $\B_f^{\{q\}}$
  accepts $u_2$. It is then the case that $\B_f^{Q_0}$ accepts $u_2$,
  which contradicts the assumption.
  % Assume now that in the abstract transition system we can reach an
  % abstract configuration $(\gamma,h,\gamma')$ such that if $Q_0 =
  % ...$ then $\A_f^{Q_0}$ is not universal. There is an infinite
  % (untimed) word $u$ such that all executions reading $u$ from any
  % of the states in $Q_0$ are tagged with $\mathsf{w}$. Write now
  % $\alpha$ for a finite timed word that is read and leads to a
  % concretization of $(\gamma,h,\gamma')$. Then from that
  % concretization we are able to build a timed word $\beta$ such that
  % $u$ is the untiming of $\beta$. We claim that $\alpha\beta$ is a
  % timed word which is not accepted by $\A$.
\end{proof}

The set $\mathcal{Z}$ is upward-closed, hence we can decide the
reachability problem w.r.t. $\mathcal{Z}$ in the well-structured
transition system $\mathsf{Abst}_{\B}$ (see for
instance~\cite{OW05}). Hence, as a consequence of
Proposition~\ref{prop:correctness} above, we get
Theorem~\ref{theo:univ-zeno-decidable}.
% the following result.

% Applying Theorem~? of~\cite{}, we get the following theorem:
% \begin{theorem}
%   The universality problem for Zeno timed words with positive
%   frequencies in one-clock timed automata is decidable.
% \end{theorem}

\begin{remarkstyleS}
  The construction made in the proof of the above theorem can be
  adapted to prove that the universality problem for Zeno timed words
  in one-clock timed automata with a standard B\"uchi acceptance
  condition is decidable.
\end{remarkstyleS}

% \patlong{en fait, je pense qu'une preuve un peu similaire pourrait
%   s'adapter au cas B\"uchi zeno}

\end{proof}
%%% Local Variables: 
%%% mode: latex
%%% TeX-master: "ICALP11"
%%% End: 

\putbib[ICALP11-2]
\end{bibunit}


\begin{thebibliography}{m}

\bibitem{ADMW-fossacs09}
R.~Alur, A.~Degorre, O.~Maler, and G.~Weiss.
\newblock On omega-languages defined by mean-payoff conditions.
\newblock In {\em Proc. 12th Intl Conf. on Foundations of Software Science and
  Computation Structures (FoSSaCS'09)}, LNCS 5504, p.~333--347. Springer, 2009.

\bibitem{AD-tcs94}
R.~Alur and D.~L. Dill.
\newblock A theory of timed automata.
\newblock {\em Theoretical Computer Science}, 126(2):183--235, 1994.

\bibitem{ALP-hscc01}
R.~Alur, S.~L. Torre, and G.~J. Pappas.
\newblock Optimal paths in weighted timed automata.
\newblock In {\em Proc. 4th Intl Workshop on Hybrid Systems: Computation and
  Control (HSCC'01)}, LNCS 2034, p.~49--62. Springer, 2001.

\bibitem{AD-concur09}
E.~Asarin and A.~Degorre.
\newblock Volume and entropy of regular timed languages: Discretization
  approach.
\newblock In {\em Proc. 20th Intl Conf. on Concurrency Theory (CONCUR'09)},
  LNCS 5710, p.~69--83. Springer, 2009.

\bibitem{BBBBG-lics08}
C.~Baier, N.~Bertrand, P.~Bouyer, {\relax Th}.~Brihaye, and M.~Gr{\"o}{\ss}er.
\newblock Almost-sure model checking of infinite paths in one-clock timed
  automata.
\newblock In {\em Proc. 23rd Annual IEEE Symp. on Logic in Computer Science
  (LICS'08)}, p.~217--226. IEEE Comp. Soc. Press, 2008.

\bibitem{BFHLPRV-hscc01}
G.~Behrmann, A.~Fehnker, {\relax Th}.~Hune, K.~G. Larsen, P.~Pettersson,
  J.~Romijn, and F.~W. Vaandrager.
\newblock Minimum-cost reachability for priced timed automata.
\newblock In {\em Proc. 4th Intl Workshop on Hybrid Systems: Computation and
  Control (HSCC'01)}, LNCS 2034, p.~147--161. Springer, 2001.

\bibitem{BFMM-qapl10}
A.~Bianco, M.~Faella, F.~Mogavero, and A.~Murano.
\newblock Quantitative fairness games.
\newblock In {\em Proc. 8th Workshop on Quantitative Aspects of Programming
  Languages (QAPL'10)}, ENTCS~28, p.~48--63, 2010.

\bibitem{BBL-fmsd08}
P.~Bouyer, E.~Brinksma, and K.~G. Larsen.
\newblock Optimal infinite scheduling for multi-priced timed automata.
\newblock {\em Formal Methods in System Design}, 32(1):3--23, 2008.

\bibitem{CDEHR-concur10}
K.~Chatterjee, L.~Doyen, H.~Edelsbrunner, {\relax Th}.~A. Henzinger, and
  {\relax Ph}.~Rannou.
\newblock Mean-payoff automaton expressions.
\newblock In {\em Proc. 21th Intl Conf. on Concurrency Theory (CONCUR'10)},
  LNCS 6269, p.~269--283. Springer, 2010.

\bibitem{CDH-acmtocl10}
K.~Chatterjee, L.~Doyen, and {\relax Th}.~A. Henzinger.
\newblock Quantitative languages.
\newblock {\em ACM Transactions on Computational Logic}, 11(4), 2010.

\bibitem{KNSS-tcs02}
M.~Z. Kwiatkowska, G.~Norman, R.~Segala, and J.~Sproston.
\newblock Automatic verification of real-time systems with discrete probability
  distributions.
\newblock {\em Theoretical Computer Science}, 282:101--150, 2002.

\bibitem{LMS-concur04}
F.~Laroussinie, N.~Markey, and {\relax Ph}.~Schnoebelen.
\newblock Model checking timed automata with one or two clocks.
\newblock In {\em Proc. 15th Intl Conf. on Concurrency Theory (CONCUR'04)},
  LNCS 3170, p.~387--401. Springer, 2004.

\bibitem{OW05}
J.~Ouaknine and J.~Worrell.
\newblock On the decidability of {M}etric {T}emporal {L}ogic.
\newblock In {\em Proc. 20th Annual Symp. on Logic in Computer Science
  (LICS'05)}, p.~188--197. IEEE Comp. Soc. Press, 2005.

\bibitem{TBG-fsttcs09}
M.~Tracol, C.~Baier, and M.~Gr{\"o}{\ss}er.
\newblock Recurrence and transience for probabilistic automata.
\newblock In {\em Proc. 29th IARCS Annual Conf. on Foundations of Software
  Technology and Theoretical Computer Science (FSTTCS'09)}, LIPIcs~4,
  p.~395--406. Schloss Dagstuhl - Leibniz-Zentrum fuer Informatik, 2009.

\end{thebibliography}


\begin{thebibliography}{h}

\bibitem{ADM04}
P.~A. Abdulla, J.~Deneux, and P.~Mahata.
\newblock Multi-clock timed networks.
\newblock In {\em Proc. 19th Annual Symp. on Logic in Computer Science
  (LICS'04)}, p.~345--354. IEEE Comp. Soc. Press, 2004.

\bibitem{ADMN04}
P.~A. Abdulla, J.~Deneux, P.~Mahata, and A.~Nylen.
\newblock Forward reachability analysis of timed {P}etri nets.
\newblock In {\em Proc. Joint Conf. on Formal Modelling and Analysis of Timed
  Systems and Formal Techniques in Real-Time and Fault Tolerant System
  (FORMATS+FTRTFT'04)}, LNCS 3253, p.~343--362. Springer, 2004.

\bibitem{BBL-fmsd08}
P.~Bouyer, E.~Brinksma, and K.~G. Larsen.
\newblock Optimal infinite scheduling for multi-priced timed automata.
\newblock {\em Formal Methods in System Design}, 32(1):3--23, 2008.

\bibitem{CDEHR-concur10}
K.~Chatterjee, L.~Doyen, H.~Edelsbrunner, {\relax Th}.~A. Henzinger, and
  {\relax Ph}.~Rannou.
\newblock Mean-payoff automaton expressions.
\newblock In {\em Proc. 21th Intl Conf. on Concurrency Theory (CONCUR'10)},
  LNCS 6269, p.~269--283. Springer, 2010.

\bibitem{LW05}
S.~Lasota and I.~Walukiewicz.
\newblock Alternating timed automata.
\newblock In {\em Proc. 8th Intl Conf. on Foundations of Software Science and
  Computation Structures (FoSSaCS'05)}, LNCS 3441, p.~250--265. Springer, 2005.

\bibitem{LW07}
S.~Lasota and I.~Walukiewicz.
\newblock Alternating timed automata.
\newblock {\em ACM Transactions on Computational Logic}, 9(2:10), 2008.

\bibitem{MH84}
S.~Miyano and T.~Hayashi.
\newblock Alternating finite automata on omega-words.
\newblock {\em Theorerical Computer Science}, 32:321--330, 1984.

\bibitem{OW04}
J.~Ouaknine and J.~Worrell.
\newblock On the language inclusion problem for timed automata: Closing a
  decidability gap.
\newblock In {\em Proc. 19th Annual Symp. on Logic in Computer Science
  (LICS'04)}, p.~54--63. IEEE Comp. Soc. Press, 2004.

\bibitem{OW05}
J.~Ouaknine and J.~Worrell.
\newblock On the decidability of {M}etric {T}emporal {L}ogic.
\newblock In {\em Proc. 20th Annual Symp. on Logic in Computer Science
  (LICS'05)}, p.~188--197. IEEE Comp. Soc. Press, 2005.

\bibitem{OW07}
J.~Ouaknine and J.~Worrell.
\newblock On the decidability and complexity of metric temporal logic over
  finite words.
\newblock {\em Logical Methods in Computer Science}, 3(1:8), 2007.

\end{thebibliography}
\end{document}